%% file: Paper.tex
\documentclass[cmp,numbook]{svjour} 
\usepackage{graphics}
\usepackage{times}
\usepackage{amsmath,amssymb}
\usepackage[dvips]{graphicx, color}
\usepackage{verbatim}
\usepackage{bbm}
\usepackage{longtable}
\usepackage{amscd}
\usepackage{enumerate}
\usepackage[all]{xy}
\usepackage{stmaryrd}
\usepackage{bm}

\DeclareMathOperator{\sgn}{sgn}

\DeclareMathOperator{\sech}{sech}

\DeclareMathOperator{\arcsinh}{arcsinh}

\newcommand{\dd}{\mathrm{d}}

\newcommand{\R}{\mathbb{R}}
 
\newcommand{\Z}{\mathbb{Z}}

\newcommand{\p}{\partial}

\spnewtheorem{convention}{Convention}[section]{\it}{\rm}

\newcommand{\smp}{\boldsymbol{\mathfrak A}}

\newcommand{\td}[1]{\tfrac{\dd}{\dd \tau} #1}

\newcommand{\spp}[3]{\ifnum #1 = 0 \else \ifnum #1 = 1 {\smp}_1 \else ({\smp}_1)^{#1} \fi \fi \ifnum #2 = 0 \else \ifnum #2 = 1 {\smp}_2 \else ({\smp}_2)^{#2} \fi \fi \ifnum #3 = 0 \else \ifnum #3 = 1 {\smp}_3 \else ({\smp}_3)^{#3} \fi \fi}
\newcommand{\sppcycl}[3]{\ifnum #1 = 0 \else \ifnum #1 = 1 {\smp}_{\xa} \else ({\smp}_{\xa})^{#1} \fi \fi \ifnum #2 = 0 \else \ifnum #2 = 1 {\smp}_{\xb} \else ({\smp}_{\xb})^{#2} \fi \fi \ifnum #3 = 0 \else \ifnum #3 = 1 {\smp}_{\xc} \else (\smp_{\xc})^{#3} \fi \fi}
\newenvironment{spexplicitparamaux}{\renewcommand{\spp}[3]{\ifnum ##1 = 0 \else n_1 \ifnum ##1 = 1 {\smp}_1 \else ({\smp}_1)^{##1} \fi \fi \ifnum ##2 = 0 \else n_2 \ifnum ##2 = 1 {\smp}_2 \else ({\smp}_2)^{##2} \fi \fi \ifnum ##3 = 0 \else n_3 \ifnum ##3 = 1 {\smp}_3 \else ({\smp}_3)^{##3} \fi \fi}\renewcommand{\sppcycl}[3]{\ifnum ##1 = 0 \else n_{\xa} \ifnum ##1 = 1 {\smp}_{\xa} \else ({\smp}_{\xa})^{##1} \fi \fi \ifnum ##2 = 0 \else n_{\xb} \ifnum ##2 = 1 {\smp}_{\xb} \else ({\smp}_{\xb})^{##2} \fi \fi \ifnum ##3 = 0 \else n_{\xc} \ifnum ##3 = 1 {\smp}_{\xc} \else (\smp_{\xc})^{##3} \fi \fi}}{}

\newcommand{\xa}{\mathbf{i}}
\newcommand{\xb}{\mathbf{j}}
\newcommand{\xc}{\mathbf{k}}

\newcommand{\dfield}[1]{\ifnum #1=1 P \else \ifnum #1=2 Q \else \ifnum #1=3 Y \else {\huge ERROR} \fi\fi\fi}
\newcommand{\dfrak}[1]{\ifnum #1=1 \boldsymbol{\mathfrak p} \else \ifnum #1=2 \boldsymbol{\mathfrak q} \else \ifnum #1=3 \slaa{\boldsymbol{\mathfrak y}} \else {\huge ERROR} \fi\fi\fi}
\newcommand{\dnofrak}[1]{\ifnum #1=1 p \else \ifnum #1=2 q \else \ifnum #1=3 \slaa{y} \else {\huge ERROR} \fi\fi\fi}

\makeatletter

\newcommand{\Rmnum}[1]{\expandafter\@slowromancap\romannumeral #1@}
\makeatother

\newcommand{\slaz}[1]{#1\hskip-4.5pt/\hskip1pt}
\newcommand{\slaa}[1]{#1\hskip-5.5pt/\hskip1pt}

\newcommand{\clball}[3]{B_{#1}[#2,#3]}

\newcommand{\gold}{\rho}

\begin{document}

\rule{0pt}{0pt}
\vskip 40mm 
{\bf \Large
\noindent The BKL Conjectures for\\
Spatially Homogeneous Spacetimes
}
\vskip 4mm
\noindent {\bf Michael Reiterer, Eugene Trubowitz}
\vskip 2mm
\noindent {Department of Mathematics, ETH Zurich, Switzerland}
\vskip 4mm
\noindent {\bf Abstract:} 
We rigorously construct and control a generic class of spatially homogeneous (Bianchi VIII and Bianchi IX) vacuum spacetimes that exhibit the oscillatory BKL phenomenology.
We investigate the causal structure of these spacetimes and show that there is a ``particle horizon''.


\vskip 2mm

\setcounter{MaxMatrixCols}{30}
\newcommand{\ua}{\alpha}
\newcommand{\ub}{\beta}

\input{SectionIntro.tex}
\input{SectionEquations.tex}

\input{SectionRational.tex}
\input{SectionLimit.tex}
\input{SectionSummary.tex}

\input{SectionLightCone.tex}

\appendix
\input{SectionAppendix1le2.tex}
\input{SectionAppendixQL.tex}

%

%
%

\end{document}

%% file: SectionIntro.tex
\section{Introduction}
The goal of this paper is to rigorously construct and explicitly control a generic class of solutions
$\Phi = \ua \oplus \ub:\, [0,\infty) \to \R^3\oplus \R^3$, with independent variable $\tau \in [0,\infty)$ and with\footnote{If $\tau \mapsto \Phi(\tau)$ is a solution to 
\eqref{eq:kshkjhkf}, so is $\tau \mapsto - \Phi(-\tau)$. The condition $(\alpha_1+\alpha_2+\alpha_3)|_{\tau = 0} < 0$ breaks this symmetry.
Solutions to \eqref{eq:kshkjhkf} with $(\alpha_1+\alpha_2+\alpha_3)|_{\tau = 0} < 0$ do not break down in finite positive time, that is, they extend to $[0,\infty)$. A proof of this fact is given later in this introduction.}
$(\alpha_1+\alpha_2+\alpha_3)|_{\tau = 0} < 0$, to the autonomous system of six ordinary differential equations
\begin{subequations}\label{eq:kshkjhkf}
\begin{align}
\label{eq:kshkjhk1}
\begin{split}
0 & =
-  \td{\ua_{\xa}}
- (\ub_{\xa})^2
+ (\ub_{\xb})^2
+ (\ub_{\xc})^2 - 2\ub_{\xb}\ub_{\xc}
\end{split}
\\ 
\label{eq:kshkjhk2}
\begin{split}
0 & =
- \td{\ub_{\xa}}  + \ub_{\xa}\ua_{\xa}
\end{split}
\end{align}
for all $(\xa,\xb,\xc)\in \mathcal{C}
\stackrel{\text{def}}{=} \{(1,2,3), (2,3,1), (3,1,2)\}$, subject to the quadratic constraint\footnote{As a quadratic form on $\R^3\oplus \R^3$, 
the right hand side of \eqref{eq:kshkjhk3} has signature
$(+,+,-,-,-,-)$.}
\begin{equation} 
\label{eq:kshkjhk3}
0  = \ua_2\ua_3 + \ua_3\ua_1 + \ua_1 \ua_2
- (\ub_1)^2
- (\ub_2)^2
- (\ub_3)^2 + 2 \ub_2\ub_3 
+ 2 \ub_3\ub_1
+ 2 \ub_1\ub_2
\end{equation}
\end{subequations}
Here, $\ua=(\ua_1,\ua_2,\ua_3)$, $\ub = (\ub_1,\ub_2,\ub_3)$. The system \eqref{eq:kshkjhkf} are the vacuum Einstein equations for spatially homogeneous (Bianchi) spacetimes, see Proposition \ref{prop:khdkhkhkhss}.
\vskip 1mm
The pioneering calculations and heuristic picture of Belinskii, Khalatnikov, Lifshitz\footnote{The work of Belinskii, Khalatnikov, Lifshitz concerns general (inhomogeneous) spacetimes, but relies on intuition about the homogeneous case.} \cite{BKL1} and Misner \cite{Mis} suggest that a generic class of solutions to \eqref{eq:kshkjhkf} are oscillatory as $\tau \to + \infty$ and that
the dynamics of one degree of freedom is closely related to the discrete dynamics of the Gauss map
$G(x) = \tfrac{1}{x}-\lfloor \tfrac{1}{x}\rfloor$, a non-invertible map from $(0,1)\setminus \mathbb{Q}$ to itself. Every element of $(0,1)\setminus \mathbb{Q}$ admits a unique infinite continued fraction expansion
\begin{equation}
\langle k_1,k_2,k_3,\ldots \rangle\;=\;\frac{1}{k_1 + \frac{1}{k_2 + \frac{1}{k_3+\ldots}}}
\end{equation}
where $(k_n)_{n \geq 1}$ are strictly positive integers. The Gauss map is the left-shift,
\begin{equation}\label{eq:kdkdhdfkjd}
G\big(\langle k_1,k_2,k_3,\ldots \rangle\big) = \langle k_2,k_3,k_4,\ldots \rangle
\end{equation}
\vskip 1mm
Rigorous results about spatially homogeneous spacetimes have been obtained by Rendall \cite{Ren} and Ringstr\"om \cite{Ri1}, \cite{Ri2}. See also Heinzle and Uggla \cite{HU2}. We refer to the very readable paper \cite{HU1} for a detailed discussion.
\vskip 1mm
The first rigorous proofs that there exist spatially homogeneous vacuum spacetimes whose asymptotic behavior is related, in a precise sense, to iterates of the Gauss map, have been obtained recently by B\'eguin \cite{Be} and by Liebscher, H\"arterich, Webster and Georgi \cite{LHWG}. These theorems apply to a dense subset of $(0,1)\setminus \mathbb{Q}$. A basic restriction of both these works is that the sequence $(k_n)_{n \geq 1}$ has to be bounded, a condition fulfilled only by a Lebesgue measure zero subset of $(0,1)\setminus \mathbb{Q}$. The results of the present paper apply to any sequence $(k_n)_{n \geq 1}$ that grows at most polynomially. The corresponding subset of $(0,1)\setminus \mathbb{Q}$ has full Lebesgue measure one.
\vskip 1mm
We point out some properties of the system \eqref{eq:kshkjhk1}, \eqref{eq:kshkjhk2}, not assuming \eqref{eq:kshkjhk3}:
\begin{enumerate}[(i)]
\item\label{item:scr1} The right hand side of \eqref{eq:kshkjhk3} is a conserved quantity.
\item\label{item:scr2} If $\tau \mapsto \Phi(\tau)$ is a solution, so is $\tau \mapsto p\,\Phi(p\tau + q)$, for all $p,q\in \R$.
\item\label{item:scr3} The signatures $(\sgn\beta_1,\sgn\beta_2,\sgn\beta_3)$ are constant.
\item\label{item:scr4} $\tfrac{\dd}{\dd \tau}|\beta_1\beta_2\beta_3|^2 = 2(\alpha_1+\alpha_2+\alpha_3)|\beta_1\beta_2\beta_3|^2$.
\item[(v)] We have\footnote{$(\beta_1)^2+(\beta_2)^2+(\beta_3)^2 - 2\beta_2\beta_3-2\beta_3\beta_1-2\beta_1\beta_2 + 3|\beta_1\beta_2\beta_3|^{2/3} \geq 0$ holds for all $\beta_1,\beta_2,\beta_3\in \R$, see \cite{HU1}. The only nontrivial cases are $\beta_1,\beta_2,\beta_3 > 0$ or $\beta_1,\beta_2,\beta_3 < 0$. In these cases, the inequality is a direct consequence of the polynomial identity\begin{multline*}
x^6+y^6+z^6-2y^3z^3-2z^3x^3-2x^3y^3+3x^2y^2z^2  =\\
\tfrac{1}{2}\,\big(x^2+y^2+z^2+yz+zx+xy\big) \Big((y-z)^2(y+z-x)^2+(z-x)^2(z+x-y)^2+(x-y)^2(x+y-z)^2\Big)
\end{multline*}
}
$\tfrac{\dd}{\dd \tau}(\alpha_1+\alpha_2+\alpha_3) \geq
- 3|\beta_1\beta_2\beta_3|^{2/3}$.
\end{enumerate}
If in addition we assume \eqref{eq:kshkjhk3}, then:
\begin{enumerate}[(i)]
\item[(vi)]\label{item:scr6}  $\tfrac{\dd}{\dd \tau}(\alpha_1+\alpha_2+\alpha_3)
= \alpha_2\alpha_3+\alpha_3\alpha_1+\alpha_1\alpha_2  \leq \tfrac{1}{3}(\alpha_1+\alpha_2+\alpha_3)^2$. 
\end{enumerate}

Let $\Phi = \alpha\oplus \beta$ be any solution to \eqref{eq:kshkjhkf}, that is \eqref{eq:kshkjhk1}, \eqref{eq:kshkjhk2}, \eqref{eq:kshkjhk3}, on the half-open interval $[0,\tau_1)$ with $0 <  \tau_1 < \infty$.
Set $\slaa{\alpha} = \alpha_1+\alpha_2+\alpha_3$ and suppose $\slaa{\alpha}(0) < 0$. Then
\begin{equation}\label{eq:dkhkdhd}
\slaa{\alpha}(\tau) \leq - |\slaa{\alpha}(0)| /  \big(1+\tfrac{1}{3}|\slaa{\alpha}(0)| \tau\big) < 0
\qquad \text{for all $\tau \in [0,\tau_1)$}
\end{equation}
by (vi). Consequently, $|\beta_1\beta_2\beta_3|$ is bounded, by (iv), and $\slaa{\alpha}$ is bounded below, by (v), on $[0,\tau_1)$. The constraint \eqref{eq:kshkjhk3} implies that\footnote{Use $2(\alpha_2\alpha_3 + \alpha_3\alpha_1 + \alpha_1\alpha_2) 
= \slaa{\alpha}^2  - (\alpha_1)^2 - (\alpha_2)^2 - (\alpha_3)^2$.}
$(\alpha_1)^2 + (\alpha_2)^2+(\alpha_3)^2 \leq 6\,|\beta_1\beta_2\beta_3|^{2/3} + \slaa{\alpha}^2$ is bounded. Now \eqref{eq:kshkjhk2} implies that $(\beta_1)^2 + (\beta_2)^2 + (\beta_3)^2$ is bounded. Therefore, solutions to  \eqref{eq:kshkjhkf}  with $\slaa{\alpha}(0) < 0$ can be extended to $[0,\infty)$. The solutions considered in this paper belong to this general class. We are especially interested in their $\tau \to + \infty$ asymptotics.
\vskip 1mm
For every solution to \eqref{eq:kshkjhkf} with $\slaa{\alpha}(0) < 0$, as in the last paragraph, the
half-infinite interval $[0,\infty)$ actually corresponds to a \emph{finite} physical duration of the associated spatially homogeneous vacuum spacetime (given in Proposition \ref{prop:khdkhkhkhss}). In fact, an increasing affine parameter along the timelike geodesics orthogonal to the level sets of $\tau$ is given by  $\tau \mapsto \int_0^{\tau} \exp(\tfrac{1}{2}\int_0^s \slaa{\alpha}) \dd s$, with uniform upper bound $6|\slaa{\alpha}(0)|^{-1}$, by \eqref{eq:dkhkdhd}. 
\vskip 1mm
In this paper, we consider only solutions to \eqref{eq:kshkjhkf} for which $\beta_1,\beta_2,\beta_3 \neq 0$ (also called Bianchi VIII or IX models). We now give an informal description  of the solutions that we construct, the phenomenological picture of \cite{BKL1}. The structure of each of these solutions is described by three sequences of compact subintervals $(\mathcal{I}_j)_{j \geq 1}$, $(\mathcal{B}_j)_{j\geq 1}$, $(\mathcal{S}_j)_{j \geq 1}$ of $[0,\infty)$, for which:
\begin{enumerate}[({a}.1)]
\item The left endpoint of $\mathcal{I}_1$ is the origin, and the right endpoint of $\mathcal{I}_j$, henceforth denoted $\tau_j$, coincides with the left endpoint of $\mathcal{I}_{j+1}$, for all $j \geq 1$. Set $\tau_0 = 0$.
\item $\bigcup_{j \geq 1}\mathcal{I}_j = [0,\infty)$, that is,
$\lim_{j \to +\infty}\tau_j = +\infty$. 
\item $\mathcal{B}_j$ is contained in the interior of $\mathcal{I}_j$, and $0 < |\mathcal{B}_j| \ll |\mathcal{I}_j|$, for all $j \geq 1$.
\item $\mathcal{S}_j$ is the closed interval of all points between $\mathcal{B}_j$ and $\mathcal{B}_{j+1}$, for all
$j \geq 1$.
\end{enumerate}
Here is a picture:
\begin{center}
\input{interval.pstex_t}
\end{center}
Let $S_3$ be the set of all permutations $(\mathbf{a},\mathbf{b},\mathbf{c})$ of the triple $(1,2,3)$.
The solution is further described by a sequence $(\pi_j)_{j \geq 1}$ in $S_3$, with $\pi_j = (\mathbf{a}(j),\mathbf{b}(j),\mathbf{c}(j))$, so that:
\begin{enumerate}[({b}.1)]
\item On $\mathcal{I}_j$, the components $\beta_{\mathbf{b}(j)}$, $\beta_{\mathbf{c}(j)}$ are so small in absolute value that the local dynamics of $\Phi = \alpha\oplus \beta$ is essentially unaffected if $\beta_{\mathbf{b}(j)}$, $\beta_{\mathbf{c}(j)}$ are set equal to zero in the four equations \eqref{eq:kshkjhk1} and \eqref{eq:kshkjhk3}.
\item
 On $\mathcal{I}_j \setminus \mathcal{B}_j$, the component $\beta_{\mathbf{a}(j)}$ is so small in absolute value that the local dynamics of $\Phi = \alpha\oplus \beta$ is essentially unaffected if $\beta_{\mathbf{a}(j)}$ is set equal to zero in the four equations \eqref{eq:kshkjhk1} and \eqref{eq:kshkjhk3}.
The component $\beta_{\mathbf{a}(j)}$ is \emph{not small} on $\mathcal{B}_j$, but the mixed products  $\beta_{\mathbf{a}(j)}\beta_{\mathbf{b}(j)}$ and $\beta_{\mathbf{a}(j)}\beta_{\mathbf{c}(j)}$ are still small.
\item Items (b.1) and (b.2) imply that mixed products of components of $\beta$ are small on all of $[0,\infty)$, and that all three components of $\beta$ are small on $\bigcup_{j \geq 1} \mathcal{S}_j$.
\item $\mathbf{a}(j) \neq \mathbf{a}(j+1)$ for all $j \geq 1$. 
\item None of the properties listed so far distinguishes $\mathbf{b}(j)$ from $\mathbf{c}(j)$. By (b.4), this ambiguity can be consistently eliminated
by stipulating $\mathbf{b}(j) = \mathbf{a}(j+1)$.
\end{enumerate}
We can draw the following heuristic consequences from the eight heuristic properties above. \emph{Separately on each interval} $\mathcal{S}_j$, $j \geq 1$:
\begin{enumerate}[({c}.1)]
\item The components of $\alpha$ are essentially constant, 
by \eqref{eq:kshkjhk1}
and (b.3), and $\log|\beta_1|$, $\log|\beta_2|$, $\log|\beta_3|$ are essentially linear functions with slopes $\alpha_1$, $\alpha_2$, $\alpha_3$,
by \eqref{eq:kshkjhk2}.
\item The constraint \eqref{eq:kshkjhk3} essentially reduces to
$\alpha_2\alpha_3 + \alpha_3\alpha_1 + \alpha_1\alpha_2 = 0$. 
As before, we require $\slaa{\alpha} = \alpha_1 + \alpha_2 + \alpha_3 < 0$.
Furthermore, we make the generic assumption that all components of $\alpha$ are nonzero.
These conditions imply that two components of $\alpha$ are negative, one component of $\alpha$ is positive, and the sum of any two is negative.
\item The single positive component of $\alpha$ has to be
$\alpha_{\mathbf{b}(j)} = \alpha_{\mathbf{a}(j+1)}$. In fact, we know that $|\beta_{\mathbf{a}(j+1)}|$ is very small
on $\mathcal{S}_j$ but is not small on $\mathcal{B}_{j+1}$. Therefore, the slope of $\log|\beta_{\mathbf{a}(j+1)}|$, which is $\alpha_{\mathbf{a}(j+1)}$ by (c.1), has to be positive on $\mathcal{S}_j$.
\item By the last three items and (b.4), there is at most one point in $\mathcal{S}_j$ where $|\beta_{\mathbf{a}(j)}| = |\beta_{\mathbf{a}(j+1)}|$. By (b.1), (b.2), there is such a point, because $|\beta_{\mathbf{a}(j)}|$ is going from not small to small on $\mathcal{S}_j$, and $|\beta_{\mathbf{a}(j+1)}|$ is going from small to not small on $\mathcal{S}_j$. By convention, this point is $\tau_j$.
\end{enumerate}
\emph{Separately on each interval} $\mathcal{I}_j$, $j \geq 1$ (in particular on $\mathcal{B}_j\subset \mathcal{I}_j$):
\begin{enumerate}[({d}.1)]
\item $\alpha_{\mathbf{a}(j)} + \alpha_{\mathbf{b}(j)}$ and $\alpha_{\mathbf{a}(j)} + \alpha_{\mathbf{c}(j)}$
are essentially constant, by \eqref{eq:kshkjhk1}, (b.1),
and they are both negative, by (c.1), (c.2).
Also,  $\log|\beta_{\mathbf{a}(j)}\beta_{\mathbf{b}(j)}|$, $\log|\beta_{\mathbf{a}(j)}\beta_{\mathbf{c}(j)}|$ are essentially linear functions with slopes 
$\alpha_{\mathbf{a}(j)} + \alpha_{\mathbf{b}(j)}$ and $\alpha_{\mathbf{a}(j)} + \alpha_{\mathbf{c}(j)}$, by \eqref{eq:kshkjhk2}.
\item Essentially 
$(\alpha_{\mathbf{a}(j)} + \alpha_{\mathbf{b}(j)})(\alpha_{\mathbf{a}(j)} + \alpha_{\mathbf{c}(j)})
= (\alpha_{\mathbf{a}(j)})^2 + (\beta_{\mathbf{a}(j)})^2$, by \eqref{eq:kshkjhk3}.
Since the left hand side is essentially constant by (d.1), so is the right hand side.
\item By (d.1), it only remains to understand the behavior of $\alpha_{\mathbf{a}(j)}$, $\beta_{\mathbf{a}(j)}$. By
\eqref{eq:kshkjhk1}, we essentially have
\begin{equation}\label{erf.322}
\tfrac{\dd}{\dd \tau}\alpha_{\mathbf{a}(j)} =
- (\beta_{\mathbf{a}(j)})^2
\qquad \tfrac{\dd}{\dd \tau} \beta_{\mathbf{a}(j)}
= \alpha_{\mathbf{a}(j)}\beta_{\mathbf{a}(j)}
\end{equation}
A special solution is $\alpha_{\mathbf{a}(j)} = - \tanh \tau$ and $\beta_{\mathbf{a}(j)} = \pm \sech \tau
= \pm (\cosh \tau)^{-1}$. The general solution is obtained from the special solution by applying the affine symmetry transformation (ii) above, with $p > 0$. Since $\mathcal{B}_j$ is essentially the interval on which $|\beta_{\mathbf{a}(j)}|$ is not small, see (b.2), we must have $p \sim |\mathcal{B}_j|^{-1}$ (here $\sim$ means ``same order of magnitude''). See \cite{BKL1}, Section 3, in particular pages 534 and 535. 

\item
Recall (c.1).  By (d.3), we have 
$\alpha_{\mathbf{a}(j)}|_{\mathcal{S}_{j-1}} = - \alpha_{\mathbf{a}(j)}|_{\mathcal{S}_{j}}$, since the hyperbolic tangent just flips the sign. Therefore, by (d.1), the net change across $\mathcal{B}_j$ of the components of $\alpha$, from right to left, is given by
\begin{alignat*}{4}
\alpha_{\mathbf{a}(j)}|_{\mathcal{S}_{j-1}} &= \alpha_{\mathbf{a}(j)}|_{\mathcal{S}_{j}} \;&-&\; 2 \alpha_{\mathbf{a}(j)}|_{\mathcal{S}_{j}}\\
\alpha_{\mathbf{b}(j)}|_{\mathcal{S}_{j-1}} &= \alpha_{\mathbf{b}(j)}|_{\mathcal{S}_{j}} &+&\; 2 \alpha_{\mathbf{a}(j)}|_{\mathcal{S}_{j}}\\
\alpha_{\mathbf{c}(j)}|_{\mathcal{S}_{j-1}} &= \alpha_{\mathbf{c}(j)}|_{\mathcal{S}_{j}} &+&\; 2 \alpha_{\mathbf{a}(j)}|_{\mathcal{S}_{j}}
\end{alignat*}
These equations make sense only for $j \geq 2$, since $\mathcal{S}_0$ has not been defined.
\end{enumerate}
In this paper, we turn the heuristic picture
of \cite{BKL1}, sketched above, into a mathematically rigorous one, globally on $[0,\infty)$, for a generic class of solutions. The first step is to construct a discrete dynamical system, that maps the state $\Phi(\tau_j)$ to the state $\Phi(\tau_{j-1})$ at the earlier time $\tau_{j-1} < \tau_j$, for all $j \geq 1$. That is, the construction proceeds from right-to-left, beginning at $\tau =  + \infty$. We refer to the discrete dynamical system maps as transfer maps. 
\vskip 1mm
For each $j \geq 0$, two components of $\beta(\tau_j)$ have the same absolute value, see (c.4), and $\Phi(\tau_j)$ satisfies the constraint \eqref{eq:kshkjhk3}. Therefore, the states of the discrete dynamical system have 4 continuous degrees of freedom. By the symmetry (ii), the transfer maps commute with rescalings. Taking the quotient, one obtains a 3-dimensional discrete dynamical system. The three ``dimensionless'' quantities that we use to parametrize the discrete states are denoted $\mathbf{f}_j = (\mathbf{h}_j,w_j,q_j)$. Morally, they are interpreted as follows:
\begin{itemize}
\item $\mathbf{h}_j \sim |\mathcal{B}_j|/|\mathcal{I}_j| > 0$. 
In the billiard picture of \cite{Mis}, it is the dimensionless ratio of the collision and free-motion times. By (a.3), one has $0 < \mathbf{h}_j \ll 1$. In fact, $\mathbf{h}_j$ is the all-important small parameter in our construction. It goes to zero rapidly as $j \to \infty$. This is necessary for us to make a global construction on $[0,\infty)$. The precise rate depends on the sequence $(k_n)_{n \geq 1}$.
The rate is the same as in Proposition \ref{prop:kjdhkhkd}, up to even smaller corrections.
\item The components of $\alpha$ are essentially constant on $\mathcal{S}_j$ and subject to the reduced constraint equation in (c.2). Thus, modulo the scaling symmetry (ii), only one degree of freedom is required to parametrize $\alpha|_{\mathcal{S}_j}$. We use $w_j \approx - (\alpha_{\mathbf{b}(j)} / (\alpha_{\mathbf{a}(j)} + \alpha_{\mathbf{b}(j)}))|_{\mathcal{S}_j}$. By (c.2) and (c.3), we have $w_j > 0$. The left-to-right discrete dynamics of $w_j$ (which is opposite to the right-to-left direction of our transfer maps) is closely related to a variant of the Gauss map, sometimes referred to as the \emph{BKL map} or \emph{Kasner map}.
\item The meaning of $q_j$ will be explained in a more indirect way.
As pointed out above, the left-to-right dynamics of $w_j$ is related to the Gauss map, which is a non-invertible left-shift, see \eqref{eq:kdkdhdfkjd}. The non-invertibility of the Gauss map seems to be at odds with the invertible dynamics of the system of ordinary differential equations \eqref{eq:kshkjhkf}. The parameter $q_j$ is introduced so that the \emph{joint} left-to-right discrete dynamics of $(w_j,q_j)$ is closely related to the left-shift on \emph{two-sided} sequences $(k_n)_{n \in \Z}$ of strictly positive integers, which is invertible. Accordingly, the right-to-left transfer maps are related to the right shift on two-sided sequences 
$(k_n)_{n \in \Z}$.
\end{itemize}
This concludes the informal discussion.
We emphasize that the notation used above is specific to the introduction. In particular, $(\mathcal{I}_j)_{j \geq 1}$, $(\mathcal{B}_j)_{j\geq 1}$, $(\mathcal{S}_j)_{j \geq 1}$ do not appear in the main text. Starting from Section \ref{sec:dkfhkjdhkfd}, all the notation is introduced from scratch.
\vskip 1mm
 We now state simplified, self-contained versions of our results.
 References to their stronger counterparts are given. Here is a short guide:
 \begin{description}
\item[\emph{Definition \ref{def:ljfdljdjfdl} (equivalent to Definition \ref{def:kdhkhskhkdssPT1111star}).}] Introduces the state vectors $\Phi_{\star}(\pi,\mathbf{f},\sigma_{\ast})$ of the 3-dimensional discrete dynamical system. The dynamics of the signature vector $\sigma_{\ast}$ is trivial, by (iii), but it affects the dynamics of $(\pi,\mathbf{f})$ in a non-trivial way.
\item[\emph{Definition \ref{def:kdhkhskhkdssPT2xyz} (this is Definition \ref{def:kdhkhskhkdssPT2}).}]Introduces explicit maps $\mathcal{P}_L$, $\mathcal{Q}_L$, $\lambda_L$ that turn out to be very good approximations to the transfer maps. It is shown in Section \ref{sec:dlhjdfll} that iterates of $\mathcal{Q}_L$ can be understood in terms of the Gauss map\,/\,continued fractions and, by a change of variables, in terms of solutions to certain linear equations.
\item[\emph{Definition \ref{cccccc} (only in the main text).}]
The essential smallness condition on $\mathbf{h} > 0$ is quantitatively encoded in an open subset $\mathcal{F}\subset (0,1)\times (0,\infty) \times ((0,\infty)\setminus \{1\})$. It determines the domain of definition of the transfer maps.
\item[\emph{Proposition \ref{prop:fdkhdkhkd} (slimmed-down version of Proposition \ref{prop:skhdkjhfd}).}]
It asserts the existence of transfer maps. The pair $(\mathcal{P}_L,\Pi)$ and the triple $(\mathcal{P}_L,\Pi,\Lambda)$ constitute the transfer maps for the 3-dimensional and 4-dimensional systems, respectively, and they are very close to $(\mathcal{P}_L,\mathcal{Q}_L)$ and $(\mathcal{P}_L,\mathcal{Q}_L,\lambda_L)$. Explicit error bounds and precise estimates for the transfer solution appear only in the full version, Proposition \ref{prop:skhdkjhfd}.
\item[\emph{Theorem \ref{thm:fdkhjkh} (simplified version of Theorems \ref{thm:main2}, \ref{thm_main3}).}]
Gives a generic class of iterates to $(\mathcal{P}_L,\Pi)$ that are super-exponentially close to iterates of $(\mathcal{P}_L,\mathcal{Q}_L)$. That is, it asserts the existence of solutions to the 3-dimensional discrete dynamical system. 
\end{description}
The overview is as follows. Every solution to the 3-dimensional discrete dynamical system as in Theorem \ref{thm:fdkhjkh} can be lifted to a unique solution to the 4-dimensional discrete dynamical system, up to an overall scale, through the map  $\Lambda$ in Proposition \ref{prop:fdkhdkhkd}. This solution corresponds to the sequence of states $(\Phi(\tau_j))_{j \geq 0}$ in the informal discussion. Proposition \ref{prop:fdkhdkhkd} gives solutions to \eqref{eq:kshkjhkf} on compact intervals that connect next-neighbor states. Symmetry (ii) is used to translate these compact intervals and place them next to each other, beginning at $\tau = 0$, just like the $(\mathcal{I}_j)_{j \geq 1}$ in the informal discussion. As in (a.2) of the informal discussion, the union of these intervals is indeed $[0,\infty)$, and a semi-global solution to \eqref{eq:kshkjhkf} is obtained. To see this, denote the states by $\lambda_j\, \Phi_{\star}(\pi_j,\mathbf{g}_j,\sigma_{\ast})$ with $\lambda_j > 0$ and $\pi_j \in S_3$ and $\mathbf{g}_j = (\mathbf{h}_j',w_j',q_j') \in \mathcal{F}$, where $j \geq 0$. One has $\lambda_j = \Lambda[\pi_{j+1},\sigma_{\ast}](\mathbf{g}_{j+1}) \lambda_{j+1} \geq \lambda_{j+1}$ by the definition of $\Lambda$
and $\mathbf{h}_j' \in (0,1)$ by the definition of $\mathcal{F}$. In particular, the sequence of products $(\lambda_j\mathbf{h}_j')_{j \geq 0}$ is bounded from above by $\lambda_0 > 0$. By Proposition \ref{prop:fdkhdkhkd}, the length of each of the intervals is bounded from below by $(2\lambda_0)^{-1} > 0$.
\begin{definition}[State vectors] 
\label{def:ljfdljdjfdl}
Let $\pi = (\mathbf{a},\mathbf{b},\mathbf{c}) \in S_3$ and $\sigma_{\ast} \in \{-1,+1\}^3$ and $\mathbf{f} = (\mathbf{h},w,q)\in (0,\infty)^2\times \R$. Let
$\Phi_{\star} = \Phi_{\star}(\pi,\mathbf{f},\sigma_{\ast}) = \alpha \oplus \beta \in \R^3 \oplus \R^3$ be the vector given by $(\sgn \beta_1,\sgn \beta_2,\sgn \beta_3) = \sigma_{\ast}$ and by
\begin{align*}
\ua_{\mathbf{a}} & = -1 &  \mathbf{h}\log |\tfrac{1}{2}\ub_{\mathbf{a}}| & = 
- \tfrac{1+w}{1+2w} (1 + \mathbf{h} \log 2)\\
\ua_{\mathbf{b}}& = \tfrac{w}{1+w} &
\mathbf{h}\log |\tfrac{1}{2} \ub_{\mathbf{b}}| & = - \tfrac{1+w}{1+2w} (1 + \mathbf{h} \log 2)\\
\ua_{\mathbf{c}} & = -w - \mu
& \mathbf{h}\log |\tfrac{1}{2} \ub_{\mathbf{c}}| & = - (1+w)q
-  \tfrac{w(1+w)}{1+2w} -  \tfrac{1+3w+w^2}{1+2w} \mathbf{h} \log 2
\end{align*}
where $\mu = \mu(\pi,\mathbf{f},\sigma_{\ast}) \in \R$ is uniquely determined by requiring that \eqref{eq:kshkjhk3} holds.
\end{definition}
\begin{definition}[Approximate transfer maps] \label{def:kdhkhskhkdssPT2xyz}
Introduce three maps
\begin{alignat*}{6}
\mathcal{P}_L: &\quad& S_3\times (0,\infty)^3 & \to S_3 \hskip20mm & ((\mathbf{a},\mathbf{b},\mathbf{c}),\mathbf{f}) & \mapsto  (\mathbf{a}',\mathbf{b}',\mathbf{c}')\\
\mathcal{Q}_L: && (0,\infty)^3 & \to (0,\infty)^2\times \R & \mathbf{f} & \mapsto  (\mathbf{h}_L,w_L,q_L)\\
\lambda_L: && (0,\infty)^3 & \to (0,\infty) & \mathbf{f} & \mapsto \lambda_L
\end{alignat*}
where $\mathbf{f} = (\mathbf{h},w,q)$ and $q_L = \mathrm{num1}_L/\mathrm{den}_L$ and $\mathbf{h}_L = \mathrm{num2}_L/\mathrm{den}_L$, and:
\vskip 1mm \noindent \;
$\bullet$ if $q \leq 1$:
\begin{align*}
(\mathbf{a}',\mathbf{b}',\mathbf{c}') & = (\mathbf{c},\mathbf{a},\mathbf{b}) &
\mathrm{num1}_L & = (1+w)(1-q)-\mathbf{h} \log 2 + \mathbf{h} w \log(2+w)\\
w_L & = \tfrac{1}{1+w} &
\mathrm{num2}_L & = \mathbf{h}(2+w)\\
\lambda_L & = 2+w &
\mathrm{den}_L & = (1+w)(q - \mathbf{h} \log 2) + \mathbf{h} (3+w) \log (2+w)
\intertext{\; $\bullet$ if $q > 1$:}
(\mathbf{a}',\mathbf{b}',\mathbf{c}') & = (\mathbf{b},\mathbf{a},\mathbf{c}) &
\mathrm{num1}_L & = (1+w)(q-1 -\mathbf{h}\log 2)  - \mathbf{h}w \log \tfrac{2+w}{1+w}\\
w_L & = 1+w &
\mathrm{num2}_L & = \mathbf{h}(2+w)\\
\lambda_L & = \tfrac{2+w}{1+w} &
\mathrm{den}_L & = (1+w)-\mathbf{h}\log 2 + \mathbf{h}(3+2w) \log \tfrac{2+w}{1+w}
\end{align*}
Observe that $\mathrm{den}_L > 0$.
\end{definition}
\begin{proposition}[Transfer maps] \label{prop:fdkhdkhkd}
Fix $\sigma_{\ast} \in \{-1,+1\}^3$ and $\pi = (\mathbf{a},\mathbf{b},\mathbf{c})\in S_3$. There exist maps\footnote{Caution: The maps $\Pi$ cannot immediately be iterated\,/\,composed, because $(0,\infty)^2\times \R \not\subset \mathcal{F}$.}
$$\Pi[\pi,\sigma_{\ast}]:\;\;\mathcal{F}\to (0,\infty)^2\times \R
\qquad
\text{and}
\qquad \Lambda[\pi,\sigma_{\ast}]:\;\;\mathcal{F} \to [1,\infty)$$ such that for every
$\lambda > 0$ and 
$\mathbf{f} = (\mathbf{h},w,q) \in \mathcal{F}$, the solution to 
 \eqref{eq:kshkjhkf} starting at $\lambda\, \Phi_{\star}(\pi,\mathbf{f},\sigma_{\ast})$ at time $0$ passes through $\lambda'\,\Phi_{\star}(\pi',\mathbf{f}',\sigma_{\ast})$ at an earlier time $\tau' < 0$, with
 $\tfrac{1}{2} \leq \mathbf{h}\,\lambda\,|\tau'| \leq 3$.
  Here 
$\mathbf{f}' = \Pi[\pi,\sigma_{\ast}](\mathbf{f})$ and
$\lambda' = \lambda\, \Lambda[\pi,\sigma_{\ast}](\mathbf{f})$ and
$\pi' = \mathcal{P}_L(\pi,\mathbf{f})$. Schematically, the transition is
$$
\lambda\, \Lambda[\pi,\sigma_{\ast}](\mathbf{f})\;\Phi_{\star}\Big(
\mathcal{P}_L(\pi,\mathbf{f}),\;\Pi[\pi,\sigma_{\ast}](\mathbf{f}),\;\sigma_{\ast}\Big)
\qquad \longleftarrow \qquad
\lambda\, \Phi_{\star}(\pi,\mathbf{f},\sigma_{\ast})
$$
Furthermore (informal): $\Pi$ and $\Lambda$ are approximated by the maps $\mathcal{Q}_L$ and $\lambda_L$, with errors that go to zero exponentially as $\mathbf{h}\downarrow 0$ (for fixed $w$, $q$). See  Proposition \ref{prop:skhdkjhfd}
\end{proposition}
\begin{theorem}\label{thm:fdkhjkh}
Fix $\sigma_{\ast}\in \{-1,+1\}^3$ and $\pi_0 \in S_3$. Fix constants $\mathbf{D} \geq 1$,\, $\gamma \geq 0$. Suppose the vector $\mathbf{f}_0 = (\mathbf{h}_0,w_0,q_0) \in (0,\infty)^3$ satisfies
\begin{enumerate}[(i)]
\item $w_0 \in (0,1)\setminus \mathbb{Q}$ and $q_0 \in (0,\infty)\setminus \mathbb{Q}$.
\item $k_n \leq \mathbf{D}\,\max\{1,n\}^{\gamma}$ for all $n \geq -2$, where the two-sided sequence of strictly positive integers $(k_n)_{n\in \Z}$ is given by
$$(1+q_0)^{-1} = \langle k_0,k_{-1},k_{-2},\ldots \rangle
\qquad w_0 = \langle k_1,k_2,k_3,\ldots \rangle$$
\item $0< \mathbf{h}_0< \mathbf{A}^{\sharp}$ where $\mathbf{A}^{\sharp} = \mathbf{A}^{\sharp}(\mathbf{D},\gamma) = 2^{-56}\mathbf{D}^{-4} (4(\gamma+1))^{-4(\gamma+1)}$.
\end{enumerate}
Then  $\mathbf{f}_0$ and $\pi_0$ are the first elements of a unique sequence
$(\mathbf{f}_j)_{j \geq 0}$ in $\mathcal{F}$
and  a unique sequence  $(\pi_j)_{j \geq 0}$  in $S_3$, respectively, with
$\pi_j = \mathcal{P}_L(\pi_{j+1},\mathbf{f}_{j+1})$ and
$\mathbf{f}_j = \mathcal{Q}_L(\mathbf{f}_{j+1})$ for all $j \geq 0$. Furthermore, there exists a sequence $(\mathbf{g}_j)_{j \geq 0}$ in $\mathcal{F}$ such that for all $j \geq 0$,
$$\mathbf{g}_j = \Pi[\pi_{j+1},\sigma_{\ast}](\mathbf{g}_{j+1}) \qquad \text{and} \qquad \pi_{j} = \mathcal{P}_L(\pi_{j+1},\mathbf{g}_{j+1})$$ and,
with $\gold_+ = \tfrac{1}{2}(1+\sqrt{5})$,
$$\|\mathbf{g}_j - \mathbf{f}_j\|_{\R^3} \leq
\exp \Big( - \tfrac{1}{\mathbf{h}_0} \mathbf{A}^{\sharp}\,\gold_+^{((\mathbf{D}^{-1}j)^{1/( \gamma+1)})}\Big)
$$
If $\gamma > 1$
and $\mathbf{D} > \tfrac{1}{\log 2} \tfrac{\gamma}{\gamma-1}$, then the set of all vectors $\mathbf{f}_0\in (0,\infty)^3$ that satisfy (i), (ii), (iii) has positive Lebesgue measure.
\end{theorem}

The class of solutions that we construct is generic in the sense of the last sentence of Theorem \ref{thm:fdkhjkh}. It would be desirable to have a stronger genericity statement, namely a genericity statement for ``the $\mathbf{g}_0$ rather than the $\mathbf{f}_0$''. 
\vskip 2mm

For the causal structure and particle horizons, see
Proposition \ref{prop:dskhdkhfd} and Section \ref{sec:fkhkhh}.

\vskip 2mm
It is a pleasure to thank J. Fr\"ohlich, G.M. Graf and T. Spencer for their support and encouragement.

%% file: interval.pstex_t
\begin{picture}(0,0)%
\includegraphics{interval.pstex}%
\end{picture}%
\setlength{\unitlength}{3158sp}%
\begingroup\makeatletter\ifx\SetFigFont\undefined%
\gdef\SetFigFont#1#2#3#4#5{%
  \reset@font\fontsize{#1}{#2pt}%
  \fontfamily{#3}\fontseries{#4}\fontshape{#5}%
  \selectfont}%
\fi\endgroup%
\begin{picture}(6066,650)(1468,-3284)
\put(2776,-3136){\makebox(0,0)[lb]{\smash{{\SetFigFont{10}{12.0}{\familydefault}{\mddefault}{\updefault}{\color[rgb]{0,0,0}$\mathcal{I}_j$}%
}}}}
\put(5165,-2813){\makebox(0,0)[lb]{\smash{{\SetFigFont{10}{12.0}{\familydefault}{\mddefault}{\updefault}{\color[rgb]{0,0,0}$\mathcal{B}_{j+1}$}%
}}}}
\put(6384,-2812){\makebox(0,0)[lb]{\smash{{\SetFigFont{10}{12.0}{\familydefault}{\mddefault}{\updefault}{\color[rgb]{0,0,0}$\mathcal{S}_{j+1}$}%
}}}}
\put(2492,-2805){\makebox(0,0)[lb]{\smash{{\SetFigFont{10}{12.0}{\familydefault}{\mddefault}{\updefault}{\color[rgb]{0,0,0}$\mathcal{B}_j$}%
}}}}
\put(5476,-3136){\makebox(0,0)[lb]{\smash{{\SetFigFont{10}{12.0}{\familydefault}{\mddefault}{\updefault}{\color[rgb]{0,0,0}$\mathcal{I}_{j+1}$}%
}}}}
\put(3858,-2812){\makebox(0,0)[lb]{\smash{{\SetFigFont{10}{12.0}{\familydefault}{\mddefault}{\updefault}{\color[rgb]{0,0,0}$\mathcal{S}_j$}%
}}}}
\put(4126,-3211){\makebox(0,0)[lb]{\smash{{\SetFigFont{10}{12.0}{\familydefault}{\mddefault}{\updefault}{\color[rgb]{0,0,0}$\tau_j$}%
}}}}
\put(7126,-3211){\makebox(0,0)[lb]{\smash{{\SetFigFont{10}{12.0}{\familydefault}{\mddefault}{\updefault}{\color[rgb]{0,0,0}$\tau_{j+1}$}%
}}}}
\end{picture}%

%% file: SectionEquations.tex
\section{Spatially homogeneous vacuum spacetimes}\label{sec:dkfhkjdhkfd}
\begin{proposition}\label{prop:khdkhkhkhss}
Let $\ua\oplus \ub: (\tau_0,\tau_1)\to \R^3\oplus \R^3$ be a solution to
\eqref{eq:kshkjhkf}
and let $\Omega \subset \R^3$ be open, with Cartesian coordinates $\mathbf{x} = (x^1,x^2,x^3)$.
Fix any $\tau_{\ast} \in (\tau_0,\tau_1)$ and let
\begin{align*}
v_1 & = \textstyle\sum_{\mu = 1}^3 {v_1}^{\mu}(\mathbf{x})\tfrac{\p}{\p x^{\mu}} &
v_2 & = \textstyle\sum_{\mu = 1}^3 {v_2}^{\mu}(\mathbf{x})\tfrac{\p}{\p x^{\mu}} &
v_3 & = \textstyle\sum_{\mu = 1}^3 {v_3}^{\mu}(\mathbf{x})\tfrac{\p}{\p x^{\mu}}
\end{align*}
be three smooth vector fields on $\Omega$ that are a frame at each point and satisfy
$$[v_{\xb},v_{\xc}] = \ub_{\xa}(\tau_{\ast})\, v_{\xa}\qquad \text{on $\Omega$}$$
for all $(\xa,\xb,\xc)\in \mathcal{C} \stackrel{\text{def}}{=} \{(1,2,3), (2,3,1), (3,1,2)\}$.
Introduce 
\begin{align*}
e_0 & = e^{\zeta(\tau)} \tfrac{\p}{\p \tau} &
e_{\xa} & = e^{\zeta_{\xa}(\tau)}
v_{\xa}
& \xa=1,2,3\\
\zeta(\tau) & = \zeta_1(\tau)+\zeta_2(\tau) + \zeta_3(\tau) &
\zeta_{\xa}(\tau) & =  -\tfrac{1}{2} \textstyle\int_{\tau_{\ast}}^{\tau}\dd s\, \ua_{\xa}(s)
& \xa=1,2,3
 \end{align*}
on the domain $(\tau_0,\tau_1)\times \Omega \subset \R^4$.
Then, the Lorentzian metric $g$ with inverse
$$
g^{-1} = - e_0\otimes e_0
+ e_1\otimes e_1 + e_2 \otimes e_2 + e_3\otimes e_3$$
is a solution to the vacuum Einstein equations $\mathrm{Ric}(g) = 0$ on $(\tau_0,\tau_1)\times \Omega$.
\end{proposition}
\begin{proof}
In this proof, everywhere $(\xa,\xb,\xc)\in \mathcal{C}$. It follows from $\td{e^{-2\zeta_{\xa}}} = \ua_{\xa} e^{-2\zeta_{\xa}}$ and $\zeta_{\xa}(\tau_{\ast}) =0$
and \eqref{eq:kshkjhk2} that $\ub_{\xa}(\tau) = \ub_{\xa}(\tau_{\ast}) e^{-2\zeta_{\xa}(\tau)}$. Now, by direct calculation,
\begin{align*}
[e_0,e_{\xa}] & = -\tfrac{1}{2} e^{\zeta} \ua_{\xa}  e_{\xa}
& 
[e_{\xb},e_{\xc}] & = e^{\zeta} \ub_{\xa} e_{\xa}
\end{align*}
Let $\nabla$ be the Levi-Civita connection associated to $g$. Then, for all $a,b,c=0,1,2,3$,
$$
g\big(\nabla_{e_{a}}e_{b},\,e_{c}\big)
= \tfrac{1}{2}\Big(
g\big([e_{a},e_{b}],e_{c}\big)
-g\big([e_{b},e_{c}],e_{a}\big)
+g\big([e_{c},e_{a}],e_{b}\big)
\Big)
$$
By direct calculation,
\begin{align*}
\nabla_{e_0}e_0 & = 0 & \nabla_{e_{\xa}}e_{\xa} & = \tfrac{1}{2}e^{\zeta} \ua_{\xa} e_0\\
\nabla_{e_0}e_{\xa} & = 0 & \nabla_{e_{\xb}}e_{\xc} & = \tfrac{1}{2}e^{\zeta}(+\ub_{\xa} - \ub_{\xb}+\ub_{\xc}) e_{\xa}\\
\nabla_{e_{\xa}}e_0 & = \tfrac{1}{2}e^{\zeta} \ua_{\xa} e_{\xa} & \nabla_{e_{\xc}}e_{\xb} & = \tfrac{1}{2}e^{\zeta}(-\ub_{\xa} - \ub_{\xb}+\ub_{\xc}) e_{\xa}
\end{align*}
and
\begin{align*}
& \mathrm{Riem}(e_{\xa},e_{\xb},e_{\xa},e_{\xb}) \\
& = \tfrac{1}{4}e^{2\zeta}
\Big(
(+ \beta_{\xa} - \beta_{\xb} - \beta_{\xc})
(+\beta_{\xa} - \beta_{\xb} + \beta_{\xc})
+ 2\beta_{\xc}(+\beta_{\xa} + \beta_{\xb} - \beta_{\xc})
+ \alpha_{\xa} \alpha_{\xb}
\Big)\\
\rule{0pt}{11pt} & \mathrm{Riem}(e_0,e_{\mathbf{a}},e_{\xa},e_{\xb})\\
& = \tfrac{1}{4} e^{2\zeta} \delta_{\mathbf{a}\mathbf{k}}
\Big(
(-\beta_{\xa} + \beta_{\xb} - \beta_{\xc}) \alpha_{\xa}
+ (+\beta_{\xa} - \beta_{\xb}-\beta_{\xc}) \alpha_{\xb} + 2 \alpha_{\xc}\beta_{\xc}
\Big)\\
\rule{0pt}{11pt} & \mathrm{Riem}(e_0,e_{\mathbf{a}},e_0,e_{\xa})\\
& =
- \tfrac{1}{4}  e^{2\zeta}\,
\delta_{\mathbf{a}\xa}\,\Big(
 2 \tfrac{\dd}{\dd \tau}  \alpha_{\xa}
- (\alpha_{\xb} + \alpha_{\xc}) \alpha_{\xa}
\Big)
\end{align*}
Furthermore, $\mathrm{Riem}(e_{\mathbf{a}},e_{\mathbf{b}},e_{\mathbf{c}},e_{\mathbf{d}}) = 0$ unless
$\{\mathbf{a},\mathbf{b}\} = \{\mathbf{c},\mathbf{d}\}$ with $\mathbf{a}\neq \mathbf{b}$. The
Riemann curvature tensor is completely
specified by these identities and by its algebraic symmetries. It follows that
\begin{alignat*}{4}
&\mathrm{Ric}(e_0,e_0) && = -\tfrac{1}{2}e^{2\zeta}\,\tfrac{\dd}{\dd \tau}(\ua_1+\ua_2+\ua_3)
+ \tfrac{1}{2}e^{2\zeta}(\ua_2\ua_3+\ua_3\ua_1+\ua_1\ua_2)\\
&\mathrm{Ric}(e_0,e_{\xa}) && = 0\\
&\mathrm{Ric}(e_{\xa},e_{\xa}) && = 
+\tfrac{1}{2} e^{2\zeta} \td{\ua_{\xa}}
+ \tfrac{1}{2}e^{2\zeta}\big(+ (\ub_{\xa})^2 - (\ub_{\xb})^2 - (\ub_{\xc})^2 + 2\ub_{\xb}\ub_{\xc} \big)\\
&\mathrm{Ric}(e_{\xb},e_{\xc}) && = 0
\end{alignat*}
The right hand sides of the first and third equation vanish by \eqref{eq:kshkjhk1} and \eqref{eq:kshkjhk3}. \qed
\end{proof}
\begin{proposition} \label{prop:dskhdkhfd}
In the context of Proposition \ref{prop:khdkhkhkhss}, let $\gamma: (\tau_0',\tau_1') \to (\tau_0,\tau_1)\times \Omega$ be a smooth curve given by $\gamma(\tau) = (\tau,\gamma^{\sharp}(\tau))$, where $\gamma^{\sharp}$ is a curve on $\Omega$. Let $g^{\sharp}$ be the Riemannian metric on $\Omega$ defined by $g^{\sharp}(v_{\mathbf{a}},v_{\mathbf{b}}) = \delta_{\mathbf{a}\mathbf{b}}$ for all $\mathbf{a},\mathbf{b}=1,2,3$.
If $\gamma$ is non-spacelike with respect to $g$, then the length of $\gamma^{\sharp}$ with respect to $g^{\sharp}$ is bounded by
\begin{equation}\label{eq:dfkhkdhkfddhfkdhdk}
\mathrm{Length}_{g^{\sharp}}(\gamma^{\sharp})\; \leq\;  \int_{\tau_0'}^{\tau_1'}
\dd \tau\, \max_{(\xa,\xb,\xc)\in \mathcal{C}}
e^{-\zeta_{\xb}-\zeta_{\xc}}
\end{equation}
The integral on the right hand side may be divergent.
\end{proposition}
\begin{proof}
Write the velocity $\tfrac{\dd}{\dd \tau}\gamma$ as
$$\tfrac{\p}{\p \tau} + \textstyle\sum_{\xa=1}^3 X^{\xa} v_{\xa}
= e^{-\zeta} e_0 + \textstyle\sum_{\xa=1}^3 X^{\xa} e^{-\zeta_{\xa}} e_{\xa}$$
with smooth coefficients $X^{\xa} = X^{\xa}(\tau)$. 
By assumption, $\gamma$ is non-spacelike:
$$0 \geq g(\tfrac{\dd}{\dd \tau}\gamma,\tfrac{\dd}{\dd \tau}\gamma)
= -e^{-2\zeta} + \textstyle\sum_{\xa=1}^3 (X^{\xa})^2 e^{-2\zeta_{\xa}}
$$
Consequently,
$\textstyle\sum_{\xa=1}^3 (X^{\xa})^2 \leq 
\max_{(\xa,\xb,\xc)\in \mathcal{C}} e^{-2\zeta_{\xb}-2\zeta_{\xc}}$.
Now, the claim follows from
$$
\mathrm{Length}_{g^{\sharp}}(\gamma^{\sharp}) = \int_{\tau_0'}^{\tau_1'} \dd \tau\, \sqrt{g^{\sharp}(\tfrac{\dd}{\dd \tau}\gamma^{\sharp},\tfrac{\dd}{\dd \tau}\gamma^{\sharp})}
= \int_{\tau_0'}^{\tau_1'} \dd \tau\, \sqrt{
\textstyle\sum_{\xa=1}^3 (X^{\xa})^2}
$$
\qed
\end{proof}
\section{Construction of the transfer maps}\label{sec:lfdjdfljfdl}
Let $(\tau_0,\tau_1)\subset \R$ be a finite or infinite open interval, parametrized by $\tau \in (\tau_0,\tau_1)$ (``time''). In this paper, the unknown field is a vector valued map $\Phi \in C^{\infty}((\tau_0,\tau_1),\R^6)$:
\begin{equation}\label{eq:ljfkdhkjhdkjdfh}
\Phi =  \ua[\Phi] \oplus \ub[\Phi]
\;:\; (\tau_0,\tau_1)\to \R^3 \oplus \R^3
\end{equation}
If no confusion can arise, we just write $\Phi = \ua\oplus \ub$.
\begin{definition}\label{def:dkhkdhfd}
To every field $\Phi = \ua\oplus \ub \in C^{\infty}((\tau_0,\tau_1),\R^6)$, every constant $\mathbf{h} > 0$ and every $n \in \R^3$, associate a field
$$
\boldsymbol{\mathfrak a}[\Phi,\mathbf{h},n] \oplus
\boldsymbol{\mathfrak b}[\Phi,\mathbf{h},n] \oplus
c[\Phi,\mathbf{h},n]
\;:\;
(\tau_0,\tau_1) \to \R^3 \oplus \R^3 \oplus \R$$
by
\begin{subequations}\label{eq:khdkhdkhfdfddf}
\begin{align}
\begin{split}
\boldsymbol{\mathfrak a}_{\xa}[\Phi,\mathbf{h},n] & =
- \mathbf{h} \td{\ua_{\xa}} 
- (n_{\xa}\ub_{\xa})^2
+ (n_{\xb}\ub_{\xb} - n_{\xc}\ub_{\xc})^2
\end{split}
\\ 
\boldsymbol{\mathfrak b}_{\xa}[\Phi,\mathbf{h},n] & =
- \mathbf{h} \td{\ub_{\xa}} 
 + \ub_{\xa}\ua_{\xa}
\\ 
\label{eq:khdkhdkhfdfddf1}
\begin{split}
c[\Phi,\mathbf{h},n] & =
\textstyle\sum_{(\xa,\xb,\xc)\in \mathcal{C}} \big(\ua_{\xb}\ua_{\xc}
- (n_{\xa}\ub_{\xa})^2 + 2n_{\xb}n_{\xc}\ub_{\xb}\ub_{\xc} \big)
\end{split}
\end{align}
\end{subequations}
for all $(\xa,\xb,\xc)\in \mathcal{C}$.
For later use, it is convenient to introduce, for all $m,n\in \R^3$,
\begin{equation}\label{eq:kdshkfdhkfdhfdkhfdkdhfkdfhkdssdfds}
\begin{split}
\boldsymbol{\mathfrak a}_{\xa}[\Phi,\mathbf{h},n,m] & = \boldsymbol{\mathfrak a}_{\xa}[\Phi,\mathbf{h},n] - \boldsymbol{\mathfrak a}_{\xa}[\Phi,\mathbf{h},m]\\
& = - (n_{\xa}\ub_{\xa})^2
+ (n_{\xb}\ub_{\xb} - n_{\xc}\ub_{\xc})^2
+ (m_{\xa}\ub_{\xa})^2
- (m_{\xb}\ub_{\xb} - m_{\xc}\ub_{\xc})^2
\end{split}
\end{equation}
\end{definition}
\begin{definition} \label{def:nnnnnns}
\begin{align*}
B_1 & = (1,0,0) & 
B_2 & = (0,1,0) &
B_3 & = (0,0,1) &
Z & = (1,1,1)
\end{align*}
These vectors will play the role of the vector $n \in \R^3$ that appears in 
Definition  \ref{def:dkhkdhfd}.
\end{definition}
\begin{proposition}[Global symmetries] \label{prop:kdhfkshkds}
Let $\chi:  (\tau_0,\tau_1)\to (\tau_0',\tau_1')$ be a linear diffeomorphism between finite or infinite intervals, $\chi(\tau) = p\tau + q$ with $p > 0$, and let $A > 0$ be a constant. Then
\begin{align*}
(\boldsymbol{\mathfrak a},\boldsymbol{\mathfrak b},c)\Big[A\,(\Phi \circ \chi),\;\tfrac{1}{p} A \mathbf{h},\,n\Big]
& = A^2\,\Big((\boldsymbol{\mathfrak a},\boldsymbol{\mathfrak b},c)[\Phi,\mathbf{h},n] \circ \chi\Big)
\end{align*}
for all fields $\Phi = \ua\oplus \ub\in C^{\infty}((\tau_0',\tau_1'),\R^6)$, all constants $\mathbf{h} > 0$ and all $n \in \R^3$.
\end{proposition}
\begin{corollary}\label{cor:dfjdkd} In Proposition \ref{prop:kdhfkshkds},
the field
$(\boldsymbol{\mathfrak a},\boldsymbol{\mathfrak b},c)[A\,(\Phi \circ \chi),\;\tfrac{1}{p} A \mathbf{h},\,n]$ vanishes identically on $(\tau_0,\tau_1)$ if and only if 
$(\boldsymbol{\mathfrak a},\boldsymbol{\mathfrak b},c)[\Phi,\mathbf{h},n]$ vanishes identically on
$(\tau_0',\tau_1')$.
\end{corollary}
\begin{remark}
The equations $(\boldsymbol{\mathfrak a},\boldsymbol{\mathfrak b},c)[\Phi,1,Z] = 0$
are identical to \eqref{eq:kshkjhkf}.
The equations $(\boldsymbol{\mathfrak a},\boldsymbol{\mathfrak b},c)[\Phi,\mathbf{h},Z] = 0$
are equivalent to  \eqref{eq:kshkjhkf}, for any $\mathbf{h} > 0$,  by Corollary \ref{cor:dfjdkd}.
\end{remark}
\begin{proposition}\label{prop:conslaw}
Recall Definition \ref{def:dkhkdhfd}. For all $\Phi = \ua\oplus \ub \in C^{\infty}((\tau_0,\tau_1),\R^6)$, all $\mathbf{h} > 0$ and all $n\in \R^3$, we have
\begin{align}
\begin{split}
\label{eq:diffid}
 0 = & - \mathbf{h} \tfrac{\dd}{\dd \tau} c  + \sum_{(\xa,\xb,\xc)\in \mathcal{C}}
\Big(\hskip-1pt
-\ua_{\xb} \boldsymbol{\mathfrak a}_{\xc}
-\ua_{\xc} \boldsymbol{\mathfrak a}_{\xb}
+2 (n_{\xa})^2 \ub_{\xa} \boldsymbol{\mathfrak b}_{\xa}
 -2 n_{\xb}n_{\xc}\ub_{\xb}\boldsymbol{\mathfrak b}_{\xc}
-2 n_{\xb}n_{\xc}\ub_{\xc}\boldsymbol{\mathfrak b}_{\xb}
\Big)
\end{split}
\end{align}
with $(\boldsymbol{\mathfrak a},\boldsymbol{\mathfrak b},c) =
(\boldsymbol{\mathfrak a},\boldsymbol{\mathfrak b},c)[\Phi,\mathbf{h},n]$. 
In particular, if $(\boldsymbol{\mathfrak a},\boldsymbol{\mathfrak b}) = 0$ identically
on $(\tau_0,\tau_1)$, then $c$ vanishes identically on $(\tau_0,\tau_1)$ if and only if $c$ vanishes at one point of $(\tau_0,\tau_1)$.
\end{proposition}
\begin{proof}
Replace all occurrences of $\boldsymbol{\mathfrak a}$, $\boldsymbol{\mathfrak b}$ and $c$
on the right hand side of \eqref{eq:diffid} by the respective right hand sides of \eqref{eq:khdkhdkhfdfddf}.
Then, verify that everything cancels. 
\qed
\end{proof}
\begin{definition} \label{def:khkhdkhkfdhfd}
For all $\mathbf{h} > 0$ and all vectors $\Phi = \ua\oplus \ub \in \R^3 \oplus \R^3$ with $\ub_1,\ub_2,\ub_3 \neq 0$, define
$A_{\mathbf{m}}[\Phi] \in (0,\infty)$ and $\varphi_{\mathbf{m}}[\Phi] \in \R$
by
\begin{alignat*}{4}
A_{\mathbf{m}}[\Phi]      & = \sqrt{|\ua_{\mathbf{m}}|^2 + |\ub_{\mathbf{m}}|^2} \;>\;|\ua_{\mathbf{m}}| \;\geq\; 0 \\
\varphi_{\mathbf{m}}[\Phi] & = - \arcsinh \tfrac{\ua_{\mathbf{m}}}{|\ub_{\mathbf{m}}|}
\end{alignat*}
for all $\mathbf{m}=1,2,3$. Equivalently,
\begin{subequations}\label{eq:khkdhdfkhkdhkdsds}
\begin{alignat}{4}
\ua_{\mathbf{m}} & = \;&- &\,A_{\mathbf{m}}[\Phi]\, \tanh &\, \varphi_{\mathbf{m}}[\Phi]\\
\ub_{\mathbf{m}} & = &
(\sgn \beta_{\mathbf{m}})\;  &\,A_{\mathbf{m}}[\Phi]\, \sech &\, \varphi_{\mathbf{m}}[\Phi]
\end{alignat}
\end{subequations}
Furthermore, define $\xi_{\mathbf{m}}[\Phi,\mathbf{h}]\in \R$ by
$$\xi_{\mathbf{m}}[\Phi,\mathbf{h}] = \mathbf{h}\, \log \big|\tfrac{1}{2}\ub_{\mathbf{m}}\big|$$
for all $\mathbf{m}=1,2,3$. Furthermore, for all $\mathbf{m},\mathbf{n}=1,2,3$, introduce the abbreviations
\begin{align*}
\ua_{\mathbf{m},\mathbf{n}}[\Phi] & = \ua_{\mathbf{m}}+\ua_{\mathbf{n}}
&
\xi_{\mathbf{m},\mathbf{n}}[\Phi,\mathbf{h}] & = \xi_{\mathbf{m}}[\Phi,\mathbf{h}]+\xi_{\mathbf{n}}[\Phi,\mathbf{h}]
\end{align*}
If no confusion can arise, we drop the explicit dependence $[\Phi]$ or $[\Phi,\mathbf{h}]$. For instance, we write $A_{\mathbf{m}} = A_{\mathbf{m}}[\Phi]$.
If $\Phi$ is not an element of $\R^3\oplus \R^3$, but rather a function of the real variable $\tau$ with values in $\R^3\oplus \R^3$,
with $\ub_1,\ub_2,\ub_3\neq 0$ everywhere, then $A_{\mathbf{m}},\varphi_{\mathbf{m}}$, 
$\xi_{\mathbf{m}}$, $\xi_{\mathbf{m},\mathbf{n}}$, $\alpha_{\mathbf{m},\mathbf{n}}$, with $\mathbf{m},\mathbf{n}=1,2,3$, are functions of $\tau$, too.
In this case, we define the additional functions $\theta_{\mathbf{m}}[\Phi,\mathbf{h}]$, $\mathbf{m}=1,2,3$,  through
$$\varphi_{\mathbf{m}}[\Phi](\tau) = \tfrac{1}{\mathbf{h}} \big( \tau - \theta_{\mathbf{m}}[\Phi,\mathbf{h}](\tau)\big)\, A_{\mathbf{m}}[\Phi](\tau)$$
\end{definition}
\begin{remark}\label{rem:kfdhdkhkfd}
In the context of Definition  \ref{def:khkhdkhkfdhfd}, we have, for all $\mathbf{m}=1,2,3$:
\begin{align*}
\mathbf{h}\,|\varphi_{\mathbf{m}}| & = - \xi_{\mathbf{m}} + \mathbf{h}\,\log \Big(|\tfrac{1}{2}\ua_{\mathbf{m}}|
+ \sqrt{|\tfrac{1}{2}\ua_{\mathbf{m}}|^2 + \exp( \tfrac{1}{\mathbf{h}} 2 \xi_{\mathbf{m}})}\;\Big)
\end{align*}
\end{remark}
\begin{lemma} \label{lem:lfdskhjkfdhkfdhkfd}
Recall Definitions \ref{def:dkhkdhfd}, \ref{def:nnnnnns}, \ref{def:khkhdkhkfdhfd}.
For all $\mathbf{h}>0$ and all $\Phi=\ua\oplus\ub \in C^{\infty}((\tau_0,\tau_1),\R^6)$ such that $\ub_1$, $\ub_2$, $\ub_3$ never vanish on $(\tau_0,\tau_1)$, we have
\begin{align*}
\frac{\dd}{\dd \tau} \begin{pmatrix}
A_{\xa} \\ \theta_{\xa}
\end{pmatrix}
= \frac{1}{(A_{\xa})^2} \begin{pmatrix}
\frac{1}{\mathbf{h}} (A_{\xa})^2 \tanh \varphi_{\xa} & & \frac{1}{\mathbf{h}} (A_{\xa})^2 \sech \varphi_{\xa} \\
\rule{0pt}{10pt} \varphi_{\xa} \tanh \varphi_{\xa} - 1&\;\, & \sinh \varphi_{\xa} + \varphi_{\xa} \sech \varphi_{\xa}
\end{pmatrix}
\begin{pmatrix}
\boldsymbol{\mathfrak a}_{\xa}[\Phi,\mathbf{h},B_{\xa}]\\
- \sigma_{\xa}\, \boldsymbol{\mathfrak b}_{\xa}[\Phi,\mathbf{h},B_{\xa}]
\end{pmatrix}
\end{align*}
for $\xa=1,2,3$ and $\sigma_{\xa} = \sgn \beta_{\xa} \in \{-1,+1\}$.
The matrix on the right hand side has determinant $\tfrac{1}{\mathbf{h}} (A_{\xa})^2 \cosh \varphi_{\xa} \neq 0$.
\end{lemma}
\begin{proof}
We have
$\boldsymbol{\mathfrak a}_{\xa}[\Phi,\mathbf{h},B_{\xa}] = - \mathbf{h} \tfrac{\dd}{\dd \tau} \ua_{\xa} - (\ub_{\xa})^2$
and 
$\boldsymbol{\mathfrak b}_{\xa}[\Phi,\mathbf{h},B_{\xa}]  = - \mathbf{h} \tfrac{\dd}{\dd \tau} \ub_{\xa} + \ua_{\xa}\ub_{\xa}$.
Replace all occurrences of $\ua_{\xa}$ and $\ub_{\xa}$ by the right hand sides of \eqref{eq:khkdhdfkhkdhkdsds}, respectively. Use $\tfrac{\dd}{\dd \tau}\varphi_{\xa} = \tfrac{1}{A_{\xa}} (\tfrac{\dd}{\dd \tau} A_{\xa}) \varphi_{\xa}
+ \tfrac{1}{\mathbf{h}} A_{\xa}(1-\tfrac{\dd}{\dd \tau}\theta_{\xa})$. Now, solve for $\tfrac{\dd}{\dd \tau}A_{\xa}$ and $\tfrac{\dd}{\dd \tau}\theta_{\xa}$. \qed
\end{proof}
\begin{remark} So far, we have stated all definitions and propositions for a $C^{\infty}$-field
$\Phi=\ua\oplus\ub$, defined on an open interval. This was just for convenience.
We will, from now on, use these definitions and propositions even when the $C^{\infty}$-requirement is not met, or when the field is defined on, say, a closed interval rather than an open interval. It will be clear in each case, that the respective definition or proposition still makes sense.
\end{remark}
\begin{definition}
\label{def:permut}
Set $S_3 = \{(1,2,3),(2,3,1),(3,1,2),(3,2,1),(1,3,2),(2,1,3)\}$, the set of all permutations of $(1,2,3)$.
\end{definition}
\begin{definition}\label{def:de}
For all $\sigma_{\ast}\in \{-1,+1\}^3$ let $\mathcal{D}(\sigma_{\ast})$ be the set of all $\Phi = \alpha \oplus \beta \in \R^3\oplus \R^3$ with 
$\big(\sgn \beta_1, \sgn \beta_2, \sgn \beta_3\big) = \sigma_{\ast}$.
For all $\tau_0,\tau_1\in \R$ with $\tau_0 < \tau_1$ let $\mathcal{E}(\sigma_{\ast};\tau_0,\tau_1)$ be the set of all continuous maps $\Phi: [\tau_0,\tau_1]\to \mathcal{D}(\sigma_{\ast})$.
\end{definition}
\begin{definition}\label{def:met1}
For all $\pi = (\mathbf{a},\mathbf{b},\mathbf{c}) \in S_3$ and $\mathbf{h}>0$ and $\sigma_{\ast} \in \{-1,+1\}^3$
define two functions $\mathcal{D}(\sigma_{\ast})\times \mathcal{D}(\sigma_{\ast})\to [0,\infty)$ by
\begin{alignat*}{6}
d_{\mathcal{D}(\sigma_{\ast}),(\pi,\mathbf{h})}(\Phi,\Psi) & = \max \Big\{
&  \big|A_{\mathbf{a}}[\Phi] & -A_{\mathbf{a}}[\Psi]\big|
& \;&,&\; \big|
\mathbf{h}\tfrac{\varphi_{\mathbf{a}}[\Phi]}{A_{\mathbf{a}}[\Phi]}
& - 
\mathbf{h}\tfrac{\varphi_{\mathbf{a}}[\Psi]}{A_{\mathbf{a}}[\Psi]}
\big|
&&,\\
&& \big|\alpha_{\mathbf{b},\mathbf{a}}[\Phi] & - \alpha_{\mathbf{b},\mathbf{a}}[\Psi]\big|
&\;&,&\; \big|\xi_{\mathbf{b},\mathbf{a}}[\Phi,\mathbf{h}] & - \xi_{\mathbf{b},\mathbf{a}}[\Psi,\mathbf{h}]\big| &&,\\
\rule{0pt}{15pt}
&& \big|\alpha_{\mathbf{c},\mathbf{a}}[\Phi] & - \alpha_{\mathbf{c},\mathbf{a}}[\Psi]\big|
&&,& \big|\xi_{\mathbf{c},\mathbf{a}}[\Phi,\mathbf{h}] & - \xi_{\mathbf{c},\mathbf{a}}[\Psi,\mathbf{h}]\big| &&\;\;\Big\}
\end{alignat*}
and
$$\slaa{d}_{\mathcal{D}(\sigma_{\ast}),\mathbf{h}} (\Phi,\Psi)
= \max_{\xa=1,2,3} \Big\{\big|\ua_{\xa}[\Phi]-\ua_{\xa}[\Psi]\big|,\, \big|\xi_{\xa}[\Phi,\mathbf{h}] - \xi_{\xa}[\Psi,\mathbf{h}]\big|\Big\}$$
Then $(\mathcal{D}(\sigma_{\ast}),d_{\mathcal{D}(\sigma_{\ast}),(\pi,\mathbf{h})})$
and
$(\mathcal{D}(\sigma_{\ast}),\slaa{d}_{\mathcal{D}(\sigma_{\ast}),\mathbf{h}})$ are metric spaces.
\end{definition}
\begin{definition}\label{def:met2}
For all $\pi  \in S_3$ and $\mathbf{h}>0$ and $\sigma_{\ast} \in \{-1,+1\}^3$
and $\tau_0,\tau_1\in \R$ with $\tau_0<\tau_1$ define a function $\mathcal{E}(\sigma_{\ast};\tau_0,\tau_1)
\times \mathcal{E}(\sigma_{\ast};\tau_0,\tau_1)\to [0,\infty)$ by
$$d_{\mathcal{E}(\sigma_{\ast};\tau_0,\tau_1),(\pi,\mathbf{h})}(\Phi,\Psi) = \textstyle\sup_{\tau \in [\tau_0,\tau_1]}d_{\mathcal{D}(\sigma_{\ast}),(\pi,\mathbf{h})}(\Phi(\tau),\Psi(\tau))$$
Then $(\mathcal{E}(\sigma_{\ast};\tau_0,\tau_1),d_{\mathcal{E}(\sigma_{\ast};\tau_0,\tau_1),(\pi,\mathbf{h})})$ is a metric space.
\end{definition}
\begin{lemma}\label{lem:fdhdkfhekjhekr}
Let $\pi = (\mathbf{a},\mathbf{b},\mathbf{c}) \in S_3$
and $\mathbf{h} > 0$ and $\sigma_{\ast} \in \{-1,+1\}^3$. Suppose $\mathbf{h} \leq 1$.
Let $C,D \geq 1$ be constants. Then, for all $\Phi, \Psi \in \mathcal{D}(\sigma_{\ast})$
such that
\begin{align*}
C^{-1}\leq A_{\mathbf{a}}[X] & \leq C &
D^{-1} \leq\mathbf{h}\,|\varphi_{\mathbf{a}}[X]| & \leq D
\end{align*}
for both $X = \Phi$ and $X=\Psi$
and such that $\sgn \varphi_{\mathbf{a}}[\Phi] = \sgn \varphi_{\mathbf{a}}[\Psi]$, we have:
\begin{itemize}
\item[(a)] $\slaa{d}_{\mathcal{D}}(\Phi,\Psi) \leq 2^3 C^2 D\,d_{\mathcal{D}}(\Phi,\Psi)$
\item[(b)] If $\exp(-\tfrac{1}{\mathbf{h}} C^{-2}D^{-1}) \leq 2^{-6} C^{-4} D^{-2}$, then $d_{\mathcal{D}}(\Phi,\Psi) \leq 2^5 C^3 D\,\slaa{d}_{\mathcal{D}}(\Phi,\Psi)$
\end{itemize}
Here, $d_{\mathcal{D}} = 
d_{\mathcal{D}(\sigma_{\ast}),(\pi,\mathbf{h})}$ and
$\slaa{d}_{\mathcal{D}} = \slaa{d}_{\mathcal{D}(\sigma_{\ast}),\mathbf{h}}$.
\end{lemma}
\begin{proof}
In this proof, 
$A$, $B$, $\alpha$, $\xi$ play  the roles of $A_{\mathbf{a}}$, $\mathbf{h} \varphi_{\mathbf{a}}/A_{\mathbf{a}}$, $\alpha_{\mathbf{a}}$, $\xi_{\mathbf{a}}$, respectively.\\
To show (a), let
$P:(0,\infty) \times \R \to \R^2$,\; $(A,B) \mapsto (\alpha(A,B),\xi(A,B))$, where
$$\alpha(A,B) = - A\tanh (\tfrac{1}{\mathbf{h}}AB) \qquad
\xi(A,B) = \mathbf{h} \log(\tfrac{1}{2} A \sech (\tfrac{1}{\mathbf{h}}AB))$$
This is a diffeomorphism. The Jacobian $J$ of $P$ is given by
\begin{align*}
J = \begin{pmatrix}
\tfrac{\p \alpha}{\p A} & \tfrac{\p \alpha}{\p B}\\
\rule{0pt}{12pt} \tfrac{\p \xi}{\p A} & \tfrac{\p \xi}{\p B}
\end{pmatrix}
= \begin{pmatrix}
-\tfrac{1}{\mathbf{h}} AB \sech^2 (\tfrac{1}{\mathbf{h}}AB) - \tanh (\tfrac{1}{\mathbf{h}}AB) &\hskip3pt&
 -\tfrac{1}{\mathbf{h}} A^2 \sech^2 (\tfrac{1}{\mathbf{h}}AB)\\
\rule{0pt}{12pt}  \tfrac{\mathbf{h}}{A} - B \tanh (\tfrac{1}{\mathbf{h}}AB)&&
 - A \tanh (\tfrac{1}{\mathbf{h}}AB)
 \end{pmatrix}
\end{align*}
Let $p_i=(A_i,B_i) \in (0,\infty)\times \R$ and set $(\alpha_i,\xi_i) = P(p_i)$, where $i=0,1$. Set $\gamma(t) =(A(t),B(t)) = (1-t)p_0 + t p_1$ where $t\in [0,1]$.
We have
$$\big(\begin{smallmatrix}
\alpha_1-\alpha_0\\
\xi_1-\xi_0
\end{smallmatrix}\big) = M
\big(\begin{smallmatrix}
A_1-A_0\\
B_1-B_0
\end{smallmatrix}\big)
\qquad\text{with}\qquad
M = \big(\begin{smallmatrix} M_{00} & M_{01}\\ M_{10} & M_{11}\end{smallmatrix}\big) = \textstyle\int_0^1 \dd t\, J(\gamma(t))$$
Suppose $C^{-1} \leq A_i \leq C$ and $(CD)^{-1} \leq |B_i| \leq CD$ and $\sgn B_0 = \sgn B_1$.
Then, 
 $C^{-1} \leq A(t) \leq C$ and $(CD)^{-1} \leq |B(t)| \leq CD$ for all $t\in [0,1]$. Observe that $|\varphi \sech^2 \varphi |\leq \tfrac{1}{2}$ for all $\varphi\in \R$. We have
$|M_{ij}| \leq 2C^2D$ for all $i,j \in \{0,1\}$. This implies (a).\\
We show that under the assumptions of (b), we have
$|\det M| \geq 2^{-3} C^{-1}$, and therefore $|(M^{-1})_{ij}| \leq 2^4C^3D$ for all $i,j\in \{0,1\}$. This would imply (b). We have $|\det M| \geq |M_{00}M_{11}| - |M_{01}M_{10}|$.
Set $\varphi(t) = \tfrac{1}{\mathbf{h}} A(t)B(t)$.
We have $|\varphi(t)| \geq \tfrac{1}{\mathbf{h}} C^{-2}D^{-1}$. By the assumption of (b), we have
$e^{-|\varphi(t)|} \leq 2^{-6}C^{-4}D^{-2}$, for all $t\in [0,1]$. We will also use the general inequalities
$0\leq 1 - \tanh |\varphi| \leq 2e^{-2|\varphi|}$ and $|\varphi \sech^2 \varphi| \leq 4|\varphi| e^{-2|\varphi|} \leq 4e^{-|\varphi|}$. We have $|-\varphi\sech^2 \varphi - \tanh \varphi|
\geq \tanh |\varphi| = 1 - (1-\tanh |\varphi|) \geq 2^{-1}$. The last inequality holds for all $t\in [0,1]$ and implies $|M_{00}| \geq 2^{-1}$, because $\varphi$ has constant sign. We have $|M_{11}|
\geq 2^{-1} C^{-1}$ and $|M_{10}| \leq 2CD$ and $|M_{01}|
\leq 2^{-4} C^{-2}D^{-1}$. This implies $|\det M| \geq 2^{-3}C^{-1}$. \qed
\end{proof}
\begin{definition}
Let $\mathcal{X} = \mathcal{D}(\sigma_{\ast})$ or $\mathcal{X} = \mathcal{E}(\sigma_{\ast};\tau_0,\tau_1)$. For all $\delta \geq 0$ and $\Phi\in \mathcal{X}$ and $\pi\in S_3$ and $\mathbf{h} > 0$, set
$\clball{\mathcal{X},(\pi,\mathbf{h})}{\delta}{\Phi} = \{\Psi \in \mathcal{X}\;|\; d_{\mathcal{X},(\pi,\mathbf{h})}(\Phi,\Psi) \leq \delta\}$.
\end{definition}
\begin{definition}[The reference field $\Phi_0$] \label{def:psizero}
For all $\pi = (\mathbf{a},\mathbf{b},\mathbf{c})\in S_3$,\, $\mathbf{f} = (\mathbf{h},w,q) \in (0,\infty)^3$,\, $\sigma_{\ast}\in \{-1,+1\}^3$
 let
$\Phi_0 = \Phi_0(\pi,\mathbf{f},\sigma_{\ast}): \R
\to  \mathcal{D}(\sigma_{\ast})$
be given by
\begin{subequations}\label{eq:kdshkhkhsk}
\begin{align}
A_{\mathbf{a}}[\Phi_0](\tau) & = 1\\
\theta_{\mathbf{a}}[\Phi_0,\mathbf{h}](\tau) & = 0\\
\alpha_{\mathbf{b},\mathbf{a}}[\Phi_0](\tau) & = -(1+w)^{-1}\\
\alpha_{\mathbf{c},\mathbf{a}}[\Phi_0](\tau) & =  -(1+w) \\ 
\xi_{\mathbf{b},\mathbf{a}}[\Phi_0,\mathbf{h}](\tau) & = -1 - \mathbf{h}\log 2 -  (1+w)^{-1}\,\tau\\
\xi_{\mathbf{c},\mathbf{a}}[\Phi_0,\mathbf{h}](\tau) & = -(1+w)q - \mathbf{h}\log 2 -  (1+w)\,\tau
\end{align}
\end{subequations}
(see Definition  \ref{def:khkhdkhkfdhfd}) for all $\tau \in \R$.
\end{definition}
\begin{remark} The field $\Phi_0$ is, up to renaming, given by equation (3.12) in \cite{BKL1}.
\end{remark}
\begin{lemma}
Let $\Phi_0$ be as in Definition \ref{def:psizero}. Then
 $(\boldsymbol{\mathfrak a},\boldsymbol{\mathfrak b},c)[\Phi_0,\mathbf{h},B_{\mathbf{a}}] = 0$ on $\R$.
\end{lemma}
\begin{proof}
Let  $\alpha = \alpha[\Phi_0]$,\; $\beta = \beta[\Phi_0]$,\; $\xi = \xi[\Phi_0,\mathbf{h}]$. We have 
$(\boldsymbol{\mathfrak a}_{\mathbf{a}},\boldsymbol{\mathfrak b}_{\mathbf{a}})[\Phi_0,\mathbf{h},B_{\mathbf{a}}] = 0$ by  
Lemma \ref{lem:lfdskhjkfdhkfdhkfd}. For $\mathbf{p}\in \{\mathbf{b},\mathbf{c}\}$, we have
$\boldsymbol{\mathfrak a}_{\mathbf{a}}[\Phi_0,\mathbf{h},B_{\mathbf{a}}]
+ 
\boldsymbol{\mathfrak a}_{\mathbf{p}}[\Phi_0,\mathbf{h},B_{\mathbf{a}}]= - \mathbf{h}\tfrac{\dd}{\dd \tau} \alpha_{\mathbf{a},\mathbf{p}} = 0$, that is $\boldsymbol{\mathfrak a}_{\mathbf{p}}[\Phi_0,\mathbf{h},B_{\mathbf{a}}] = 0$.
We also have 
$\beta_{\mathbf{a}}^{-1} \boldsymbol{\mathfrak b}_{\mathbf{a}}[\Phi_0,\mathbf{h},B_{\mathbf{a}}]
+
\beta_{\mathbf{p}}^{-1} 
\boldsymbol{\mathfrak b}_{\mathbf{p}}[\Phi_0,\mathbf{h},B_{\mathbf{a}}]
= - \tfrac{\dd}{\dd \tau}\xi_{\mathbf{a},\mathbf{p}} + \alpha_{\mathbf{a},\mathbf{p}} = 0$, that is
$\boldsymbol{\mathfrak b}_{\mathbf{p}}[\Phi_0,\mathbf{h},B_{\mathbf{a}}] = 0$.
Finally, 
$c[\Phi_0,\mathbf{h},B_{\mathbf{a}}] = -\alpha_{\mathbf{a}}^2 -\beta_{\mathbf{a}}^2 
+ \alpha_{\mathbf{a},\mathbf{b}}\alpha_{\mathbf{a},\mathbf{c}}
= - A_{\mathbf{a}}^2 + \alpha_{\mathbf{a},\mathbf{b}}\alpha_{\mathbf{a},\mathbf{c}} = 0$. Here,
$A_{\mathbf{a}} = A_{\mathbf{a}}[\Phi_0]$. \qed
\end{proof}
\begin{definition}\label{def:taump}
For all $\mathbf{f} = (\mathbf{h},w,q)\in (0,\infty)^3$ set
\begin{alignat*}{4}
\tau_-(\mathbf{f}) & = - \big(1 - \tfrac{1}{2+w}\big)\, \min\{1,q\} &\quad & < 0\\
\tau_+(\mathbf{f}) & = 1 + \tfrac{1}{w} && > 0
\end{alignat*}
\end{definition}
\begin{lemma}[Technical Lemma 1] \label{lem:tl1}
Let $\pi = (\mathbf{a},\mathbf{b},\mathbf{c}) \in S_3$,\, $\mathbf{f} = (\mathbf{h},w,q)\in  (0,\infty)^3$,\, $\sigma_{\ast}\in \{-1,+1\}^3$  and fix
$\delta > 0$,\,
 $\epsilon_- \in (0,-\tau_-)$,\, $\epsilon_+ \in (0,\tau_+)$ where
 $\tau_{\pm} = \tau_{\pm}(\mathbf{f})$. Set
\begin{subequations} \label{eq:fdlkkdhfdkjhfdkdf}
\begin{align}
\tau_{0-} & = \tau_- + \epsilon_- < 0 &
\Phi_0 & = \Phi_0(\pi,\mathbf{f},\sigma_{\ast})\big|_{
[\tau_{0-},\tau_{0+}]}\\
\tau_{0+} & = \tau_+ - \epsilon_+ > 0 &
\mathcal{E} & = \mathcal{E}(\sigma_{\ast};
\tau_{0-},\tau_{0+})
\end{align}
\end{subequations}
Then $\Phi_0 \in \mathcal{E}$. Furthermore, if the inequality
\begin{align}\label{eq:kdshskhskhkdsds}
\delta & \leq 2^{-4} \, \min \big\{1,\, w,\,\epsilon_-,\, \tfrac{\epsilon_+}{\tau_+\tau_{0+}}\big\}
\end{align}
holds, then for all $\Phi = \ua\oplus \ub \in \clball{\mathcal{E},(\pi,\mathbf{h})}{\delta}{\Phi_0}$ the estimates
\begin{align*}
\max \big\{|\ub_{\mathbf{b}}|^2,\,|\ub_{\mathbf{c}}|^2,\,|\ub_{\mathbf{b}}\ub_{\mathbf{a}}|,\,|\ub_{\mathbf{c}}\ub_{\mathbf{a}}|\big\}& \;\leq \; 2^4 \exp \big( - \tfrac{1}{4\mathbf{h}}\, \min \{1,\, \epsilon_-,\, \tfrac{\epsilon_+}{\tau_+}\}\big)\\
|A_\mathbf{a}[\Phi] - 1| & \;\leq\; 2^{-1}\\
|\varphi_{\mathbf{a}}[\Phi]| & \;\leq\;
\tfrac{1}{\mathbf{h}} 2(1 + |\tau|)\\
|\beta_{\mathbf{a}}| & \; \leq \; 2
\end{align*}
hold on $[\tau_{0-},\tau_{0+}]$. 
\end{lemma}
\begin{proof}
The following estimates hold for the components of $\Phi$, for all $\tau \in [\tau_{0-},\tau_{0+}]$:
\begin{align*}
|\ub_{\mathbf{b}}\ub_{\mathbf{a}}| & = 4 \exp \big( \tfrac{1}{\mathbf{h}} \xi_{\mathbf{b},\mathbf{a}}\big)\\
& \leq 4 \exp \big( \tfrac{1}{\mathbf{h}} \xi_{\mathbf{b},\mathbf{a}}[\Phi_0,\mathbf{h}]
+ \tfrac{1}{\mathbf{h}} \delta \big)\\
& \leq 2 \exp \big(
- \tfrac{1}{\mathbf{h}} - \tfrac{1}{\mathbf{h}} (1+w)^{-1}\tau 
+ \tfrac{1}{\mathbf{h}} \delta \big)\\
& \leq 2 \exp \big(
- \tfrac{1}{\mathbf{h}} - \tfrac{1}{\mathbf{h}} (1+w)^{-1}\tau_- 
+ \tfrac{1}{\mathbf{h}} \delta \big)\\
& \leq 2 \exp \big(
- \tfrac{1}{\mathbf{h}} + \tfrac{1}{\mathbf{h}} (2+w)^{-1}
+ \tfrac{1}{\mathbf{h}} \delta \big)\\
& \leq 2 \exp \big(
- \tfrac{1}{4\mathbf{h}}  \big) \displaybreak[0]\\
|\ub_{\mathbf{c}}\ub_{\mathbf{a}}| & = 4 \exp \big( \tfrac{1}{\mathbf{h}} \xi_{\mathbf{c},\mathbf{a}}\big)\\
& \leq 4 \exp \big( \tfrac{1}{\mathbf{h}} \xi_{\mathbf{c},\mathbf{a}}[\Phi_0,\mathbf{h}]
+ \tfrac{1}{\mathbf{h}} \delta \big)\\
& \leq 2 \exp \big(
- \tfrac{1}{\mathbf{h}}(1+w)q - \tfrac{1}{\mathbf{h}} (1+w) \tau 
+ \tfrac{1}{\mathbf{h}} \delta \big)\\
& \leq 2 \exp \big(
- \tfrac{1}{\mathbf{h}}(1+w)q - \tfrac{1}{\mathbf{h}} (1+w) (\tau_- + \epsilon_-)
+ \tfrac{1}{\mathbf{h}} \delta \big)\\
& \leq 2 \exp \big(
- \tfrac{1}{\mathbf{h}}(1+w)q + \tfrac{1}{\mathbf{h}} \tfrac{(1+w)^2}{2+w} q
- \tfrac{1}{\mathbf{h}} \epsilon_-
+ \tfrac{1}{\mathbf{h}} \delta \big)\\
& \leq 2 \exp \big( - \tfrac{1}{2\mathbf{h}} \epsilon_- \big) \displaybreak[0]\\
|\varphi_{\mathbf{a}}| & = \tfrac{1}{\mathbf{h}} A_{\mathbf{a}} |\tau - \theta_{\mathbf{a}}|\\
& \leq \tfrac{1}{\mathbf{h}} (1 + \delta)\big(|\tau|  + \delta \big)\\
& \leq \tfrac{1}{\mathbf{h}} \big(|\tau|  + \delta|\tau| + 2\delta \big)\\
& \leq \tfrac{1}{\mathbf{h}} 2 (1 + |\tau|) \displaybreak[0]\\
|\ub_{\mathbf{a}}|^{-1} & = |A_{\mathbf{a}}|^{-1} \cosh \varphi_{\mathbf{a}}\\
& \leq 2 \exp \big( |\varphi_{\mathbf{a}}|\big)\\
& \leq 2 \exp \big(
\tfrac{1}{\mathbf{h}} |\tau|
+\tfrac{1}{\mathbf{h}} \delta|\tau|
+\tfrac{1}{\mathbf{h}} 2\delta
\big) \displaybreak[0]\\
|\ub_{\mathbf{b}}| & = 
|\ub_{\mathbf{b}}\ub_{\mathbf{a}}|\cdot |\ub_{\mathbf{a}}|^{-1} \\
& \leq 4 \exp \big(
- \tfrac{1}{\mathbf{h}} - \tfrac{1}{\mathbf{h}} (1+w)^{-1}\tau 
+ \tfrac{1}{\mathbf{h}} |\tau|
+\tfrac{1}{\mathbf{h}} \delta|\tau|
+\tfrac{1}{\mathbf{h}} 3\delta
\big) \\
 & \leq 4 \exp \big( \tfrac{1}{\mathbf{h}} \max \big\{ -1 -\tfrac{2+w}{1+w} \tau_{0-}
 - \delta \tau_{0-},\, -1 + \tfrac{w}{1+w}\tau_{0+} + \delta \tau_{0+} \big\} +\tfrac{1}{\mathbf{h}} 3\delta \big) \\
  & \leq 4 \exp \big( \tfrac{1}{\mathbf{h}} \max\{-\epsilon_- - \delta \tau_{0-}, -\tfrac{\epsilon_+}{\tau_+} + \delta \tau_{0+}\}
  +\tfrac{1}{\mathbf{h}} 3 \delta \big) \\
  & \leq 4 \exp \big(-\tfrac{1}{\mathbf{h}} \tfrac{15}{16}  \min\{\epsilon_-,\tfrac{\epsilon_+}{\tau_+}\} + \tfrac{1}{\mathbf{h}}3\delta \big)\\
& \leq 4 \exp \big( -\tfrac{1}{2\mathbf{h}} \min\{\epsilon_-, \tfrac{\epsilon_+}{\tau_+}\} \big)
\intertext{The last step uses $\delta \leq 2^{-3} \tfrac{\epsilon_+}{\tau_+}$.
In the case $\epsilon_+ \geq \tfrac{1}{2}\tau_+$, this follows from $\delta \leq 2^{-4}$. If $\epsilon_+
\leq \tfrac{1}{2}\tau_+$, then this follows from $\delta \leq 2^{-4}\tfrac{\epsilon_+}{\tau_+\tau_{0+}}$, because
$\tau_{0+} = \tau_+ - \epsilon_+ \geq \tfrac{1}{2} \tau_+ \geq \tfrac{1}{2}$. } 
|\ub_{\mathbf{c}}| & = 
|\ub_{\mathbf{c}}\ub_{\mathbf{a}}|\cdot |\ub_{\mathbf{a}}|^{-1} \\
& \leq 4 \exp \big(
- \tfrac{1}{\mathbf{h}}(1+w)q - \tfrac{1}{\mathbf{h}} (1+w) \tau 
+ \tfrac{1}{\mathbf{h}} |\tau|
+\tfrac{1}{\mathbf{h}} \delta|\tau|
+\tfrac{1}{\mathbf{h}} 3\delta
\big) \\
& \leq 4 \exp \big(
- \tfrac{1}{\mathbf{h}}(1+w)q
\\
& \hskip 30mm + \tfrac{1}{\mathbf{h}}
\max\big\{
-  (2+w + \delta) \tau_{0-},
-  (w-\delta) \tau_{0+}
\big\}
+\tfrac{1}{\mathbf{h}} 3\delta
\big)\\
& \leq 4 \exp \big(
- \tfrac{1}{\mathbf{h}}(1+w)q
+ \tfrac{1}{\mathbf{h}}
 (2+w) |\tau_-|
 - \tfrac{1}{\mathbf{h}}
 (2+w + \delta) \epsilon_-
+\tfrac{1}{\mathbf{h}} 4\delta
\big) \displaybreak[0] \\
& \leq 4 \exp \big(
 - \tfrac{1}{\mathbf{h}}
 2 \epsilon_-
+\tfrac{1}{\mathbf{h}} 4\delta
\big) \\
& \leq 4 \exp \big(
 - \tfrac{1}{\mathbf{h}}
 \epsilon_-
\big) 
\end{align*}
This concludes the proof. \qed
\end{proof}
\begin{lemma} \label{lem:expl}
Recall Definitions \ref{def:dkhkdhfd} and \ref{def:nnnnnns}. We have
\begin{align*}
\boldsymbol{\mathfrak a}_{\mathbf{a}} [\Phi,\mathbf{h},Z,B_{\mathbf{a}}] & =
+ \ub_{\mathbf{b}}^2 + \ub_{\mathbf{c}}^2 - 2\ub_{\mathbf{b}}\ub_{\mathbf{c}}\\
\boldsymbol{\mathfrak a}_{\mathbf{b}} [\Phi,\mathbf{h},Z,B_{\mathbf{a}}] & =
- \ub_{\mathbf{b}}^2 + \ub_{\mathbf{c}}^2 - 2\ub_{\mathbf{a}}\ub_{\mathbf{c}}\\
\boldsymbol{\mathfrak a}_{\mathbf{c}} [\Phi,\mathbf{h},Z,B_{\mathbf{a}}] & =
+ \ub_{\mathbf{b}}^2 - \ub_{\mathbf{c}}^2 - 2\ub_{\mathbf{a}}\ub_{\mathbf{b}}
\end{align*}
for all $(\mathbf{a},\mathbf{b},\mathbf{c})\in S_3$.
\end{lemma}
\begin{remark}
Lemma \ref{lem:expl}
displays the differences between the equations $\boldsymbol{\mathfrak a} [\Phi,\mathbf{h},Z] = 0$
and $\boldsymbol{\mathfrak a} [\Phi,\mathbf{h},B_{\mathbf{a}}] = 0$.  Lemma \ref{lem:tl1} gives bounds for the terms that appear in these differences. Informally, they tend exponentially to zero as  $\mathbf{h} \downarrow 0$.
This quantifies a basic guiding intuition of \cite{BKL1}.
\end{remark}
\begin{definition}
For  all vectors $\Phi = \ua\oplus \ub \in \R^3\oplus \R^3$ with $\ub_1,\ub_2,\ub_3\neq 0$, all  $\pi = (\mathbf{a},\mathbf{b},\mathbf{c})\in S_3$ and all $\mathbf{h} > 0$, define four real numbers by
\begin{align*}
\mathbf{I}_1[\Phi,\mathbf{h},\pi] & = 
-\tfrac{1}{\mathbf{h}}\,
\boldsymbol{\mathfrak a}_{\mathbf{a}}[\Phi,\mathbf{h},Z,B_{\mathbf{a}}]\,
\tanh \varphi_{\mathbf{a}}[\Phi] \\
\mathbf{I}_2[\Phi,\mathbf{h},\pi] & = 
\big(\hskip1pt A_{\mathbf{a}}[\Phi] \hskip1pt\big)^{-2}
\,\boldsymbol{\mathfrak a}_{\mathbf{a}}[\Phi,\mathbf{h},Z,B_{\mathbf{a}}]\,
\Big(1-\varphi_{\mathbf{a}}[\Phi] \tanh \varphi_{\mathbf{a}}[\Phi]\Big)\\
\mathbf{I}_{(3,\mathbf{p})}[\Phi,\mathbf{h},\pi] & =
\tfrac{1}{\mathbf{h}} \boldsymbol{\mathfrak a}_{\mathbf{p}} [\Phi,\mathbf{h},Z,B_{\mathbf{a}}] + 
\tfrac{1}{\mathbf{h}} \boldsymbol{\mathfrak a}_{\mathbf{a}} [\Phi,\mathbf{h},Z,B_{\mathbf{a}}]
\end{align*}
where $\mathbf{p} \in\{\mathbf{b}, \mathbf{c}\}$. 
If $\Phi$ is not an element of $\R^3\oplus \R^3$, but rather a function with values in $\R^3\oplus \R^3$, with $\ub_1,\ub_2,\ub_3\neq 0$ everywhere, then $\mathbf{I}_1$, $\mathbf{I}_2$, $\mathbf{I}_{(3,\mathbf{b})}$, $\mathbf{I}_{(3,\mathbf{c})}$ are functions, too.
\end{definition}
\begin{lemma}[Technical Lemma 2] \label{lem:tl2}
In the context of Lemma \ref{lem:tl1}, if $\delta > 0$ satisfies \eqref{eq:kdshskhskhkdsds}, then, for all
$\Phi,\Psi \in \clball{\mathcal{E},(\pi,\mathbf{h})}{\delta}{\Phi_0}$
and all $S \in \{1,2,(3,\mathbf{b}),(3,\mathbf{c})\}$, the estimates 
\begin{subequations}
\begin{align}\label{eq:lkskhdkhd1}
\big|\mathbf{I}_{S}[\Phi]\big| & \;\leq\;
2^{11} \max \{1,\tfrac{1}{\mathbf{h}},\tfrac{1}{\mathbf{h}}|\tau| \}\, \exp\big(-\tfrac{1}{4\mathbf{h}} \min \{1,\epsilon_-,\tfrac{\epsilon_+}{\tau_+}\}\big) \\
\label{eq:lkskhdkhd2}
\big|\mathbf{I}_{S}[\Phi] - \mathbf{I}_{S}[\Psi]\big| & \;\leq 
\;2^{17} \big(\max \{1,\tfrac{1}{\mathbf{h}},\tfrac{1}{\mathbf{h}}|\tau| \}\big)^2 \exp\big(-\tfrac{1}{4\mathbf{h}} \min \{1,\epsilon_-,\tfrac{\epsilon_+}{\tau_+}\}\big) d_{\mathcal{E}}(\Phi,\Psi)
\end{align}
\end{subequations}
hold on $[\tau_{0-},\tau_{0+}]$. Here, $\mathbf{I}_{S}[\Phi] = \mathbf{I}_{S}[\Phi,\mathbf{h},\pi]$,\, $\mathbf{I}_{S}[\Psi] = \mathbf{I}_{S}[\Psi,\mathbf{h},\pi]$ and
$d_{\mathcal{E}} = d_{\mathcal{E},(\pi,\mathbf{h})}$.
\end{lemma}
\begin{proof}
In this proof, we simplify the notation by suppressing $\mathbf{h}>0$ and abbreviating $$M =
 \exp(-\tfrac{1}{4\mathbf{h}} \min \{1,\epsilon_-,\tfrac{\epsilon_+}{\tau_+}\})
 \quad M_1 = \max \{1,\tfrac{1}{\mathbf{h}},\tfrac{1}{\mathbf{h}}|\tau| \}$$
 Lemmas  \ref{lem:tl1}, \ref{lem:expl} imply
$\big|\boldsymbol{\mathfrak a}_{\xa} [\Phi,\mathbf{h},Z,B_{\mathbf{a}}]\big|
\leq  2^6 M$, $\xa=1,2,3$, and
$|\varphi_{\mathbf{a}}[\Phi]| \leq 2^2M_1$ and $(A_{\mathbf{a}}[\Phi])^{-2} \leq 2^2$. This implies \eqref{eq:lkskhdkhd1}. To show \eqref{eq:lkskhdkhd2}, observe that (here $\mathbf{p},\mathbf{q}\in \{\mathbf{b},\mathbf{c}\}$)
\begin{align*}
\begin{split}
\big|\varphi_{\mathbf{a}}[\Phi] - \varphi_{\mathbf{a}}[\Psi]\big|
& \leq \tfrac{1}{\mathbf{h}} 
|A_{\mathbf{a}}[\Phi] - A_{\mathbf{a}}[\Psi]|\,|\tau|
+ \tfrac{1}{\mathbf{h}}|
A_{\mathbf{a}}[\Phi]\theta_{\mathbf{a}}[\Phi] - A_{\mathbf{a}}[\Psi]
\theta_{\mathbf{a}}[\Psi]|\\
& \leq 
\tfrac{1}{\mathbf{h}} (1 + |\tau|)
|A_{\mathbf{a}}[\Phi] - A_{\mathbf{a}}[\Psi]|
+ \tfrac{1}{\mathbf{h}} 2 |\theta_{\mathbf{a}}[\Phi]- \theta_{\mathbf{a}}[\Psi]|\\
& \leq 2^2 M_1\, d_{\mathcal{E}}(\Phi,\Psi)
\end{split}\\
\begin{split}
\big|
\xi_{\mathbf{a}}[\Phi]
- \xi_{\mathbf{a}}[\Psi]
\big| & \leq \mathbf{h}|\log A_{\mathbf{a}}[\Phi] - \log A_{\mathbf{a}}[\Psi]|\\
& \hskip 20mm
+ \mathbf{h}\,|\log \cosh \varphi_{\mathbf{a}}[\Phi] - \log \cosh \varphi_{\mathbf{a}}[\Psi]|\\
& \leq 2^3 \mathbf{h}\, M_1\, d_{\mathcal{E}}(\Phi,\Psi)
\end{split}\\
\begin{split}
\big|\ub_{\mathbf{p}}[\Phi]\ub_{\mathbf{a}}[\Phi] - \ub_{\mathbf{p}}[\Psi]\ub_{\mathbf{a}}[\Psi]\big| & \leq
\tfrac{1}{\mathbf{h}}\, \max \big\{
|\ub_{\mathbf{p}}[\Phi]\ub_{\mathbf{a}}[\Phi]|,
|\ub_{\mathbf{p}}[\Psi]\ub_{\mathbf{a}}[\Psi]| \big\}\\
& \hskip 50mm \times\big|\xi_{\mathbf{a},\mathbf{p}}[\Phi]-\xi_{\mathbf{a},\mathbf{p}}[\Psi]\big|\\
& \leq 2^4 M_1 M\, d_{\mathcal{E}}(\Phi,\Psi)
\end{split}\\
\begin{split}
\big|\ub_{\mathbf{p}}[\Phi] - \ub_{\mathbf{p}}[\Psi]\big| & \leq
\tfrac{1}{\mathbf{h}}\, \max \big\{
|\ub_{\mathbf{p}}[\Phi]|,
|\ub_{\mathbf{p}}[\Psi]| \big\}  \big|\xi_{\mathbf{p}}[\Phi]-\xi_{\mathbf{p}}[\Psi]\big|\\
& \leq \tfrac{1}{\mathbf{h}} 2^2 M^{1/2}\, \big(\,
\big|\xi_{\mathbf{a},\mathbf{p}}[\Phi]-\xi_{\mathbf{a},\mathbf{p}}[\Psi]\big|
+ \big|\xi_{\mathbf{a}}[\Phi]-\xi_{\mathbf{a}}[\Psi]\big|\,\big)\\
& \leq  2^6 M_1 M^{1/2}\,d_{\mathcal{E}}(\Phi,\Psi) 
\end{split}\\
\begin{split}
\big|\ub_{\mathbf{p}}[\Phi]\ub_{\mathbf{q}}[\Phi] - \ub_{\mathbf{p}}[\Psi]\ub_{\mathbf{q}}[\Psi]\big| & \leq
2^9 M_1M\,d_{\mathcal{E}}(\Phi,\Psi)
\end{split}
\end{align*}
Consequently, for $\xa=1,2,3$,
$$\big|\boldsymbol{\mathfrak a}_{\xa} [\Phi,\mathbf{h},Z,B_{\mathbf{a}}]
- \boldsymbol{\mathfrak a}_{\xa} [\Psi,\mathbf{h},Z,B_{\mathbf{a}}]\big|\;\leq\;
2^{11}M_1M\,d_{\mathcal{E}}(\Phi,\Psi)
$$
With these estimates, \eqref{eq:lkskhdkhd2} follows.
Observe that $\R\to \R, x \mapsto x \tanh x$ is Lipschitz with
Lipschitz-constant $L > 0$ determined by $L\tanh L = 1$, in particular $L<2$.
\qed
\end{proof}
\begin{definition} \label{def:kdhkhskhkdssPT1111star}
For all $\pi = (\mathbf{a},\mathbf{b},\mathbf{c})\in S_3$ and  $\mathbf{f}=(\mathbf{h},w,q) \in (0,\infty)^2 \times \R$
(we don't require $q > 0$ here) and $\sigma_{\ast}\in \{-1,+1\}^3$, 
let $\Phi_{\star} = \Phi_{\star}(\pi,\mathbf{f},\sigma_{\ast}) \in \mathcal{D}(\sigma_{\ast})$ be given by
\begin{align*}
\ua_{\mathbf{a}}[\Phi_{\star}] & = -1 &\qquad
\xi_{\mathbf{a}}[\Phi_{\star},\mathbf{h}] & =  - \tfrac{1+w}{1+2w} (1 + \mathbf{h}\log 2)\\
\ua_{\mathbf{b}}[\Phi_{\star}] & = \tfrac{w}{1+w} &
\xi_{\mathbf{b}}[\Phi_{\star},\mathbf{h}] & =  - \tfrac{1+w}{1+2w} (1 + \mathbf{h}\log 2)\\
\ua_{\mathbf{c}}[\Phi_{\star}] & = -w - \mu &
\xi_{\mathbf{c}}[\Phi_{\star},\mathbf{h}] & =  -(1+w)q
- \tfrac{w(1+w)}{1+2w} -  \tfrac{1+3w+w^2}{1+2w} \mathbf{h}\log 2
\end{align*}
and
\begin{equation}\label{dshdkhkdhd}
\mu = (1+w)
\big(
\ub_{1}^2
+ \ub_{2}^2
+ \ub_{3}^2
-2\ub_2\ub_3
-2\ub_3\ub_1
-2\ub_1\ub_2
\big)|_{\beta = \beta[\Phi_{\star}]}
\end{equation}
\end{definition}
\begin{definition}\label{def:sur}
For all $\pi = (\mathbf{a},\mathbf{b},\mathbf{c})\in S_3$,\;$\sigma_{\ast}\in \{-1,+1\}^3$
let $\mathcal{H}(\pi,\sigma_{\ast}) \subset \mathcal{D}(\sigma_{\ast})$ be the set of all vectors 
$\Phi = \alpha \oplus \beta \in \mathcal{D}(\sigma_{\ast})$ with
\begin{subequations} \label{eq:kdshkfdhfdkjhfd23}
\begin{align}
\label{eq:kdshkfdhfdkjhfd} |\beta_{\mathbf{a}}| & = |\beta_{\mathbf{b}}| & 
\textstyle\sum_{(\xa,\xb,\xc)\in \mathcal{C}}\big(
\alpha_{\xb}\alpha_{\xc} - (\beta_{\xa})^2 + 2\beta_{\xb}\beta_{\xc}
\big) & =0\\
\label{eq:kdshkfdhfdkjhfd2} 0 < \alpha_{\mathbf{b}} & < -\alpha_{\mathbf{a}}
 & 
\big(\alpha_{\mathbf{b}} + |\alpha_{\mathbf{a}}|\big)\, \log |\beta_{\mathbf{a}}/\alpha_{\mathbf{a}}| & < \alpha_{\mathbf{b}} \log 2
\end{align}
\end{subequations}
\end{definition}
\begin{lemma}\label{lem:fdjlkjdlkjfdl}
Let $\pi = (\mathbf{a},\mathbf{b},\mathbf{c})\in S_3$ and $\sigma_{\ast}\in \{-1,+1\}^3$.
The set $\mathcal{H}(\pi,\sigma_{\ast}) \subset \mathcal{D}(\sigma_{\ast})$ is a smooth 4-dimensional submanifold. The map
\begin{equation}\label{eq:ksdhkjfdhf}
\begin{split}
(0,\infty)^3 \times \R & \to \mathcal{H}(\pi,\sigma_{\ast})\\
(\lambda,\mathbf{h},w,q) & \mapsto \lambda\,\Phi_{\star}(\pi,(\mathbf{h},w,q),\sigma_{\ast})
\end{split}
\end{equation}
is a diffeomorphism. Its inverse is given by
\begin{subequations}\label{eq:dkhkhfdfd}
\begin{align}
\label{eq:dkhkhfdfd1} 
w & = -\alpha_{\mathbf{b}}/(\alpha_{\mathbf{a}} + \alpha_{\mathbf{b}}) &
\tfrac{1}{\mathbf{h}} & =
- \tfrac{1+2w}{1+w}\log |\beta_{\mathbf{a}}/\alpha_{\mathbf{a}}| + \tfrac{w}{1+w}\log 2\\
\label{eq:dkhkhfdfd2} \lambda & = - \alpha_{\mathbf{a}} &
q & =
-\tfrac{1}{1+w}\,\mathbf{h}\log |\beta_{\mathbf{c}}/\alpha_{\mathbf{a}}| - \tfrac{w}{1+2w}\big(1+ \mathbf{h}\log 2\big)
\end{align}
\end{subequations}
\end{lemma}
\begin{proof}
$\mathcal{H}(\pi,\sigma_{\ast})$ is the graph of a smooth map from an open subset of $\R^4$ to $\R^2$. Namely the map given by solving \eqref{eq:kdshkfdhfdkjhfd} for $(\alpha_{\mathbf{c}},\beta_{\mathbf{b}})$ in terms of $(\alpha_{\mathbf{a}},\alpha_{\mathbf{b}},\beta_{\mathbf{a}},\beta_{\mathbf{c}})$, whose domain 
 is given by \eqref{eq:kdshkfdhfdkjhfd2} and appropriate sign conditions inherited from $\mathcal{D}(\sigma_{\ast})$.
The map \eqref{eq:ksdhkjfdhf} is well-defined, i.e.
$\lambda\,\Phi_{\star}(\pi,(\mathbf{h},w,q),\sigma_{\ast}) \in \mathcal{H}(\pi,\sigma_{\ast})$.
The map \eqref{eq:dkhkhfdfd} is well-defined, because the two right hand sides in
\eqref{eq:dkhkhfdfd1} and the first right hand side in \eqref{eq:dkhkhfdfd2}
are positive, by \eqref{eq:kdshkfdhfdkjhfd2}. By direct calculation, the two maps are inverses. \qed
\end{proof}
\begin{definition} \label{def:kdhkhskhkdssPT1tau}
For all $\mathbf{f} = (\mathbf{h},w,q) \in (0,\infty)^3$ set
\begin{alignat*}{4}
\tau_{1-}(\mathbf{f}) & = \begin{cases}
- \tfrac{1+w}{3+w}\,q - \tfrac{1}{3+w}\,\mathbf{h}\log 2 & \text{if $q \leq 1$}\\
-\tfrac{1+w}{3+2w} - \tfrac{1+w}{3+2w} \,\mathbf{h}\log 2 & \text{if $q > 1$}
\end{cases} &\hskip 12mm& < 0\\
\tau_{1+}(\mathbf{f}) & = (1 + \mathbf{h}\log 2) \tfrac{1+w}{1+2w} && > 0
\end{alignat*}
\end{definition}
\begin{definition} 
For all $\pi = (\mathbf{a},\mathbf{b},\mathbf{c})\in S_3$ and 
$\mathbf{f} = (\mathbf{h},w,q) \in (0,\infty)^3$ and $\sigma_{\ast}\in \{-1,+1\}^3$
let $\Phi_1 = \Phi_1(\pi,\mathbf{f},\sigma_{\ast}):\; \R \to \mathcal{D}(\sigma_{\ast})$ be given by
\begin{align*}
A_{\mathbf{a}}[\Phi_1](\tau) & = A_{\mathbf{a}}[\Phi_{\star}] &\quad
\alpha_{\mathbf{p},\mathbf{a}}[\Phi_1](\tau) & = \alpha_{\mathbf{p},\mathbf{a}}[\Phi_{\star}]\\
\theta_{\mathbf{a}}[\Phi_1,\mathbf{h}](\tau) & = \theta_{\mathbf{a}}[\Phi_{\star},\mathbf{h}] &
\xi_{\mathbf{p},\mathbf{a}}[\Phi_1,\mathbf{h}](\tau) & = \xi_{\mathbf{p},\mathbf{a}}[\Phi_{\star},\mathbf{h}] + (\tau - \tau_{1+})\alpha_{\mathbf{p},\mathbf{a}}[\Phi_{\star}]
\end{align*}
for all $\tau \in \R$ and
$\mathbf{p}\in \{\mathbf{b},\mathbf{c}\}$. Here,
$\tau_{1+} = \tau_{1+}(\mathbf{f})$ and $\Phi_{\star} = \Phi_{\star}(\pi,\mathbf{f},\sigma_{\ast})$.
\end{definition}
\begin{lemma} \label{lem:kdshkdhkdhkd} For all
$\pi = (\mathbf{a},\mathbf{b},\mathbf{c})\in S_3$,\,
 $\mathbf{f} = (\mathbf{h},w,q) \in (0,\infty)^3$,\, $\sigma_{\ast}\in \{-1,+1\}^3$,
set 
$\Phi_0 = \Phi_0(\pi,\mathbf{f},\sigma_{\ast})$ and
$\Phi_1 = \Phi_1(\pi,\mathbf{f},\sigma_{\ast})$ and $\tau_{1+} = \tau_{1+}(\mathbf{f})$ and $d_{\mathcal{D}} = d_{\mathcal{D}(\sigma_{\ast}),(\pi,\mathbf{h})}$
and
$\slaa{d}_{\mathcal{D}} = \slaa{d}_{\mathcal{D}(\sigma_{\ast}),\mathbf{h}}$. Then
\begin{itemize}
\item[(a)] $|\ub_{\mathbf{a}}[\Phi_1](\tau_{1+})| = |\ub_{\mathbf{b}}[\Phi_1](\tau_{1+})|$
\item[(b)] $c[\Phi_1,\mathbf{h},Z](\tau_{1+}) = 0$, see Definitions \ref{def:dkhkdhfd} and \ref{def:nnnnnns} for $c$ and $Z$, respectively
\item[(c)] $\slaa{d}_{\mathcal{D}}(\Phi_0(\tau_{1+}),\Phi_1(\tau_{1+}))\,
\leq \, 2^7 \max \{1+w,\mathbf{h}\}
 \exp(-\tfrac{1}{2\mathbf{h}}\min\{1,w+q\})$ 
\item[(d)] $d_{\mathcal{D}}(\Phi_0(\tau),\Phi_1(\tau))
\leq \big(1 + |\tau-\tau_{1+}|\big)\, d_{\mathcal{D}}(\Phi_0(\tau_{1+}),\Phi_1(\tau_{1+}))$ 
for all $\tau \in \R$
\end{itemize}
\end{lemma}
\begin{proof}
We discuss (c) only. By direct calculation,
\begin{align*}
\alpha_{\mathbf{a}}[\Phi_0](\tau_{1+})- \alpha_{\mathbf{a}}[\Phi_1](\tau_{1+}) & = - X & 
\xi_{\mathbf{a}}[\Phi_0,\mathbf{h}](\tau_{1+}) - \xi_{\mathbf{a}}[\Phi_1,\mathbf{h}](\tau_{1+})  & =  - Y\\
\alpha_{\mathbf{b}}[\Phi_0](\tau_{1+})- \alpha_{\mathbf{b}}[\Phi_1](\tau_{1+}) & = + X &
\xi_{\mathbf{b}}[\Phi_0,\mathbf{h}](\tau_{1+}) - \xi_{\mathbf{b}}[\Phi_1,\mathbf{h}](\tau_{1+}) & = 
+ Y\\
\alpha_{\mathbf{c}}[\Phi_0](\tau_{1+})- \alpha_{\mathbf{c}}[\Phi_1](\tau_{1+})& =  + X + \mu & 
\xi_{\mathbf{c}}[\Phi_0,\mathbf{h}](\tau_{1+}) - \xi_{\mathbf{c}}[\Phi_1,\mathbf{h}](\tau_{1+}) & =  + Y
\end{align*}
with
$X = -1 + \tanh \big(\tfrac{1}{\mathbf{h}} \tau_{1+}\big)$
and $Y = \mathbf{h} \log\big( 1 + \exp(-2\tfrac{1}{\mathbf{h}} \tau_{1+} ) \big)$. The estimates 
\begin{align*}
|X| & \leq 2  \exp(-2\tfrac{1}{\mathbf{h}} \tau_{1+} )
\leq 2\exp(-\tfrac{1}{\mathbf{h}} )\\
|Y| & \leq \mathbf{h}  \exp(-2\tfrac{1}{\mathbf{h}} \tau_{1+} )
\leq \mathbf{h} \exp(-\tfrac{1}{\mathbf{h}} )\\
|\mu| & \leq (1+w)2^6 \exp(-\tfrac{1}{2\mathbf{h}} \min\{1,w+q\})
\end{align*}
imply (c). \qed
\end{proof}
\begin{definition} \label{def:kdhkhskhkdssPT2}
This is, verbatim, Definition \ref{def:kdhkhskhkdssPT2xyz} in the Introduction.
\end{definition}
\begin{lemma} 
In the context of Definition \ref{def:kdhkhskhkdssPT2}, the identities
\begin{subequations}\label{eq:lkdkjdhkjdhkdfhkdf}
\begin{align}
\label{eq:lkdkjdhkjdhkdfhkdf1}
\lambda_L & = 1 - \alpha_{\mathbf{a},\mathbf{a}'}[\Phi_0](\tau_{1-})
= 1 - \alpha_{\mathbf{a},\mathbf{a}'}[\Phi_0](\tau)\\
\label{eq:lkdkjdhkjdhkdfhkdf2}
w_L & = - \big(\alpha_{\mathbf{a},\mathbf{a}'}[\Phi_0](\tau_{1-})\big)^{-1}
= - \big(\alpha_{\mathbf{a},\mathbf{a}'}[\Phi_0](\tau)\big)^{-1}\\
\label{eq:lkdkjdhkjdhkdfhkdf3}
\tfrac{\mathbf{h}}{\mathbf{h}_L}
& = 
\tfrac{1+2w_L}{1+w_L}\big(
- \tau_{1-}  + \mathbf{h}\log \lambda_L\big) - \mathbf{h}\,\log 2\\
\label{eq:lkdkjdhkjdhkdfhkdf4}
\begin{split}
q_L & = \tfrac{1}{1+w_L}\Big(
\mathbf{h}_L \log \lambda_L -
\tfrac{\mathbf{h}_L}{\mathbf{h}}
\xi_{\mathbf{a},\mathbf{c}'}[\Phi_0,\mathbf{h}](\tau_{1-})
+
\tfrac{\mathbf{h}_L}{\mathbf{h}}
\tau_{1-}
- \tfrac{w_L(1+w_L)}{1+2w_L}\\
& \hskip60mm 
- \tfrac{1+3w_L+(w_L)^2}{1+2w_L}\,\mathbf{h}_L\log 2
\Big)
\end{split}
\end{align}
\end{subequations}
hold, where $\Phi_0 = \Phi_0(\pi,\mathbf{f},\sigma_{\ast})$
and $\tau_{1-} = \tau_{1-}(\mathbf{f})$ and $\tau \in \R$.
Furthermore, 
\begin{align}\label{eq:ddhkdfddfdf}
\big(\xi_{\mathbf{a}}[\Phi_0,\mathbf{h}] - \xi_{\mathbf{a}'}[\Phi_0,\mathbf{h}]\big)\, F
& = \tau - \tau_{1-} - 2\mathbf{h}\log\big(1 + e^{2\tau/\mathbf{h}}\big)\, F
\end{align}
for all $\tau \in \R$, where
\begin{align} \label{eq:ddhkdfddfdf1}
F & = \begin{cases}
\tfrac{1}{3+w} 
 & q \leq 1\\
\tfrac{1+w}{3+2w}
& q > 1
\end{cases}
\end{align}
\end{lemma}
\begin{proof}
By direct calculation. In each case, distinguish $q \leq 1$ and $q > 1$. \qed
\end{proof}
\begin{definition}\label{cccc} For all $\mathbf{f} = (\mathbf{h},w,q) \in (0,\infty)^3$  set
$$ \tau_{\ast}(\mathbf{f}) = \begin{cases} \tfrac{q}{1+w} & \text{if $q \leq 1$}\\
1 & \text{if $q > 1$}\end{cases}$$
\end{definition}
\begin{definition}\label{ccccc}
For all $\mathbf{f} = (\mathbf{h},w,q) \in (0,\infty)^3$  set
\begin{equation}\label{eq:ldsjdslsj1}
\mathbf{K}(\mathbf{f}) = 2^{40} \big(\tfrac{1}{\mathbf{h}}\big)^2 \max\{(\tfrac{1}{w})^2,w^3\}\, \max\{(\tfrac{1}{q})^2,q\}\,
\exp\big(-\tfrac{1}{\mathbf{h}} 2^{-7}\tau_{\ast}(\mathbf{f})\big)
\end{equation}
\end{definition}
\begin{definition}\label{cccccc}
Let $\mathcal{F}$ be the open set of all $\mathbf{f} = (\mathbf{h},w,q) \in (0,\infty)^3$ for which
\begin{align} \label{eq:ldsjdslsj2}
q &\neq 1 &
\mathbf{K}(\mathbf{f}) & < 1 &
\mathbf{h} & < 2^{-7}\tau_{\ast}(\mathbf{f}) 
\end{align}
\end{definition}
\begin{proposition} \label{prop:skhdkjhfd}
For all $\pi = (\mathbf{a},\mathbf{b},\mathbf{c})\in S_3$,\, $\sigma_{\ast} \in \{-1,+1\}^3$,
there are {\bf unique} maps
\begin{align*}
\Pi = \Pi[\pi,\sigma_{\ast}]:\qquad \mathcal{F} &\to (0,\infty)^2\times \R
\\
\Lambda = \Lambda[\pi,\sigma_{\ast}]:\qquad \mathcal{F} & \to [1,\infty)\\
\tau_{2-} = \tau_{2-}[\pi,\sigma_{\ast}]:\qquad \mathcal{F} & \to (-\infty,0)
\end{align*}
so that for all $\mathbf{f} = (\mathbf{h},w,q)\in \mathcal{F}$
(see Definitions 
\ref{def:de},
\ref{def:met2},
\ref{def:psizero},
\ref{def:taump},
\ref{def:kdhkhskhkdssPT1111star},
\ref{def:kdhkhskhkdssPT1tau},
\ref{def:kdhkhskhkdssPT2}
)
\begin{enumerate}[(a)]
\item 
\rule{0pt}{10pt}$\|\Pi(\mathbf{f}) - \mathcal{Q}_L(\mathbf{f}) \|_{\R^3} \leq \mathbf{K}(\mathbf{f})$
\item 
\rule{0pt}{10pt}$|\Lambda(\mathbf{f}) -\lambda_L(\mathbf{f})| \leq \mathbf{K}(\mathbf{f})$
\item
\rule{0pt}{10pt}$\tau_-(\mathbf{f}) < \tau_{2-}(\mathbf{f}) < \tfrac{1}{2}\tau_{1-}(\mathbf{f})$
and 
$|\tau_{2-}(\mathbf{f}) - \tau_{1-}(\mathbf{f})|
\leq \mathbf{K}(\mathbf{f})$
\item\rule{0pt}{10pt}$\Pi$, $\Lambda$ and $\tau_{2-}$ are continuous
\item\label{item:exist}
\rule{0pt}{10pt}if we set $\tau_{2-} = \tau_{2-}(\mathbf{f})$,\; $\tau_{2+} = \tau_{1+}(\mathbf{f})$,\; $\pi' = (\mathbf{a}',\mathbf{b}',\mathbf{c}') = \mathcal{P}_L(\pi,\mathbf{f})$,\;
$\lambda = \Lambda(\mathbf{f})$ and $\mathbf{f}' = (\mathbf{h}',w',q') = \Pi(\mathbf{f})$,
then
$\tfrac{1}{2}\leq \tau_{2+}-\tau_{2-} \leq 3$ and there is a smooth field $$\Phi = \alpha \oplus \beta \;\in \; \mathcal{E} = \mathcal{E}(\sigma_{\ast};\tau_{2-},\tau_{2+})$$ that satisfies
\begin{itemize}
\item[(\ref{item:exist}.1)]\rule{0pt}{9pt}$(\boldsymbol{\mathfrak a},\boldsymbol{\mathfrak b},c)[\Phi,\mathbf{h},Z]
= 0$ on $[\tau_{2-},\tau_{2+}]$
\item[(\ref{item:exist}.2)]\rule{0pt}{9pt}$\Phi(\tau_{2+}) = \Phi_{\star}(\pi,\mathbf{f},\sigma_{\ast})$ and $\Phi(\tau_{2-}) = \lambda\,\Phi_{\star}(\pi',\mathbf{f}',\sigma_{\ast})$, in particular
$$\Phi(\tau_{2+}) \in \mathcal{H}(\pi,\sigma_{\ast})
\qquad \text{and} \qquad
\Phi(\tau_{2-}) \in \mathcal{H}(\pi',\sigma_{\ast})
$$
\item[(\ref{item:exist}.3)]\rule{0pt}{9pt}$|\beta_{\mathbf{a}}[\Phi](\tau)| \geq |\beta_{\mathbf{a}'}[\Phi](\tau)|$
for all $\tau \in [\tau_{2-},\tfrac{1}{2}\tau_{1-}(\mathbf{f})]$ with equality iff $\tau = \tau_{2-}$
\item[(\ref{item:exist}.4)]\rule{0pt}{9pt}$d_{\mathcal{E},(\pi,\mathbf{h})}(\Phi,\Phi_0) \leq \mathbf{K}(\mathbf{f})$, where
$\Phi_0 = \Phi_0(\pi,\mathbf{f},\sigma_{\ast})|_{[\tau_{2-},\tau_{2+}]}$
\item[(\ref{item:exist}.5)]\rule{0pt}{9pt}$\sup_{\tau \in [\tau_{2-},\tau_{2+}]} \max\{\alpha_{\mathbf{b},\mathbf{c}}[\Phi],\alpha_{\mathbf{c},\mathbf{a}}[\Phi],\alpha_{\mathbf{a},\mathbf{b}}[\Phi]\}(\tau) 
\leq -2^{-2} \min\{w^2,w^{-1}\}$
\end{itemize}
\end{enumerate}
\end{proposition}
\begin{proof}
The main part of this proof is the construction of the field $\Phi$ that appears in (\ref{item:exist}). To make the proof more transparent, we replace some numerical constants in
\eqref{eq:ldsjdslsj1} and \eqref{eq:ldsjdslsj2} by the components of a
parameter vector $\ell = (\ell_1,\ldots,\ell_8) \in \R^8$.
In the course of the construction of $\Phi$, we require a finite number of inequalities
 of the form  $\ell \geq \ell'$.
Each inequality of this kind is marked by $(\bullet)$ and is \emph{assumed to hold for the rest of the proof}, once it has been stated. At the end of the construction, we check that the particular parameters appearing in \eqref{eq:ldsjdslsj1} and \eqref{eq:ldsjdslsj2} satisfy all these inequalities.\\
Let $\pi = (\mathbf{a},\mathbf{b},\mathbf{c}) \in S_3$ and $\sigma_{\ast}\in \{-1,+1\}^3$. Fix any $\mathbf{f} = (\mathbf{h},w,q)\in (0,\infty)^3$ with $q \neq 1$ and $\mathbf{h} \leq 1$. Set $\tau_{\ast} = \tau_{\ast}(\mathbf{f})$.
 For any $s = (s_1,\ldots,s_7) \in \R^7$, set
\begin{multline*}
\mathbf{X}(s) = \mathbf{X}(s_1,\ldots,s_7) = \\
2^{s_1} \big(\tfrac{1}{\mathbf{h}}\big)^{s_2}\times
\left\{ \begin{array}{l c c} (\tfrac{1}{w})^{s_3} &\;\;& \text{if $w \leq 1$} \\
\rule{0pt}{10pt} w^{s_4} & & \text{if $w > 1$} \end{array} \right\}\times
\left\{ \begin{array}{l c c} (\tfrac{1}{q})^{s_5} &\;\;& \text{if $q \leq 1$} \\
\rule{0pt}{10pt} q^{s_6} & & \text{if $q > 1$} \end{array} \right\}\times
\exp\big(\tfrac{1}{\mathbf{h}} s_7 \tau_{\ast}\big)
\end{multline*}
\emph{Basic properties of $\mathbf{X}(s)$.} The quantity $\mathbf{X}(s)$ is positive, non-decreasing  in each of its  seven arguments (recall $0 < \mathbf{h} \leq 1$), and $\mathbf{X}(s)\mathbf{X}(s') = \mathbf{X}(s+s')$ for all $s,s' \in \R^7$, and $\mathbf{X}(0,\ldots,0) = 1$. Also, we have
$\tau_{\ast} \geq \mathbf{X}(-1,0,0,-1,-1,0,0)$.\\
\emph{Basic smallness assumptions.} Introduce a parameter vector $\ell = (\ell_1,\ldots,\ell_8)\in \R^8$ with
$$(\ell_1,\ldots,\ell_7) \geq (0,0,0,0,0,0,-\infty) \qquad \text{and}\qquad \ell_8 \geq 0 \qquad (\bullet)_1$$
Our basic assumptions on the vector $\mathbf{f} = (\mathbf{h},w,q)$ are:
\begin{align}\label{eq:lsjlsjdfhjld}
q & \neq 1 & \mathbf{K} \stackrel{\text{def}}{=} \mathbf{X}(\ell_1,\ldots,\ell_7) & < 1 &
\mathbf{h} & < 2^{-\ell_8} \tau_{\ast}
\end{align}
Observe that our previous assumptions $q\neq 1$ and $\mathbf{h}\leq 1$ are subsumed in \eqref{eq:lsjlsjdfhjld}.
\emph{Abbreviations.}  $\tau_{\pm} = \tau_{\pm}(\mathbf{f})$ and $\tau_{1\pm} = \tau_{1\pm}(\mathbf{f})$ and
$\tau_{2+} = \tau_{1+}(\mathbf{f})$
and
\begin{align*}
\tau_{0-} & =
\tfrac{1}{2} \tau_{1-} + \tfrac{1}{2} \tau_-  < 0 &
\tau_{0+} & = \tau_{1+} > 0
\end{align*}
and
$\mathcal{D} = \mathcal{D}(\sigma_{\ast})$ and
$\mathcal{E} = \mathcal{E}(\sigma_{\ast};\tau_{0-},\tau_{0+})$
and
$\Phi_0 = \Phi_0(\pi,\mathbf{f},\sigma_{\ast})|_{[\tau_{0-},\tau_{0+}]}$ and
$\Phi_1 = \Phi_1(\pi,\mathbf{f},\sigma_{\ast})|_{[\tau_{0-},\tau_{0+}]}$ and
$\Phi_{\star} = \Phi_{\star}(\pi,\mathbf{f},\sigma_{\ast})$
and
$d_{\mathcal{E}} = d_{\mathcal{E},(\pi,\mathbf{h})}$
and
$d_{\mathcal{D}} = d_{\mathcal{D},(\pi,\mathbf{h})}$
and
$\slaa{d}_{\mathcal{D}} = \slaa{d}_{\mathcal{D},\mathbf{h}}$
and
$\clball{\mathcal{E}}{\,\cdot\,}{\,\cdot\,} =
\clball{\mathcal{E},(\pi,\mathbf{h})}{\,\cdot\,}{\,\cdot\,}$
and
$\pi' = (\mathbf{a}',\mathbf{b}',\mathbf{c}') = \mathcal{P}_L(\pi,\mathbf{f})$.\\
\emph{Preliminaries 1.} Introduce $\epsilon_-$ and $\epsilon_+$ by
$\tau_{0-} = \tau_- + \epsilon_-$ and 
$\tau_{0+} = \tau_+ - \epsilon_+$,
just as in Lemma \ref{lem:tl1}. We have
\begin{align*}
\epsilon_+ & = \tau_+ - \tau_{1+}
= \tfrac{1+w}{1+2w}\big(1+\tfrac{1}{w}-\mathbf{h}\log 2\big)\\
\epsilon_- & = \tfrac{1}{2}(\tau_{1-} - \tau_-)
= \begin{cases}
\tfrac{1}{2(3+w)}\big(\tfrac{1+w}{2+w}\,q - \mathbf{h}\log 2\big) & \text{if $q<1$}\\
\tfrac{1+w}{2(3+2w)}\big(\tfrac{1+w}{2+w} - \mathbf{h}\log 2\big) & \text{if $q>1$}
\end{cases}
\end{align*}
Require $\ell_8 \geq 2$ $(\bullet)_2$.
Then $\mathbf{h}\log 2 \leq \mathbf{h} \leq 2^{-2}\min \{1,q\}$, and
(recall that $\tau_+ = 1+\tfrac{1}{w}$)
$$2^{-2} \leq \epsilon_+/\tau_+ \leq 1
\hskip 20mm
2^{-5} \leq \epsilon_-/\tau_{\ast} \leq 2^{-1}$$
and $\epsilon_- \in (0,-\tau_-)$ and $\epsilon_+ \in (0,\tau_+)$,
as required by Lemma \ref{lem:tl1}.
We have
$$-1 < \tau_- < \tau_{0-} < \tau_{1-} < 0 < \tfrac{1}{2} < \tau_{0+} = \tau_{1+} = \tau_{2+} < \min\{2,\tau_+\}$$ Set
\begin{equation}\label{def:kdshdkhkdhd}
\delta \stackrel{\text{def}}{=} 2^{-9}\min \{1,w\}\,\tau_{\ast}
= \mathbf{X}(-9,0,-1,0,0,0,0)\,\tau_{\ast}
\geq \mathbf{X}(-10,0,-1,-1,-1,0,0)
\end{equation}
This implies $\delta \leq 2^{-4} \min\{1,w,\epsilon_-,\tfrac{\epsilon_+}{\tau_+ \tau_{0+}}\}$, the main hypothesis of Lemma \ref{lem:tl1}. This lemma will be applied later.\\
\emph{Preliminaries 2.} Require $\ell_8 \geq 7$ $(\bullet)_3$. Then
\begin{align}
\notag d_{\mathcal{E}}(\Phi_0,\Phi_1) & \leq 2^2 d_{\mathcal{D}}(\Phi_0(\tau_{1+}),\Phi_{\star}) \leq 2^{11} \slaa{d}_{\mathcal{D}}(\Phi_0(\tau_{1+}),\Phi_{\star})\\
\notag & \leq 2^{18} (1+w) \exp\big(-\tfrac{1}{2\mathbf{h}} \min\{1,q\}\big)  \leq 2^{18} (1+w) \exp\big(-\tfrac{1}{2\mathbf{h}} \tau_{\ast}\big)\\
\label{eq:dkdhfdh28763283} & \leq \mathbf{X}(19,0,0,1,0,0,-2^{-1})
 \leq 2^{-2} \delta\, \mathbf{X}(31,0,1,2,1,0,-2^{-1})
\end{align}
The first and third inequality follow from (d) and (c) in Lemma \ref{lem:kdshkdhkdhkd}, respectively, using
 $\sup_{\tau \in [\tau_{0-},\tau_{0+}]}(1 + |\tau - \tau_{1+}|) \leq 2^2$. The second inequality
follows from Lemma 
\ref{lem:fdhdkfhekjhekr} (b),
with $C=D=2$. Its assumptions  are satisfied, because $\mathbf{h}\leq 2^{-7}$ and $A_{\mathbf{a}}[\Phi_0](\tau_{1+}) = 1$ and
$\mathbf{h}\,\varphi_{\mathbf{a}}[\Phi_0](\tau_{1+}) = \tau_{1+} \in [\tfrac{1}{2},2]$
and
$\xi_{\mathbf{a}}[\Phi_{\star},\mathbf{h}] \in[-\tfrac{3}{2}, - \tfrac{1}{2}]$ and
$|\beta_{\mathbf{a}}[\Phi_{\star}]| \leq 2 \exp(-\tfrac{1}{2\mathbf{h}}) < 1$ 
and $A_{\mathbf{a}}[\Phi_{\star}] \in [1,2]$
and $0 \leq \mathbf{h}|\varphi_{\mathbf{a}}[\Phi_{\star}]| + \xi_{\mathbf{a}}[\Phi_{\star},\mathbf{h}] \leq 2^{-3}$
(see Remark \ref{rem:kfdhdkhkfd}) and $\mathbf{h}\,|\varphi_{\mathbf{a}}[\Phi_{\star}]| \in [\tfrac{1}{2},2]$ 
 and $\sgn \varphi_{\mathbf{a}}[\Phi_{\star}]
= - \sgn \alpha_{\mathbf{a}}[\Phi_{\star}] = + 1$, and because
$\ell_8 \geq 7$ implies $\mathbf{h} \leq 2^{-7}$ and therefore
$\exp(-\tfrac{1}{\mathbf{h}} 2^{-3}) \leq 2^{-12}$.
\\
Require $(\ell_1,\ldots,\ell_7) \geq (31,0,1,2,1,0,-2^{-1})$ $(\bullet)_4$. Then,
by \eqref{eq:lsjlsjdfhjld},
\begin{equation}\label{eq:kdjkhkjhkjhdkdfs}
\Phi_1 \in \clball{\mathcal{E}}{2^{-2}\delta}{\Phi_0}
\end{equation}
\emph{Construction of $\Phi$.} Define a map
$P: \clball{\mathcal{E}}{\delta}{\Phi_0}
\to \clball{\mathcal{E}}{\delta}{\Phi_0}, \Psi \mapsto P(\Psi)$ by
\begin{subequations}\label{eq:kdhskdfhkjdhddffd}
\begin{align}
\label{eq:kdhskdfhkjdhddffd1}
A_{\mathbf{a}}[P(\Psi)](\tau) - A_{\mathbf{a}}[\Phi_1](\tau) & =  \textstyle\int_{\tau_{0+}}^{\tau}\dd \tau'\, \mathbf{I}_1[\Psi,\mathbf{h},\pi](\tau')\\
\label{eq:kdhskdfhkjdhddffd2}
\theta_{\mathbf{a}}[P(\Psi),\mathbf{h}](\tau) -  \theta_{\mathbf{a}}[\Phi_1,\mathbf{h}](\tau) & = 
\textstyle \int_{\tau_{0+}}^{\tau}\dd \tau'\, \mathbf{I}_2[\Psi,\mathbf{h},\pi](\tau')\\
\label{eq:kdhskdfhkjdhddffd3}
\ua_{\mathbf{p},\mathbf{a}}[P(\Psi)](\tau) - \ua_{\mathbf{p},\mathbf{a}}[\Phi_1](\tau)
 & =  \textstyle\int_{\tau_{0+}}^{\tau}\dd \tau'\, \mathbf{I}_{(3,\mathbf{p})}[\Psi,\mathbf{h},\pi](\tau')\\
 \label{eq:kdhskdfhkjdhddffd4}
\xi_{\mathbf{p},\mathbf{a}}[P(\Psi),\mathbf{h}](\tau) - \xi_{\mathbf{p},\mathbf{a}}[\Phi_1,\mathbf{h}](\tau)
& =  \textstyle\int_{\tau_{0+}}^{\tau}\dd \tau'' \int_{\tau_{0+}}^{\tau''} \dd \tau'\, \mathbf{I}_{(3,\mathbf{p})}[\Psi,\mathbf{h},\pi](\tau')
\end{align}
\end{subequations}
for all $\mathbf{p}\in \{\mathbf{b},\mathbf{c}\}$ and $\tau \in [\tau_{0-},\tau_{0+}]$.
To make sure that $P$ is well defined, we require $(\ell_1,\ldots,\ell_7) \geq (28,2,1,1,1,0,-2^{-7})$ $(\bullet)_5$, in which case
Lemma  \ref{lem:tl2} (see Preliminaries 1) implies the uniform estimates
\begin{subequations}
\begin{align}
\label{eq:kdshkhfdkhdkfhksadm1} & |\mathbf{I}_S[\Psi,\mathbf{h},\pi]| \leq \mathbf{X}(12,1,0,0,0,0,-2^{-7})
\leq 2^{-6} \delta\,\mathbf{X}(28,1,1,1,1,0,-2^{-7}) \leq 2^{-6}\delta \\
\label{eq:kdshkhfdkhdkfhksad} & |\mathbf{I}_S[\Psi,\mathbf{h},\pi]
- \hskip-1pt\mathbf{I}_S[\Psi',\mathbf{h},\pi]| \leq 2^{-5} \mathbf{X}(24,2,0,0,0,0,-2^{-7}) d_{\mathcal{E}}(\Psi,\Psi')\leq 2^{-5} d_{\mathcal{E}}(\Psi,\Psi')
\end{align}
\end{subequations}
on the interval $[\tau_{0-},\tau_{0+}]$, for all $\Psi,\Psi' \in \clball{\mathcal{E}}{\delta}{\Phi_0}$ and all
$S\in \{1,2,(3,\mathbf{b}),(3,\mathbf{c})\}$. 
Since $\sup_{\tau \in [\tau_{0-},\tau_{0+}]} |\tau - \tau_{0+}|\leq 4$, we have:
\begin{itemize}
\item  $A_{\mathbf{a}}[P(\Psi)] > \tfrac{1}{2}$ on $[\tau_{0-},\tau_{0+}]$,
which makes  $P(\Psi)$ a well defined element of $\mathcal{E}$. 
\item Each right hand side of \eqref{eq:kdhskdfhkjdhddffd} is $\leq 2^{-2}\delta$, hence $P(\Psi) \in \clball{\mathcal{E}}{\tfrac{1}{2}\delta}{\Phi_0}$.
\item The map $P$ is Lipschitz-continuous with constant $\leq \tfrac{1}{2}$.
\end{itemize}
The metric space $\clball{\mathcal{E}}{\delta}{\Phi_0}$ is nonempty and complete. By the Banach Fixed Point Theorem, the contraction $P$ admits a unique fixed point
\begin{equation}\label{eq:kdhkdhjhd}
\Phi \in \clball{\mathcal{E}}{\tfrac{1}{2}\delta}{\Phi_0}
\end{equation}
\emph{Proof that the fixed point satisfies 
$(\boldsymbol{\mathfrak a},\boldsymbol{\mathfrak b},c)[\Phi,\mathbf{h},Z] = 0$.}
The fixed point $\Phi$ is smooth. We have $\Phi(\tau_{0+}) = \Phi_1(\tau_{0+}) = \Phi_{\star}$ and $c[\Phi,\mathbf{h},Z](\tau_{0+}) = 0$, by Lemma \ref{lem:kdshkdhkdhkd} (b), and because $\tau_{0+} = \tau_{1+}$. Set $\Psi = P(\Psi) = \Phi$ in \eqref{eq:kdhskdfhkjdhddffd} and differentiate with respect to $\tau$.
The result of differentiating  \eqref{eq:kdhskdfhkjdhddffd1} and \eqref{eq:kdhskdfhkjdhddffd2} can be written as
\begin{multline*}
\frac{\dd}{\dd \tau}  \begin{pmatrix}
A_{\mathbf{a}} \\ \theta_{\mathbf{a}}
\end{pmatrix} = \frac{1}{(A_{\mathbf{a}})^2} \begin{pmatrix}
\frac{1}{\mathbf{h}} (A_{\mathbf{a}})^2 \tanh \varphi_{\mathbf{a}} & & \frac{1}{\mathbf{h}} (A_{\mathbf{a}})^2 \sech \varphi_{\mathbf{a}} \\
\rule{0pt}{10pt} \varphi_{\mathbf{a}} \tanh \varphi_{\mathbf{a}} - 1&\;\, & \sinh \varphi_{\mathbf{a}} + \varphi_{\mathbf{a}} \sech \varphi_{\mathbf{a}}
\end{pmatrix}\\
\times \begin{pmatrix}
\boldsymbol{\mathfrak a}_{\mathbf{a}}[\Phi,\mathbf{h},B_{\mathbf{a}}]
- \boldsymbol{\mathfrak a}_{\mathbf{a}}[\Phi,\mathbf{h},Z]\\
- (\sigma_{\ast})_{\mathbf{a}}\, \boldsymbol{\mathfrak b}_{\mathbf{a}}[\Phi,\mathbf{h},B_{\mathbf{a}}]
+ (\sigma_{\ast})_{\mathbf{a}}\, \boldsymbol{\mathfrak b}_{\mathbf{a}}[\Phi,\mathbf{h},Z]
\end{pmatrix}
\end{multline*}
where $A_{\mathbf{a}} = A_{\mathbf{a}}[\Phi]$, $\theta_{\mathbf{a}} = \theta_{\mathbf{a}}[\Phi,\mathbf{h}]$, $\varphi_{\mathbf{a}} = \varphi_{\mathbf{a}}[\Phi]$, because
 $\boldsymbol{\mathfrak b}_{\mathbf{a}}[\Phi,\mathbf{h},B_{\mathbf{a}}] = \boldsymbol{\mathfrak b}_{\mathbf{a}}[\Phi,\mathbf{h},Z]$. Now, Lemma \ref{lem:lfdskhjkfdhkfdhkfd} implies 
$\boldsymbol{\mathfrak a}_{\mathbf{a}}[\Phi,\mathbf{h},Z] = \boldsymbol{\mathfrak b}_{\mathbf{a}}[\Phi,\mathbf{h},Z] = 0$. Differentiation of \eqref{eq:kdhskdfhkjdhddffd3} gives
$
\tfrac{\dd}{\dd \tau} \ua_{\mathbf{p},\mathbf{a}}[\Phi] 
 = 
\tfrac{1}{\mathbf{h}} \boldsymbol{\mathfrak a}_{\mathbf{p}} [\Phi,\mathbf{h},Z,B_{\mathbf{a}}] + 
\tfrac{1}{\mathbf{h}} \boldsymbol{\mathfrak a}_{\mathbf{a}} [\Phi,\mathbf{h},Z,B_{\mathbf{a}}]
$. 
Together with
$\boldsymbol{\mathfrak a}_{\mathbf{a}}[\Phi,\mathbf{h},Z] = 0$ and 
the general identity
$\boldsymbol{\mathfrak a}_{\mathbf{p}} [\Phi,\mathbf{h},B_{\mathbf{a}}] + 
\boldsymbol{\mathfrak a}_{\mathbf{a}} [\Phi,\mathbf{h},B_{\mathbf{a}}]
= - \mathbf{h}\tfrac{\dd}{\dd \tau} \ua_{\mathbf{p},\mathbf{a}}[\Phi]$, we obtain
$\boldsymbol{\mathfrak a}_{\mathbf{p}} [\Phi,\mathbf{h},Z] = 0$. 
Differentiating \eqref{eq:kdhskdfhkjdhddffd4} and simplifying the result with
 \eqref{eq:kdhskdfhkjdhddffd3} gives
$
\tfrac{\dd}{\dd \tau} \xi_{\mathbf{p},\mathbf{a}}[\Phi,\mathbf{h}] =  \ua_{\mathbf{p},\mathbf{a}}[\Phi] 
$
which, by $\boldsymbol{\mathfrak b}_{\mathbf{a}}[\Phi,\mathbf{h},Z] = 0$, implies 
$\boldsymbol{\mathfrak b}_{\mathbf{p}}[\Phi,\mathbf{h},Z] = 0$.
Now, Proposition \ref{prop:conslaw} and the fact that 
 $c[\Phi,\mathbf{h},Z](\tau_{0+}) = 0$ imply that $c[\Phi,\mathbf{h},Z] = 0$ identically on $[\tau_{0-},\tau_{0+}]$.\\
\emph{Estimates on $\Phi$.} By the fixed point equation $P(\Phi) = \Phi$ and by 
\eqref{eq:dkdhfdh28763283} and \eqref{eq:kdshkhfdkhdkfhksadm1},
\begin{align}
\notag d_{\mathcal{E}}(\Phi_0,\Phi) & \leq d_{\mathcal{E}}(\Phi_0,\Phi_1) + d_{\mathcal{E}}(\Phi_1,P(\Phi))\\
\notag & \leq \mathbf{X}(19,0,0,1,0,0,-2^{-1}) + \mathbf{X}(16,1,0,0,0,0,-2^{-7})\\
\label{eq:kkdhkhfdkssdsdk} & \leq \mathbf{X}(20,1,0,1,0,0,-2^{-7})
\end{align}
We require $(\ell_1,\ldots,\ell_7) \geq (20,1,0,1,0,0,-2^{-7})$ $(\bullet)_6$, which implies
$d_{\mathcal{E}}(\Phi_0,\Phi) \leq \mathbf{K}$.
To apply Lemma \ref{lem:fdhdkfhekjhekr} (a), set $\mathcal{J} = [\tau_{0-},\tfrac{1}{2}\tau_{1-}]\subset [\tau_{0-},\tau_{0+}]$ and $C=2$ and $D = 12 \max\{1,q^{-1}\}$.
We check the assumptions of Lemma \ref{lem:fdhdkfhekjhekr}. The inequalities
\begin{align*}
C^{-1}\leq A_{\mathbf{a}}[X] & \leq C &
D^{-1} \leq\mathbf{h}\,|\varphi_{\mathbf{a}}[X]| & \leq D
\end{align*}
hold for both $X = \Phi_0(\tau)$ and $X=\Phi(\tau)$, for all $\tau \in \mathcal{J}$.
The inequality for $A_{\mathbf{a}}$ follows from
$A_{\mathbf{a}}[\Phi_0](\tau)=1$ and the bound $d_{\mathcal{E}}(\Phi_0,\Phi) \leq \delta \leq 2^{-9}$.
To check the inequality for $\varphi_{\mathbf{a}}$, observe that
$\mathbf{h} \varphi_{\mathbf{a}}[\Phi_0](\tau) = \tau \in \mathcal{J} \subset [-(D/2),-(D/2)^{-1}]$, see the definitions of $\tau_-$ and $\tau_{1-}$. 
Furthermore, for all $\tau \in \mathcal{J}$, we have
$$|\mathbf{h}\varphi_{\mathbf{a}}[\Phi] -
\mathbf{h}\varphi_{\mathbf{a}}[\Phi_0]|
\leq |\tau|\,|A_{\mathbf{a}}[\Phi] - 1| + A_{\mathbf{a}}[\Phi]\,|\theta_{\mathbf{a}}[\Phi,\mathbf{h}]| \leq 4\delta \leq (2D)^{-1}$$
This implies $\mathbf{h}\,\varphi_{\mathbf{a}}[\Phi](\tau) \in [-D,-D^{-1}]$ and
$\sgn \varphi_{\mathbf{a}}[\Phi_0](\tau) = \sgn \varphi_{\mathbf{a}}[\Phi](\tau) = -1$ for all $\tau \in \mathcal{J}$.
Now, Lemma \ref{lem:fdhdkfhekjhekr} (a) and $2^3C^2D \leq \mathbf{X}(9,0,0,0,1,0,0)$ imply
for all $\tau \in \mathcal{J}$
\begin{equation}\label{defMM}
\slaa{d}_{\mathcal{D}}(\Phi_0(\tau),\Phi(\tau)) \leq
2^3C^2D\,d_{\mathcal{E}}(\Phi_0,\Phi) \leq \mathbf{X}(29,1,0,1,1,0,-2^{-7})
\stackrel{\text{def}}{=} \mathbf{M}
\end{equation}
\emph{Construction of $\tau_{2-}$.} 
Recall that $\mathbf{a}' = \mathbf{c}$ if $q < 1$ and $\mathbf{a}' = \mathbf{b}$ if $q > 1$.
By \eqref{eq:ddhkdfddfdf},
\begin{subequations}\label{eq:kdshdkjhkfdhkdfd}
\begin{align}
\label{eq:kdshdkjhkfdhkdfd1}
\big(\xi_{\mathbf{a}}[\Phi_0,\mathbf{h}] - \xi_{\mathbf{a}'}[\Phi_0,\mathbf{h}]\big)\, F
& = \tau - \tau_{1-} - T_1\\
\label{eq:kdshdkjhkfdhkdfd2}
\big( 
\xi_{\mathbf{a}}[\Phi,\mathbf{h}] - \xi_{\mathbf{a}'}[\Phi,\mathbf{h}]
\big)\, F &  = \tau - \tau_{1-} - T_2
\end{align}
\end{subequations}
for all $\tau \in \mathcal{J}$, where $F$ is given by \eqref{eq:ddhkdfddfdf1}, and
\begin{align*}
T_1 = T_1(\tau) & =  2\mathbf{h}\log\big(1 + e^{2\tau/\mathbf{h}}\big)\, F\\
T_2 = T_2(\tau) & =  T_1
 - \big(
 \xi_{\mathbf{a}}[\Phi,\mathbf{h}] - \xi_{\mathbf{a}}[\Phi_0,\mathbf{h}]
 \big)\,F
 + \big(
 \xi_{\mathbf{a}'}[\Phi,\mathbf{h}] - \xi_{\mathbf{a}'}[\Phi_0,\mathbf{h}]
 \big)\,F
\end{align*}
For all $\tau \in \mathcal{J}$ we have $0 < T_1 \leq 2\mathbf{h}\, e^{2\tau/\mathbf{h}} F \leq 2\mathbf{h}\, e^{\tau_{1-}/\mathbf{h}} F
\leq \mathbf{M} F$ and therefore
$|T_2| \leq 3 \mathbf{M}F \leq \tfrac{3}{2} \mathbf{M}$. Estimate
\begin{equation*}
\mathrm{dist}_{\R}\big(\tau_{1-},\,\R\setminus \mathcal{J}\big)
= \min\big\{\tfrac{1}{2} |\tau_{1-}|,\epsilon_-\big\} \geq 2^{-5}\tau_{\ast}
\geq \mathbf{X}(-6,0,0,-1,-1,0,0)
\end{equation*}
Therefore, the condition
$(\ell_1,\ldots,\ell_7) \geq (37,1,0,2,2,0,-2^{-7})$ $(\bullet)_7$ yields
\begin{equation}\label{eq:sgggggggsdhs}
|T_2| \leq \tfrac{1}{2} \mathrm{dist}_{\R}\big(\tau_{1-},\,\R\setminus \mathcal{J}\big)
\end{equation}
for all $\tau \in \mathcal{J}$.
Set
\begin{equation}\label{c5}
\tau_{2-} = \sup\big\{\,\tau \in \mathcal{J}\;\big|\;
\xi_{\mathbf{a}}[\Phi,\mathbf{h}](\tau) \leq \xi_{\mathbf{a}'}[\Phi,\mathbf{h}](\tau)\, \big\}
\end{equation}
The set on the right is nonempty, by \eqref{eq:kdshdkjhkfdhkdfd2} and \eqref{eq:sgggggggsdhs}, it contains $\tau_{0-}$. We have $\tau_{2-} \in (\tau_{0-},\tfrac{1}{2}\tau_{1-}) \subset \mathcal{J}$
and, by continuity, $\xi_{\mathbf{a}}[\Phi,\mathbf{h}](\tau_{2-}) = \xi_{\mathbf{a}'}[\Phi,\mathbf{h}](\tau_{2-})$, and  $|\tau_{2-} - \tau_{1-}|  \leq \tfrac{3}{2}\mathbf{M}$. For all $\tau \in [\tau_{2-},\tfrac{1}{2}\tau_{1-}]$, we have
$|\beta_{\mathbf{a}}[\Phi](\tau)| \geq |\beta_{\mathbf{a}'}[\Phi](\tau)|$
with equality iff $\tau = \tau_{2-}$. The condition $(\ell_1,\ldots,\ell_7) \geq (31,1,0,1,1,0,-2^{-7})$ $(\bullet)_8$
implies $|\tau_{2-} - \tau_{1-}| \leq \mathbf{K}$.\\
\emph{Estimates on $\Phi_0$.} For all $\tau \in \mathcal{J}$, we have
\begin{align*}
|\alpha_{\mathbf{a}}[\Phi_0](\tau) - 1| & = 
|\tanh \tfrac{1}{\mathbf{h}}|\tau| - 1| \leq 2\exp(-\tfrac{2}{\mathbf{h}}|\tau|\big)\\
|\xi_{\mathbf{a}}[\Phi_0,\mathbf{h}](\tau) - \tau| & = 
\big|\mathbf{h}\log (2 \cosh \tfrac{1}{\mathbf{h}}|\tau|) - |\tau|\big|
 \leq \mathbf{h}\, \exp\big(-\tfrac{2}{\mathbf{h}}|\tau|\big)\\
 \exp(-\tfrac{2}{\mathbf{h}}|\tau|\big)
& \leq \exp(-\tfrac{1}{\mathbf{h}}|\tau_{1-}|\big) \leq
\exp(-\tfrac{1}{\mathbf{h}} 2^{-2} \tau_{\ast}\big) \leq 2^{-29}\mathbf{M}
\end{align*}
These estimates will be used without further comment.\\
\emph{Construction of $\lambda$.} 
Set $\lambda_L = \lambda_L(\mathbf{f})$
and recall \eqref{eq:lkdkjdhkjdhkdfhkdf1}. Set
\begin{equation}\label{c4}
\lambda = - \alpha_{\mathbf{a}'}[\Phi](\tau_{2-})
\end{equation}
 Then,
\begin{align*}
|\lambda - \lambda_L| & \leq \big| \alpha_{\mathbf{a}'}[\Phi](\tau_{2-}) - \alpha_{\mathbf{a}'}[\Phi_0](\tau_{2-})\big|
+ \big| \alpha_{\mathbf{a}'}[\Phi_0](\tau_{2-}) - \alpha_{\mathbf{a}'}[\Phi_0](\tau_{1-})\big|
\\ & \hskip2mm +
\big| \alpha_{\mathbf{a}'}[\Phi_0](\tau_{1-}) + \big(1 - \alpha_{\mathbf{a},\mathbf{a}'}[\Phi_0](\tau_{1-})\big)\big|\\
& \leq \big| \alpha_{\mathbf{a}'}[\Phi](\tau_{2-}) - \alpha_{\mathbf{a}'}[\Phi_0](\tau_{2-})\big|\\
& \hskip2mm + \Big(\big| \alpha_{\mathbf{a}}[\Phi_0](\tau_{2-}) -1\big|
+ \big|\alpha_{\mathbf{a}}[\Phi_0](\tau_{1-})-1\big|\Big)
+ \big| 1 - \alpha_{\mathbf{a}}[\Phi_0](\tau_{1-})\big| \leq 2\mathbf{M}
\end{align*}
See \eqref{defMM}. Require $(\ell_1,\ldots,\ell_7) \geq (32,1,0,2,1,0,-2^{-7})$ $(\bullet)_9$. Then $4 \mathbf{M} (1+w) \leq \mathbf{K}\leq 1$ and $|\lambda-\lambda_L| \leq \tfrac{1}{2}(1+w)^{-1} \mathbf{K}$. In particular $\lambda \geq \lambda_L - (1+w)^{-1} \geq 1$.
\vskip 2mm
We now construct the components of $\mathbf{f}' = (\mathbf{h}',w',q')$.\\
\emph{Construction of $w'$.} Require
$(\ell_1,\ldots,\ell_7) \geq (32,1,0,2,1,0,-2^{-7})$ $(\bullet)_{10}$ and set
\begin{equation}\label{c2}
w' = \frac{\alpha_{\mathbf{a}}[\Phi](\tau_{2-})}{-\alpha_{\mathbf{a},\mathbf{a}'}[\Phi](\tau_{2-})} > 0
\end{equation}
To check that the denominator is nonzero and that $w' > 0$, note that for all $\tau \in \mathcal{J}$:
\begin{align*}
 \big|
\alpha_{\mathbf{a},\mathbf{a}'}[\Phi](\tau)
- 
\alpha_{\mathbf{a},\mathbf{a}'}[\Phi_0](\tau)
\big| & \leq 2\mathbf{M}\\
 |\alpha_{\mathbf{a}}[\Phi](\tau)-1| & \leq  \big|\alpha_{\mathbf{a}}[\Phi](\tau) - \alpha_{\mathbf{a}}[\Phi_0](\tau)\big| + \big|\alpha_{\mathbf{a}}[\Phi_0](\tau)-1\big| \leq 2\mathbf{M}
\end{align*}
and $4 \mathbf{M}
\leq
\mathbf{X}(-1,0,0,-1,0,0,0) \leq \tfrac{1}{1+w}\leq |\alpha_{\mathbf{a},\mathbf{a}'}[\Phi_0](\tau)|$
and $4\mathbf{M}\leq 1$. Hence,
$$|\alpha_{\mathbf{a},\mathbf{a}'}[\Phi](\tau) - \alpha_{\mathbf{a},\mathbf{a}'}[\Phi_0](\tau)|
\leq \tfrac{1}{2}|\alpha_{\mathbf{a},\mathbf{a}'}[\Phi_0](\tau)|
\qquad |\alpha_{\mathbf{a}}[\Phi](\tau) -1| \leq \tfrac{1}{2}$$
In particular, $\alpha_{\mathbf{a},\mathbf{a}'}[\Phi](\tau_{2-}) \leq \tfrac{1}{2}
 \alpha_{\mathbf{a},\mathbf{a}'}[\Phi_0](\tau_{2-}) < 0$ and $\alpha_{\mathbf{a}}[\Phi](\tau_{2-}) > 0$. Consequently, $w'$ is well defined and positive. Recall \eqref{eq:lkdkjdhkjdhkdfhkdf2} and estimate
\begin{align*}
|w'-w_L| &  \leq \Big|\frac{\alpha_{\mathbf{a}}[\Phi](\tau_{2-})}{\alpha_{\mathbf{a},\mathbf{a}'}[\Phi](\tau_{2-})}
- \frac{\alpha_{\mathbf{a}}[\Phi_0](\tau_{2-})}{\alpha_{\mathbf{a},\mathbf{a}'}[\Phi](\tau_{2-})}\Big|
+ \Big|\frac{\alpha_{\mathbf{a}}[\Phi_0](\tau_{2-})}{\alpha_{\mathbf{a},\mathbf{a}'}[\Phi](\tau_{2-})}
- \frac{\alpha_{\mathbf{a}}[\Phi_0](\tau_{2-})}{\alpha_{\mathbf{a},\mathbf{a}'}[\Phi_0](\tau_{2-})}\Big|\\
& \qquad +
\Big|\frac{\alpha_{\mathbf{a}}[\Phi_0](\tau_{2-})}{\alpha_{\mathbf{a},\mathbf{a}'}[\Phi_0](\tau_{2-})}
- \frac{1}{\alpha_{\mathbf{a},\mathbf{a}'}[\Phi_0](\tau_{2-})} \Big|\\
& \leq 2w_L\mathbf{M} + 4w_L^2\mathbf{M} + w_L\mathbf{M} \leq
2^3(1+w)^2\mathbf{M} \leq 
\tfrac{1}{2}\mathbf{X}(6,0,0,2,0,0,0)\,\mathbf{M}
\end{align*}
We require $(\ell_1,\ldots,\ell_7) \geq (35,1,0,3,1,0,-2^{-7})$ $(\bullet)_{11}$. Hence $|w'-w_L|\leq \tfrac{1}{2}\mathbf{K}\leq \tfrac{1}{2}$.\\
\emph{Construction of $\mathbf{h}'$.}
Let $\lambda$ and $w'$ be given by \eqref{c4} and \eqref{c2}. Set
$$
\mu = \tfrac{1+2w'}{1+w'}\big(
- \xi_{\mathbf{a}}[\Phi,\mathbf{h}](\tau_{2-}) + \mathbf{h}\, \log \lambda\big) - \mathbf{h} \log 2
$$
Recall \eqref{eq:lkdkjdhkjdhkdfhkdf3} and estimate
\begin{align*}
\big|\mu - \tfrac{\mathbf{h}}{\mathbf{h}_L}\big| & \leq \tfrac{1 + 2w'}{1+w'} \big|
- \xi_{\mathbf{a}}[\Phi,\mathbf{h}](\tau_{2-}) + \mathbf{h}\, \log \lambda
+ \tau_{1-} - \mathbf{h}\,\log \lambda_L\big|\\
& \qquad + 
\big| \tfrac{1 + 2w'}{1+w'} - \tfrac{1 + 2w_L}{1+w_L}\big|
\,\big| \tau_{1-} - \mathbf{h}\,\log \lambda_L\big|\\
& \leq
2 \big|\tau_{1-} - \tau_{2-}\big|
+ 2 \big|\tau_{2-} - \xi_{\mathbf{a}}[\Phi_0,\mathbf{h}](\tau_{2-})\big|\\
& \qquad + 2 \big|\xi_{\mathbf{a}}[\Phi_0,\mathbf{h}](\tau_{2-}) - \xi_{\mathbf{a}}[\Phi,\mathbf{h}](\tau_{2-})\big| + 4 \mathbf{h}\, |\lambda - \lambda_L| 
+ 4\,\tfrac{|w'-w_L|}{(1+w_L)^2}\\
& \leq 2^2\mathbf{M} + \mathbf{M} + 2\mathbf{M} + 2^3\mathbf{M} + 2^4\mathbf{M} \leq 2^5\mathbf{M}
\end{align*}
For the second inequality, use $(1+2w') \leq 2(1+w')$ and
$\lambda,\lambda_L \geq \tfrac{1}{2}$ and $|\tau_{1-}|\leq 1$ and $|\mathbf{h}\log \lambda_L| \leq |\tau_{\ast}\log \lambda_L| \leq 1$, see \eqref{eq:lsjlsjdfhjld}, and
$1+w' \geq \tfrac{1}{2}(1+w_L)$. By inspection,
$$
\tfrac{\mathbf{h}}{\mathbf{h}_L} \geq \tfrac{1+w}{2+w}\,\min\{1,q\}
\geq \mathbf{X}(-1,0,0,0,-1,0,0)
$$
To make sure that $\mu > 0$, we require $(\ell_1,\ldots,\ell_7) \geq (36,1,0,1,2,0,-2^{-7})$ $(\bullet)_{12}$, so that
$2^5\mathbf{M}\leq \tfrac{1}{2} \mathbf{X}(-1,0,0,0,-1,0,0)\, \mathbf{K}\leq \tfrac{1}{2}\tfrac{\mathbf{h}}{\mathbf{h}_L}$, that is $|\mu -\tfrac{\mathbf{h}}{\mathbf{h}_L}| \leq \tfrac{1}{2}\tfrac{\mathbf{h}}{\mathbf{h}_L}$ and $\mu > 0$. Set
\begin{equation}\label{c1}
\mathbf{h}' = \mathbf{h}/\mu > 0
\end{equation}
Require $\ell_8 \geq 7$ $(\bullet)_{13}$, so that $\mathbf{h} \leq \mathbf{X}(-7,0,0,0,-1,0,0)$ and
\begin{align*}
\big| \mathbf{h}'-\mathbf{h}_L\big| & = \mathbf{h} \tfrac{\mathbf{h}/\mathbf{h}_L}{\mu}\,\big(\tfrac{\mathbf{h}_L}{\mathbf{h}}\big)^2\,\big|\mu - \tfrac{\mathbf{h}}{\mathbf{h}_L}\big|
\leq \tfrac{1}{2}\, \mathbf{X}(2,0,0,0,1,0,0)\,\mathbf{M}
\end{align*}
We require $(\ell_1,\ldots,\ell_7) \geq (31,1,0,1,2,0,-2^{-7})$ $(\bullet)_{14}$. Then $|\mathbf{h}'-\mathbf{h}_L|\leq \tfrac{1}{2}\mathbf{K} \leq \tfrac{1}{2}$.\\
\emph{Construction of $q'$.} Set
\begin{equation}\label{c3}
q' = \tfrac{1}{1+w'}\Big(
\mathbf{h}'\log \lambda - \tfrac{\mathbf{h}'}{\mathbf{h}} \xi_{\mathbf{c}'}[\Phi,\mathbf{h}](\tau_{2-})
- \tfrac{w'(1+w')}{1+2w'}
- \tfrac{1+3w'+(w')^2}{1+2w'}\,\mathbf{h}'\log 2
\Big)
\end{equation}
Recall \eqref{eq:lkdkjdhkjdhkdfhkdf4} and estimate
\begin{align*}
& |q'-q_L|\\
& \leq \big|\tfrac{1}{1+w'}\mathbf{h}'\log \lambda
- \tfrac{1}{1+w_L} \mathbf{h}_L \log \lambda_L\big|\\
& \hskip 6pt +
\big| \tfrac{1}{1+w'}\tfrac{\mathbf{h}'}{\mathbf{h}} \xi_{\mathbf{c}'}[\Phi,\mathbf{h}](\tau_{2-})
- \tfrac{1}{1+w_L} \tfrac{\mathbf{h}_L}{\mathbf{h}}\,\big(
\xi_{\mathbf{a},\mathbf{c}'}[\Phi_0,\mathbf{h}](\tau_{1-})-\tau_{1-}\big)
\big|\\
& \hskip 6pt +
 \big|\tfrac{w'}{1+2w'}
- \tfrac{w_L}{1+2w_L}\big|
+ \big|\tfrac{1+3w'+(w')^2}{(1+w')(1+2w')}\,\mathbf{h}'
- \tfrac{1+3w_L+(w_L)^2}{(1+w_L)(1+2w_L)}\,\mathbf{h}_L\big| \displaybreak[0]\\
& \leq
 \tfrac{1}{1+w'}\mathbf{h}' \big|\log \lambda - \log \lambda_L\big|
 +\tfrac{1}{1+w'}  \big|\mathbf{h}'- \mathbf{h}_L \big| \log \lambda_L 
 +\big|\tfrac{1}{1+w'}- \tfrac{1}{1+w_L}\big| \mathbf{h}_L \log \lambda_L \\
& \hskip 6pt
+ \tfrac{1}{1+w'}\tfrac{\mathbf{h}'}{\mathbf{h}} \big|  \xi_{\mathbf{c}'}[\Phi,\mathbf{h}](\tau_{2-})
-  \xi_{\mathbf{c}'}[\Phi_0,\mathbf{h}](\tau_{2-})\big|
+ \tfrac{1}{1+w'} \big| \tfrac{\mathbf{h}'}{\mathbf{h}} 
- \tfrac{\mathbf{h}_L}{\mathbf{h}} \big|\,\big|\xi_{\mathbf{c}'}[\Phi_0,\mathbf{h}](\tau_{2-})\big|\\
& \hskip 6pt 
+ \big| \tfrac{1}{1+w'}
- \tfrac{1}{1+w_L}\big|\,
\tfrac{\mathbf{h}_L}{\mathbf{h}}\,\big|\xi_{\mathbf{c}'}[\Phi_0,\mathbf{h}](\tau_{2-})\big|
\\
& \hskip 6pt 
+ \tfrac{1}{1+w_L}\tfrac{\mathbf{h}_L}{\mathbf{h}} \big|  \xi_{\mathbf{c}'}[\Phi_0,\mathbf{h}](\tau_{2-})
- 
\xi_{\mathbf{a},\mathbf{c}'}[\Phi_0,\mathbf{h}](\tau_{1-}) + \tau_{1-}
\big|\\
& \hskip 6pt +
 \big|\tfrac{w'}{1+2w'}
- \tfrac{w_L}{1+2w_L}\big|
+ \big|\tfrac{1+3w'+(w')^2}{(1+w')(1+2w')}
- \tfrac{1+3w_L+(w_L)^2}{(1+w_L)(1+2w_L)}\big|\,\mathbf{h}'
+ \big|\mathbf{h}' - \mathbf{h}_L\big|  \displaybreak[0]\\
& \leq
2^3\big|\lambda - \lambda_L\big|
 +  \big|\mathbf{h}'- \mathbf{h}_L \big| (1+w)
 +
4 \tfrac{|w'-w_L|}{(1+w_L)^2} (1+w)\\
& \hskip 6pt
+ 2\tfrac{\mathbf{h}_L}{\mathbf{h}} \big|  \xi_{\mathbf{c}'}[\Phi,\mathbf{h}](\tau_{2-})
-  \xi_{\mathbf{c}'}[\Phi_0,\mathbf{h}](\tau_{2-})\big|
+  \big| \tfrac{\mathbf{h}'}{\mathbf{h}} 
- \tfrac{\mathbf{h}_L}{\mathbf{h}} \big|\,\big|\xi_{\mathbf{c}'}[\Phi_0,\mathbf{h}](\tau_{2-})\big|\\
& \hskip 6pt 
+ 2\tfrac{|w'-w_L|}{(1+w_L)^2}
\, \tfrac{\mathbf{h}_L}{\mathbf{h}}\,\big|\xi_{\mathbf{c}'}[\Phi_0,\mathbf{h}](\tau_{2-})\big|
+ \tfrac{\mathbf{h}_L}{\mathbf{h}} \big|  \xi_{\mathbf{a},\mathbf{c}'}[\Phi_0,\mathbf{h}](\tau_{2-})
- 
\xi_{\mathbf{a},\mathbf{c}'}[\Phi_0,\mathbf{h}](\tau_{1-})
\big|\\
& \hskip 6pt 
+ \tfrac{\mathbf{h}_L}{\mathbf{h}} \big|  \xi_{\mathbf{a}}[\Phi_0,\mathbf{h}](\tau_{2-})
- \tau_{2-}
\big|
+ \tfrac{\mathbf{h}_L}{\mathbf{h}} \big|  \tau_{2-}-\tau_{1-}
\big|
\\ & \hskip 6pt
+
2\tfrac{|w'-w_L|}{(1+w_L)^2}
+ 2^3\tfrac{|w'-w_L|}{(1+w_L)^2}
+ \big|\mathbf{h}' - \mathbf{h}_L\big|  \displaybreak[0]\\
& \leq 2^4\mathbf{M} + \mathbf{X}(2,0,0,1,1,0,0)\,\mathbf{M} + \mathbf{X}(5,0,0,1,0,0,0)\,\mathbf{M}\\
& \hskip 6pt
+ \mathbf{X}(2,0,0,0,1,0,0)\,\mathbf{M} + \mathbf{X}(6,1,0,1,1,1,0)\,\mathbf{M}\\
& \hskip 6pt
+ \mathbf{X}(9,0,0,1,1,1,0)\,\mathbf{M} + \mathbf{X}(3,0,0,1,1,0,0)\,\mathbf{M}\\
& \hskip 6pt
+ \mathbf{X}(1,0,0,0,1,0,0)\,\mathbf{M} + \mathbf{X}(2,0,0,0,1,0,0)\,\mathbf{M}\\
& \hskip 6pt 
+ 2^3\mathbf{M} + 2^5 \mathbf{M} + \mathbf{X}(1,0,0,0,1,0,0)\,\mathbf{M}\\
& \leq \tfrac{1}{2}\mathbf{X}(11,1,0,1,1,1,0)\,\mathbf{M}
\end{align*}
For the third inequality, use $\mathbf{h}_L \leq \mathbf{X}(1,0,0,0,1,0,0)\,\mathbf{h}
\leq 2$ and $\mathbf{h}' = \mathbf{h}/\mu \leq 2\mathbf{h}_L \leq 2^2$ and $\lambda,\lambda_L\geq \tfrac{1}{2}$ and $|\mathbf{h}\log \lambda_L|\leq 1$ and $\log \lambda_L \leq 1+w$ and $(1+w') \geq \tfrac{1}{2}(1+w_L)$.
For the fourth inequality, use $(1+w)\leq \mathbf{X}(1,0,0,1,0,0,0)$ and 
$\tfrac{|w'-w_L|}{(1+w_L)^2} \leq 2^2\mathbf{M}$ and
\begin{align*}
\big|\xi_{\mathbf{c}'}[\Phi_0,\mathbf{h}](\tau_{2-})\big|
& \leq 
\big|\xi_{\mathbf{a},\mathbf{c}'}[\Phi_0,\mathbf{h}](\tau_{2-})\big|
+ \big|\xi_{\mathbf{a}}[\Phi_0,\mathbf{h}](\tau_{2-}) - \tau_{2-}\big| + |\tau_{2-}|\\
& \leq \mathbf{X}(3,0,0,1,0,1,0) + 1 + 1 \leq \mathbf{X}(5,0,0,1,0,1,0)
\end{align*}
We require $(\ell_1,\ldots,\ell_7) \geq
(40,2,0,2,2,1,-2^{-7})$ $(\bullet)_{15}$, such that $|q'-q_L| \leq \tfrac{1}{2}\mathbf{K} \leq \tfrac{1}{2}$. \\
\emph{The maximum of $\alpha_{\mathbf{b},\mathbf{c}}$, $\alpha_{\mathbf{c},\mathbf{a}}$, $\alpha_{\mathbf{a},\mathbf{b}}$.}
Require $(\ell_1,\ldots,\ell_7)\hskip-1pt \geq \hskip-1pt(24,1,2,2,0,0,-2^{-7})$ $(\bullet)_{16}$.
Then $\mathbf{X}(20,1,0,1,0,0,-2^{-7}) \leq 2^{-3}(1+w)^{-1} \min\{w^2,1\}$. By the inequality
\eqref{eq:kkdhkhfdkssdsdk},
we have $d_{\mathcal{D}}(\Phi_0(\tau),\Phi(\tau)) \leq 2^{-3}(1+w)^{-1} \min\{w^2,1\}$, for all $\tau \in [\tau_{0-},\tau_{0+}]$. Hence,
\begin{align*}
\alpha_{\mathbf{a},\mathbf{p}}[\Phi] & \leq 
\alpha_{\mathbf{a},\mathbf{p}}[\Phi_0] + 
d_{\mathcal{D}}(\Phi_0,\Phi) \leq - (1+w)^{-1} 
+ d_{\mathcal{D}}(\Phi_0,\Phi) \leq - 2^{-1}(1+w)^{-1}\\
\alpha_{\mathbf{b},\mathbf{c}}[\Phi] & \leq 
\alpha_{\mathbf{a},\mathbf{b}}[\Phi] + 
\alpha_{\mathbf{a},\mathbf{c}}[\Phi]
- 2 \alpha_{\mathbf{a}}[\Phi] \leq  
\alpha_{\mathbf{a},\mathbf{b}}[\Phi] + 
\alpha_{\mathbf{a},\mathbf{c}}[\Phi]
+  2 A_{\mathbf{a}}[\Phi] \\
&  \leq 2^2 d_{\mathcal{D}}(\Phi_0,\Phi) 
+ \alpha_{\mathbf{a},\mathbf{b}}[\Phi_0] + 
\alpha_{\mathbf{a},\mathbf{c}}[\Phi_0]
+  2 A_{\mathbf{a}}[\Phi_0] \leq - 2^{-1} (1+w)^{-1} w^2
\end{align*}
for all $\tau \in [\tau_{0-},\tau_{0+}]$ and all $\mathbf{p}\in \{\mathbf{b},\mathbf{c}\}$.\\
\emph{Definition of the maps $\Pi$, $\Lambda$ and $\tau_{2-}$.}
Set $(\ell_1,\ldots,\ell_7) = (40,2,2,3,2,1,-2^{-7})$ and $\ell_8 = 7$. With this choice, all inequalities $(\bullet)$ hold.
The constant $\mathbf{K}$ defined by \eqref{eq:lsjlsjdfhjld} coincides with $\mathbf{K}(\mathbf{f})$, defined by \eqref{eq:ldsjdslsj1}. Furthermore, a vector $\mathbf{f} = (\mathbf{h},w,q)\in (0,\infty)^3$ satisfies our basic assumption \eqref{eq:lsjlsjdfhjld} if and only if $\mathbf{f}\in \mathcal{F}$. Therefore, we can set
\begin{align*}
\Pi[\pi,\sigma_{\ast}]:& & \mathcal{F} &\to (0,\infty)^2\times \R&
\mathbf{f} & \mapsto \text{right hand sides of (\eqref{c1}, \eqref{c2}, \eqref{c3})} \\
\Lambda[\pi,\sigma_{\ast}]:& & \mathcal{F} & \to [1,\infty) &
\mathbf{f} & \mapsto \text{right hand side of \eqref{c4}} \\
\tau_{2-}[\pi,\sigma_{\ast}]:& & \mathcal{F} & \to (-\infty,0) &
\mathbf{f} & \mapsto \text{right hand side of \eqref{c5}}
\end{align*}
Properties (a), (b), (c) and (e) in Proposition \ref{prop:skhdkjhfd} are by construction, where it is understood that the fixed point $\Phi$ of the map $P$, whose domain of definition is $[\tau_{0-},\tau_{0+}]$, has to be restricted to the subinterval $[\tau_{2-},\tau_{2+}]$ to comply with the statement in Proposition \ref{prop:skhdkjhfd} (e). The statements of (e.1), (e.3), (e.4), (e.5) have already been discussed in this proof. Equation $\Phi(\tau_{2+}) = \Phi_{\star}$ in (e.2), with $\Phi_{\star} = \Phi_{\star}(\pi,\mathbf{f},\sigma_{\ast})$,  follows from the fixed point equation $P(\Phi) = \Phi$, see \eqref{eq:kdhskdfhkjdhddffd}, and from $\Phi_1(\tau_{1+}) = \Phi_{\star}$ and $\tau_{0+} = \tau_{1+} = \tau_{2+}$. Equation  $\Phi(\tau_{2-}) = \lambda\, \Phi_{\star}(\pi',\mathbf{f}',\sigma_{\ast})$ in (e.2) with $\mathbf{f}' = (\mathbf{h}',w',q')$ is equivalent to (recall $\mathbf{b}' = \mathbf{a}$)
\begin{subequations}
\begin{align}
\label{ziezuwz1}\alpha_{\mathbf{a}'}[\Phi](\tau_{2-}) & = -\lambda\\
\label{ziezuwz2}\alpha_{\mathbf{a}}[\Phi](\tau_{2-}) & = \lambda\, \tfrac{w'}{1+w'}\\
\label{ziezuwz3}\alpha_{\mathbf{c}'}[\Phi](\tau_{2-}) & = \lambda\, (-w'-\mu') \displaybreak[0]\\
\label{ziezuwz4}\tfrac{1}{\mathbf{h}}\xi_{\mathbf{a}'}[\Phi,\mathbf{h}](\tau_{2-}) & = \log \lambda +  \tfrac{1}{\mathbf{h}'}\big\{-\,\tfrac{1+w'}{1+2w'}(1+\mathbf{h}'\log 2)\big\} \\
\label{ziezuwz5}\tfrac{1}{\mathbf{h}} \xi_{\mathbf{a}}[\Phi,\mathbf{h}](\tau_{2-}) & = \log \lambda + \tfrac{1}{\mathbf{h}'}\big\{-\tfrac{1+w'}{1+2w'}(1+\mathbf{h}'\log 2)\big\}\\
\label{ziezuwz6}\tfrac{1}{\mathbf{h}} \xi_{\mathbf{c}'}[\Phi,\mathbf{h}](\tau_{2-}) & = \log \lambda +  \tfrac{1}{\mathbf{h}'}\big\{-(1+w')q' - \tfrac{w'(1+w')}{1+2w'} - \tfrac{1+3w'+(w')^2}{1+2w'} \mathbf{h}'\log 2\big\}
\end{align} 
\end{subequations}
with $\mu' = (1+w')
\big(
\ub_{1}^2
+ \ub_{2}^2
+ \ub_{3}^2
-2\ub_{2}\ub_{3}
-2\ub_{3}\ub_{1}
-2\ub_{1}\ub_{2}
\big)|_{\beta = \beta[\Phi_{\star}(\pi',\mathbf{f}',\sigma_{\ast})]}$.
By inspection: \eqref{ziezuwz1} follows from \eqref{c4}; \eqref{ziezuwz2} follows from \eqref{c4} and \eqref{c2}; \eqref{ziezuwz5} follows from \eqref{c1}; \eqref{ziezuwz6} follows from \eqref{c3}; \eqref{ziezuwz4} follows from \eqref{ziezuwz5} and the discussion following \eqref{c5}. These five equations and $c[\Phi,\mathbf{h},Z](\tau_{2-}) = 0$ imply \eqref{ziezuwz3}. We have now checked (e.2). We now discuss (d).\\
\emph{Continuity of the maps $\Pi$, $\Lambda$ and $\tau_{2-}$.} 
\newcommand{\subin}{B}
Fix $\mathbf{f}^\Psi = (\mathbf{h}^\Psi,w^\Psi,q^\Psi) \in \mathcal{F}$. Let $r > 0$ and let
$\mathbf{f}^\Upsilon = (\mathbf{h}^\Upsilon,w^\Upsilon,q^\Upsilon)\in \mathcal{F}$ with $\|\mathbf{f}^\Psi-\mathbf{f}^\Upsilon\|_{\R^3}\leq r$. All the objects and abbreviations that have been introduced for a single element of $\mathcal{F}$ before, now come in two versions, one associated to each of $\mathbf{f}^B \in \mathcal{F}$ with $B=\Psi,\Upsilon$. By convention, these two versions are distinguished by a superscript $B$. For instance, 
$\tau_{0+}^B = \tau_{+1}^B = \tau_{2+}^B = \tau_{+1}(\mathbf{f}^B)$ and
$\Phi_0^B = \Phi_0(\pi,\mathbf{f}^B,\sigma_{\ast})|_{[\tau_{0-}^B,\tau_{0+}^B]}$
and $\mathcal{E}^B = \mathcal{E}(\sigma_{\ast};\tau_{0-}^B,\tau_{0+}^B)$ and so forth.
Following this convention, the contraction mapping fixed points are denoted $\Phi^B \in \mathcal{E}^B$.
However, we also write  $\Phi^{\Psi} = \Psi$ and $\Phi^{\Upsilon} = \Upsilon$. Suppose $r \leq \tfrac{1}{2}|q^\Psi-1|$. Then
\begin{equation}\label{eq:khdkhkfdhfdk}
0 \neq \sgn(q^\Psi-1) = \sgn(q^\Upsilon-1)
\end{equation}
Define $\chi:\R \to \R$ by
$\chi(\tau) = \tfrac{\mathbf{h}^\Upsilon}{\mathbf{h}^\Psi}(\tau - \tau_{0+}^\Psi) + \tau_{0+}^\Upsilon$.
Introduce four closed intervals $\mathcal{I}^{\subin} = [\tau_{0-}^{\subin},\tau_{0+}^{\subin}]$, $B = \Psi,\Upsilon$, and
$\mathcal{I}^{\Xi} = [\chi^{-1}(\tau_{0-}^{\Upsilon}),\tau_{0+}^\Psi]$
and $\mathcal{I} = \mathcal{I}^{\Psi}\cap \mathcal{I}^{\Xi}$. Observe that $\chi(\mathcal{I}^{\Xi}) = \mathcal{I}^{\Upsilon}$. 
By Proposition \ref{prop:kdhfkshkds}, the field $\Xi = \Upsilon \circ (\chi|_{\mathcal{I}^{\Xi}})$ satisfies $(\boldsymbol{\mathfrak a},\boldsymbol{\mathfrak b},c)[\Xi,\mathbf{h}^\Psi,Z] = 0$ on $\mathcal{I}^{\Xi}$. 
Recall $\mathcal{J}^B = [\tau_{0-}^B,\tfrac{1}{2}\tau_{1-}^B] \subset \mathcal{I}^B$ and
$|\tau_{2-}^B-\tau_{1-}^B| \leq \tfrac{1}{2}\mathrm{dist}_{\R}(\tau_{1-}^B,\R\setminus \mathcal{J}^B)$,
see \eqref{eq:kdshdkjhkfdhkdfd2} and \eqref{eq:sgggggggsdhs}. 
Set $\mathcal{J} = \mathcal{J}^{\Psi}\cap \mathcal{J}^{\Xi} \subset \mathcal{I}$ with $\mathcal{J}^{\Xi} = \chi^{-1}(\mathcal{J}^{\Upsilon})$.
If $r > 0$ is sufficiently small, then
$$\tau_{2-}^{\Psi}\in \mathcal{J}
\qquad \text{and}\qquad \chi^{-1}(\tau_{2-}^{\Upsilon}) \in \mathcal{J}$$
These inclusions have similar proofs. We only verify $\tau_{2-}^{\Psi}\in \mathcal{J}$. We have $\tau_{2-}^{\Psi}\in \mathcal{J}^\Psi$ and
\begin{align}
\notag |\chi(\tau_{2-}^{\Psi}) - \tau_{1-}^\Upsilon|
& \leq |\chi(\tau_{2-}^{\Psi}) - \chi(\tau_{1-}^\Psi)| + |\chi(\tau_{1-}^\Psi) - \tau_{1-}^{\Upsilon}|\\
\label{eq:lsdkjfdlkfd21} & \leq \tfrac{\mathbf{h}^{\Upsilon}}{\mathbf{h}^{\Psi}} \tfrac{1}{2}\mathrm{dist}_{\R}(\tau_{1-}^{\Psi},\R\setminus \mathcal{J}^{\Psi}) + |\chi(\tau_{1-}^\Psi) - \tau_{1-}^{\Upsilon}|
\end{align}
The right hand side of \eqref{eq:lsdkjfdlkfd21} is a continuous function of 
$\mathbf{f}^{\Upsilon} \in \mathcal{F}$ (with $\mathbf{f}^{\Psi}$ fixed) and is equal to
$\tfrac{1}{2}\mathrm{dist}_{\R}(\tau_{1-}^{\Upsilon},\R\setminus \mathcal{J}^{\Upsilon}) > 0$
when $\mathbf{f}^{\Upsilon} = \mathbf{f}^{\Psi}$. Therefore \eqref{eq:lsdkjfdlkfd21} is $< \mathrm{dist}_{\R}(\tau_{1-}^{\Upsilon},\R\setminus \mathcal{J}^{\Upsilon})$ if $r > 0$ is small enough. Hence $\chi(\tau_{2-}^{\Psi}) \in \mathcal{J}^{\Upsilon}$, that is
$\tau_{2-}^{\Psi}\in \mathcal{J}^{\Xi}$.\\
 Set
$\mathcal{D} = \mathcal{D}^{\Psi} = \mathcal{D}^{\Upsilon} = \mathcal{D}(\sigma_{\ast})$ and
 $\mathcal{E} = \mathcal{E}(\sigma_{\ast};\mathcal{I})$ and $\Phi_0 = \Phi_0(\pi,\mathbf{f}^{\Psi},\sigma_{\ast})|_{\mathcal{I}}$. Equivalently, $\Phi_0 = \Phi_0^{\Psi}|_{\mathcal{I}}$.
 Abbreviate $d_{\mathcal{X}} = d_{\mathcal{X},(\pi,\mathbf{h}^{\Psi})}$ for
 $\mathcal{X} = \mathcal{E},\mathcal{D}$
 and
 $d_{\mathcal{X}^{B}} = d_{\mathcal{X}^{B},(\pi,\mathbf{h}^B)}$ for $B=\Psi,\Upsilon$. By \eqref{eq:kdhkdhjhd}, we have $d_{\mathcal{E}^B}(B,\Phi_0^B) \leq \tfrac{1}{2}\delta^B$ for $B=\Psi,\Upsilon$. If $r > 0$ is sufficiently small, then
\begin{equation}\label{eq.fddfg}
d_{\mathcal{E}}(\Psi|_{\mathcal{I}},\Phi_0) \leq \delta^{\Psi}
\qquad
\text{and}
\qquad
d_{\mathcal{E}}(\Xi|_{\mathcal{I}},\Phi_0) \leq \delta^{\Psi}
\end{equation}
The first follows from 
$d_{\mathcal{E}}(\Psi|_{\mathcal{I}},\Phi_0) \leq
d_{\mathcal{E}^{\Psi}}(\Psi,\Phi_0^{\Psi}) \leq \tfrac{1}{2}\delta^{\Psi}$. The second follows from
\begin{align}
\notag d_{\mathcal{E}}(\Xi|_{\mathcal{I}},\Phi_0)
& \leq d_{\mathcal{E}}(\Upsilon \circ \chi |_{\mathcal{I}},\Phi_0^{\Upsilon}\circ \chi|_{\mathcal{I}})
+ d_{\mathcal{E}}(\Phi_0^{\Upsilon}\circ \chi|_{\mathcal{I}},\Phi_0)\\
\notag & \leq \max\{1,\tfrac{\mathbf{h}^{\Psi}}{\mathbf{h}^{\Upsilon}}\}\, 
d_{\mathcal{E}^{\Upsilon}}(\Upsilon ,\Phi_0^{\Upsilon})
+ d_{\mathcal{E}}(\Phi_0^{\Upsilon}\circ \chi|_{\mathcal{I}},\Phi_0)\\
\label{eq:fddfdljlfd} & \leq \max\{1,\tfrac{\mathbf{h}^{\Psi}}{\mathbf{h}^{\Upsilon}}\}\, \tfrac{1}{2}\delta^{\Upsilon} + d_{\mathcal{E}}(\Phi_0^{\Upsilon}\circ \chi|_{\mathcal{I}},\Phi_0)
\end{align}
and because the right hand side of \eqref{eq:fddfdljlfd} is a continuous function of $\mathbf{f}^{\Upsilon} \in \mathcal{F}$ (with $\mathbf{f}^{\Psi}$ fixed), see \eqref{def:kdshdkhkdhd},  that is equal to $\tfrac{1}{2}\delta^{\Psi}$ when $\mathbf{f}^{\Upsilon} = \mathbf{f}^{\Psi}$.\\
Both $X=\Psi|_{\mathcal{I}}$ and $X=\Xi|_{\mathcal{I}}$ satisfy
$(\boldsymbol{\mathfrak a},\boldsymbol{\mathfrak b},c)[X,\mathbf{h}^\Psi,Z] = 0$ on $\mathcal{I}$,
\begin{subequations} \label{eq:khdkhkddkdkdkdhsl}
\begin{alignat}{4}
A_{\mathbf{a}}[X](\tau) & = &A_{\mathbf{a}}[X](\tau_{0+}^{\Psi})
&+ \textstyle\int_{\tau_{0+}^{\Psi}}^{\tau}\dd \tau'\, \mathbf{I}_1[X,\mathbf{h}^{\Psi},\pi](\tau')\\
\theta_{\mathbf{a}}[X,\mathbf{h}^{\Psi}](\tau) & = &\;\; \theta_{\mathbf{a}}[X,\mathbf{h}^{\Psi}](\tau_{0+}^{\Psi})
&+ \textstyle\int_{\tau_{0+}^{\Psi}}^{\tau}\dd \tau'\, \mathbf{I}_2[X,\mathbf{h}^{\Psi},\pi](\tau')\\
\alpha_{\mathbf{p},\mathbf{a}}[X](\tau) & = & \alpha_{\mathbf{p},\mathbf{a}}[X](\tau_{0+}^{\Psi})
&+ \textstyle\int_{\tau_{0+}^{\Psi}}^{\tau}\dd \tau'\, \mathbf{I}_{(3,\mathbf{p})}[X,\mathbf{h}^{\Psi},\pi](\tau')\\
 \xi_{\mathbf{p},\mathbf{a}}[X,\mathbf{h}^{\Psi}](\tau) & = & \;\;\xi_{\mathbf{p},\mathbf{a}}[X,\mathbf{h}^{\Psi}](\tau_{0+}^{\Psi}) & + \alpha_{\mathbf{p},\mathbf{a}}[X](\tau_{0+}^{\Psi})\,(\tau-\tau_{0+}^{\Psi})\\
\notag   &&  &+ \textstyle\int_{\tau_{0+}^{\Psi}}^{\tau}\dd \tau'' \int_{\tau_{0+}^{\Psi}}^{\tau''}\dd \tau'\, \mathbf{I}_{(3,\mathbf{p})}[X,\mathbf{h}^{\Psi},\pi](\tau') 
\end{alignat}
\end{subequations}
for all $\mathbf{p}\in \{\mathbf{b},\mathbf{c}\}$ and $\tau \in \mathcal{I}$. 
By \eqref{eq:kdshkhfdkhdkfhksad}, 
\eqref{eq.fddfg}, \eqref{eq:khdkhkddkdkdkdhsl}
and by $\sup_{\tau \in \mathcal{I}} |\tau - \tau_{0+}^{\Psi}|\leq 4$, we have
$d_{\mathcal{E}}(\Psi|_{\mathcal{I}},\Xi|_{\mathcal{I}})
\leq 2^3 d_{\mathcal{D}}(\Psi(\tau_{0+}^{\Psi}),\Xi(\tau_{0+}^{\Psi}))
+ 2^{-1}d_{\mathcal{E}}(\Psi|_{\mathcal{I}},\Xi|_{\mathcal{I}})$, and consequently
$$d_{\mathcal{E}}(\Psi|_{\mathcal{I}},\Xi|_{\mathcal{I}}) \leq 2^4 d_{\mathcal{D}}(\Psi(\tau_{0+}^{\Psi}),\Xi(\tau_{0+}^{\Psi})) = 2^4\, d_{\mathcal{D}}\big(
\Phi_{\star}(\pi,\mathbf{f}^{\Psi},\sigma_{\ast}),
\Phi_{\star}(\pi,\mathbf{f}^{\Upsilon},\sigma_{\ast})\big)$$
 In particular, $d_{\mathcal{E}}(\Psi|_{\mathcal{I}},\Xi|_{\mathcal{I}})\to 0$ as $\mathbf{f}^{\Upsilon}\to \mathbf{f}^{\Psi}$.
Furthermore,
\begin{align*}
& d_{\mathcal{D}}\big(\lambda^{\Psi}\,\Phi_{\star}(\pi',\mathbf{f}'^{\,\Psi},\sigma_{\ast}),
\lambda^{\Upsilon}\,\Phi_{\star}(\pi',\mathbf{f}'^{\,\Upsilon},\sigma_{\ast})\big)\\
& = d_{\mathcal{D}}\big(\Psi(\tau_{2-}^{\Psi}),\Xi(\chi^{-1}(\tau_{2-}^{\,\Upsilon}))\big)\\
& \leq 
d_{\mathcal{D}}\big(\Psi(\tau_{2-}^{\Psi}),
\Psi(\chi^{-1}(\tau_{2-}^{\Upsilon}))\big)
+
d_{\mathcal{D}}\big(\Psi(\chi^{-1}(\tau_{2-}^{\Upsilon})),\Xi(\chi^{-1}(\tau_{2-}^{\,\Upsilon}))\big)\\
& \leq 
d_{\mathcal{D}}\big(\Psi(\tau_{2-}^{\Psi}),
\Psi(\chi^{-1}(\tau_{2-}^{\Upsilon}))\big)
+
2^4\, d_{\mathcal{D}}\big(
\Phi_{\star}(\pi,\mathbf{f}^{\Psi},\sigma_{\ast}),
\Phi_{\star}(\pi,\mathbf{f}^{\Upsilon},\sigma_{\ast})\big)
\end{align*}
By the last inequality, if we can show that $\chi^{-1}(\tau_{2-}^{\Upsilon}) \to \tau_{2-}^{\Psi}$
as $\mathbf{f}^{\Upsilon}\to \mathbf{f}^{\Psi}$, then $\tau_{2-}^{\Upsilon} \to \tau_{2-}^{\Psi}$
and $\lambda^{\Upsilon}\to \lambda^{\Psi}$ and $\mathbf{f}'^{\,\Upsilon}\to \mathbf{f}'^{\,\Psi}$. In other words, to show that  $\Pi$, $\Lambda$ and $\tau_{2-}$ are continuous, it suffices to show that
 $\chi^{-1}(\tau_{2-}^{\Upsilon}) \to \tau_{2-}^{\Psi}$ as $\mathbf{f}^{\Upsilon}\to \mathbf{f}^{\Psi}$.\\
By the discussion after \eqref{c2}, we have
$\alpha_{\mathbf{a}}[\Psi](\tau) \geq \tfrac{1}{2}$ and $\alpha_{\mathbf{a},\mathbf{a}'}[\Psi](\tau) \leq 0$
for all $\tau \in \mathcal{J}\subset \mathcal{J}^{\Psi}$. Hence, for all $\tau \in \mathcal{J}$,
$$
\tfrac{\dd}{\dd \tau}\big(\xi_{\mathbf{a}}[\Psi,\mathbf{h}^{\Psi}]
- \xi_{\mathbf{a}'}[\Psi,\mathbf{h}^{\Psi}]\big) = \alpha_{\mathbf{a}}[\Psi] - \alpha_{\mathbf{a}'}[\Psi]
= 2\alpha_{\mathbf{a}}[\Psi] - \alpha_{\mathbf{a},\mathbf{a}'}[\Psi]
\geq 1$$
Hence, $|\tau - \tau_{2-}^{\Psi}| \leq \big|
\xi_{\mathbf{a}}[\Psi,\mathbf{h}^{\Psi}](\tau) - \xi_{\mathbf{a}'}[\Psi,\mathbf{h}^{\Psi}](\tau)\big|$
 if $\tau \in \mathcal{J}$. Set $\tau = \chi^{-1}(\tau_{2-}^{\Upsilon})\in \mathcal{J}$:
\begin{align}
\notag |\chi^{-1}(\tau_{2-}^{\Upsilon}) - \tau_{2-}^{\Psi}| & \leq \big|
\xi_{\mathbf{a}}[\Psi,\mathbf{h}^{\Psi}](\chi^{-1}(\tau_{2-}^{\Upsilon})) - \xi_{\mathbf{a}'}[\Psi,\mathbf{h}^{\Psi}](\chi^{-1}(\tau_{2-}^{\Upsilon}))\big|\\
\label{eq:fdfkhfkjd} & \leq 2\,\slaa{d}_{\mathcal{D},\mathbf{h}^{\Psi}}\big(\Psi(\chi^{-1}(\tau_{2-}^{\Upsilon})),
\Xi(\chi^{-1}(\tau_{2-}^{\Upsilon}))\big)
\end{align}
The last inequality follows from $\xi_{\mathbf{a}}[\Xi,\mathbf{h}^{\Psi}](\chi^{-1}(\tau_{2-}^{\Upsilon}))
= \xi_{\mathbf{a}'}[\Xi,\mathbf{h}^{\Psi}](\chi^{-1}(\tau_{2-}^{\Upsilon}))$
 and the triangle inequality.
Since $\mathbf{f}^{\Upsilon}\to \mathbf{f}^{\Psi}$ implies
  $d_{\mathcal{E}}(\Psi|_{\mathcal{I}},\Xi|_{\mathcal{I}})\to 0$, also the right hand side of 
\eqref{eq:fdfkhfkjd} goes to zero, that is, $|\chi^{-1}(\tau_{2-}^{\Upsilon}) - \tau_{2-}^{\Psi}| \to 0$, as required.\\
\emph{Uniqueness of $\Pi$, $\Lambda$ and $\tau_{2-}$.}
Suppose we have two triples
$\Pi_i$, $\Lambda_i$, $\tau_{2-,i}$ with $i=1,2$. Let $\mathbf{f}\in \mathcal{F}$ and let
$\Phi_i$ be the corresponding fields in (e). By (e.1) and (e.2)
and the local uniqueness for solutions to ODE's, we have $\Phi_1 = \Phi_2$ on the intersection of their domains of definition $[\max\{
\tau_{2-,1}(\mathbf{f}),\tau_{2-,2}(\mathbf{f})
\},\tau_{2+}]$.
Observe that $\tau_-(\mathbf{f}) < \tau_{2-,1}(\mathbf{f}),\tau_{2-,2}(\mathbf{f}) < \tfrac{1}{2}\tau_{1-}(\mathbf{f})$, by (c).
By (e.3), we have $\tau_{2-,1}(\mathbf{f}) = \tau_{2-,2}(\mathbf{f})$. By (e.2),
we have
$\Pi_1(\mathbf{f}) = 
\Pi_2(\mathbf{f})$,\;
$\Lambda_1(\mathbf{f}) = 
\Lambda_2(\mathbf{f})$. \qed
\end{proof}
\begin{remark}
In Proposition \ref{prop:skhdkjhfd}, the signature vector $\sigma_{\ast}$ appears to play a passive role.
However, observe that $\Phi_{\star} = \Phi_{\star}(\pi,\mathbf{f},\sigma_{\ast})$ in (e.2) depends on it in a crucial way, see Definition \ref{def:kdhkhskhkdssPT1111star}. For instance, while $\alpha_{\mathbf{a}}[\Phi_{\star}]$ and $\alpha_{\mathbf{b}}[\Phi_{\star}]$ do not depend on $\sigma_{\ast}$ at all, and $\beta_{\mathbf{i}}[\Phi_{\star}]$, $i=1,2,3$ only in a trivial way through their signs, the component $\alpha_{\mathbf{c}}[\Phi_{\star}]$ does depend on $\sigma_{\ast}$
in a more important way, because the right hand side of \eqref{dshdkhkdhd} does. That $\sigma_{\ast}$ plays a role is not surprising, after all it distinguishes Bianchi VIII and IX.
\end{remark}


%% file: SectionRational.tex
\section{The approximate epoch-to-epoch and era-to-era maps}
\label{sec:dlhjdfll}
This section is logically self-contained, and the notation is introduced from scratch.
Its goal is to study two maps, denoted $\mathcal{Q}_{R}$ and $\mathcal{E}_{R}$, that we informally refer to (following \cite{BKL1}) as the \emph{epoch-to-epoch}
and \emph{era-to-era} maps. The two maps are related, the second is some iterate of the first.
 The subscript $R$ is for \emph{right} (as opposed to \emph{left}).
For the moment, the definition of $\mathcal{Q}_R$ is taken for granted without motivation. To understand its role, see Part 3 of Proposition \ref{prop:kjdhkhkd} and its proof.
\begin{definition}[Epoch-to-epoch map] \label{def:epep}
Set
\begin{align*}
\mathcal{Q}_R:\quad (0,\infty)\setminus \mathbb{Q} & \to (0,\infty)\setminus \mathbb{Q}\\
w & \mapsto  \mathcal{Q}_R(w) =  \begin{cases} \tfrac{1}{w}-1 & \text{if $w < 1$} \\
w-1 & \text{if $w > 1$}
\end{cases}
\end{align*}
For every  $w \in (0,\infty)\setminus \mathbb{Q}$, set
$$
\boldsymbol{\mathcal{Q}}_R\{w\}(q,\mathbf{h}) = \left(\,\frac{\mathrm{num1}}{\mathrm{den}}\,,\,\frac{\mathrm{num2}}{\mathrm{den}}\,\right)
$$
where, if $w < 1$,
\begin{subequations}\label{eq:kdshkhkfdhfd}
\begin{align}
\mathrm{num1} & = 1+w+\mathbf{h} \log 2 - \mathbf{h} (1+2 w) \log (1+\tfrac{1}{w})\\
\mathrm{num2} & = \mathbf{h}\\
\label{eq:kdshkhkfdhfdX1}
\mathrm{den} & = 
(1+w)(1+q+\mathbf{h}\log 2) - \mathbf{h}(2+w) \log(1 + \tfrac{1}{w})
\end{align}
\end{subequations}
and, if $w > 1$,
\begin{subequations}\label{eq:kdshkhkfdhfd2}
\begin{align}
\mathrm{num1} & = 
(1+w)(1 + q + \mathbf{h}\log 2) - \mathbf{h}(2+w) \log(1+\tfrac{1}{w})\\
\mathrm{num2} & = \mathbf{h} w\\
\label{eq:kdshkhkfdhfdX2}
\mathrm{den} & = 1 + w + \mathbf{h}\log 2 - \mathbf{h}(1+2w) \log(1+\tfrac{1}{w})
\end{align}
\end{subequations}
Here, we regard $\boldsymbol{\mathcal{Q}}_R\{w\}$ as a pair of rational functions over $\R$ of degree one in the pair of abstract variables $(q,\mathbf{h})$. Finally, for all  $w \in (0,\infty)\setminus \mathbb{Q}$ and all integers $n\geq 0$, set
\begin{align*}
\mathcal{Q}_R^n(w) & = \big(\underbrace{\mathcal{Q}_R \circ \cdots \circ \mathcal{Q}_R}_{n}\big)(w)\\
\boldsymbol{\mathcal{Q}}_{R}^n\{w\} & = \boldsymbol{\mathcal{Q}}_R\{\mathcal{Q}_R^{n-1}(w)\}\circ \cdots \circ \boldsymbol{\mathcal{Q}}_R\{\mathcal{Q}_R^2(w)\}
\circ \boldsymbol{\mathcal{Q}}_R\{\mathcal{Q}_R(w)\}
\circ \boldsymbol{\mathcal{Q}}_R\{w\}
\end{align*}
Warning: $\boldsymbol{\mathcal{Q}}_R^n\{w\}$ is not the $n$-fold composition of $\boldsymbol{\mathcal{Q}}_R\{w\}$ with itself.
\end{definition}
The goal of this section is to 
understand the bulk behavior of $\boldsymbol{\mathcal{Q}}_R^n\{w\}$ for large $n\geq 0$.
\begin{definition}\label{def:floor} The floor function is $\R \ni x \mapsto \lfloor x\rfloor
= \max\{n\in \Z\,|\, n\leq x\}$.
\end{definition}
\begin{definition}[Era-to-era map]
Define
 $\mathcal{E}_R: (0,1)\setminus \mathbb{Q} \to (0,1)\setminus \mathbb{Q}$ by
$\mathcal{E}_R(w) = \mathcal{Q}_R^{\lfloor 1/w\rfloor}(w)$. For every $w\in (0,1)\setminus \mathbb{Q}$, denote by $\boldsymbol{\mathcal{E}}_R\{w\}$ the pair of rational functions over $\R$ given by
$\boldsymbol{\mathcal{E}}_R\{w\} = \boldsymbol{\mathcal{Q}}_R^{\lfloor 1/w\rfloor}\{w\}$.
Finally, for all $w\in (0,1)\setminus \mathbb{Q}$ and all integers $n\geq 0$, set
\begin{align*}
\mathcal{E}_R^n(w) & = \big(\underbrace{\mathcal{E}_R \circ \cdots \circ \mathcal{E}_R}_{n}\big)(w)\\
\boldsymbol{\mathcal{E}}_{R}^n\{w\} & = \boldsymbol{\mathcal{E}}_R\{\mathcal{E}_R^{n-1}(w)\}\circ \cdots \circ \boldsymbol{\mathcal{E}}_R\{\mathcal{E}_R^2(w)\}
\circ \boldsymbol{\mathcal{E}}_R\{\mathcal{E}_R(w)\}
\circ \boldsymbol{\mathcal{E}}_R\{w\}
\end{align*}
\end{definition}
\begin{lemma}
For all integers $m,n\geq 0$,
\begin{itemize}
\item
$\boldsymbol{\mathcal{Q}}_R^{m+n}\{w\} = \boldsymbol{\mathcal{Q}}^m_R\{\mathcal{Q}_R^n(w)\}\circ \boldsymbol{\mathcal{Q}}_R^n\{w\}$
for $w\in (0,\infty)\setminus \mathbb{Q}$
\item 
$\boldsymbol{\mathcal{E}}_R^{m+n}\{w\} = \boldsymbol{\mathcal{E}}^m_R\{\mathcal{E}_R^n(w)\}\circ \boldsymbol{\mathcal{E}}_R^n\{w\}$
for $w\in (0,1)\setminus \mathbb{Q}$
\end{itemize}
\end{lemma}
\newcommand{\spind}{r}
\begin{proposition}
 Let $w \in (0,1)\setminus \mathbb{Q}$.
Then, for every integer $1\leq \spind \leq \lfloor \tfrac{1}{w}\rfloor$, \begin{equation}\label{eq:dhjkdhdkhdfkddf}
\boldsymbol{\mathcal{Q}}_R^\spind\{w\}
(q,\mathbf{h}) = \left(\frac{\mathrm{num1}_\spind}{\mathrm{den}_\spind},\frac{\mathrm{num2}_\spind}{\mathrm{den}_\spind}\right)
\end{equation}
where
\begin{subequations} \label{eq:dkfdhkfjhkjdhkjdhf}
\begin{align}
\begin{split}
\mathrm{num1}_\spind & = 
(1+w)\big(\spind + \spind q - q\big) + \mathbf{h}\,A_1(w,\spind)
\end{split}\\
\begin{split}
\mathrm{num2}_\spind & =
\mathbf{h} \big(w + 1 - w\spind\big)
\end{split}\\
\begin{split}
\mathrm{den}_\spind & = (1+w)(1+q) + \mathbf{h}\, A_2(w,\spind)
\end{split}
\end{align}
\end{subequations}
and where
\begin{subequations} \label{eq:kdhkfdhkdhkfdh}
\begin{align}
\label{eq:kdhkfdhkdhkfdh1}
\begin{split}
A_1(w,\spind) & = 
\big(2\spind - 1 + w\spind - w\big)\log 2 - \big(2\spind -1 + w\spind + w  \big)\log(1+\tfrac{1}{w})\\
& \quad + \textstyle\sum_{k=1}^{\spind -1}
\big(1 + 2k - 2k^2 w -w\big) \log\big(1 + \tfrac{w}{1-kw}\big)\\
& \quad + \spind\,\textstyle\sum_{k=1}^{\spind -1}
\big(
(2k-1)w-2 \big)
\log\big(1 + \tfrac{w}{1-kw}\big)
\end{split}\\
\label{eq:kdhkfdhkdhkfdh2}
\begin{split}
A_2(w,\spind) & = 
\left(1 + w\spind\right)\log 2 
- (2+w) \log(1 + \tfrac{1}{w})\\
& \quad + \textstyle\sum_{k=1}^{\spind -1} \big((2k-1)w - 2\big) \log\big(
1 + \frac{w}{1-kw}\big)
\end{split}
\end{align}
\end{subequations}
Furthermore,
$\mathcal{E}_R(w) = \tfrac{1}{w}-\lfloor \tfrac{1}{w}\rfloor$, that is,
$\mathcal{E}_R$ is the Gauss map, and 
$$
\boldsymbol{\mathcal{E}}_R\{w\}(q,\mathbf{h}) = \left(\,\frac{\mathrm{num1}_{\lfloor 1/w\rfloor }}{\mathrm{den}_{\lfloor 1/w\rfloor }}\,,\,\frac{\mathrm{num2}_{\lfloor 1/w\rfloor }}{\mathrm{den}_{\lfloor 1/w\rfloor }}\,\right)
$$
\end{proposition}
\begin{remark}
In equation \eqref{eq:kdhkfdhkdhkfdh}, we have $0 \leq \tfrac{w}{1-kw} \leq 1$ for all $1 \leq k \leq \spind-1$.
\end{remark}
\begin{proof} Let $w \in (0,1)\setminus \mathbb{Q}$.  We show \eqref{eq:dhjkdhdkhdfkddf} by induction
over $1 \leq \spind \leq \lfloor \tfrac{1}{w}\rfloor$.  
The $\spind=1$ base case of the induction argument, $\boldsymbol{\mathcal{Q}}_R\{w\}(q,\mathbf{h}) = (\mathrm{num1}_1/\mathrm{den}_1,\mathrm{num2}_1/\mathrm{den}_1)$, is by direct inspection, using \eqref{eq:kdshkhkfdhfd}.
The induction step becomes the identity
\begin{equation}\label{eq:kdshkjhkjdhfkhfdfd}
\boldsymbol{\mathcal{Q}}_R\{\mathcal{Q}_R^{\spind-1}(w)\}  
\left(\tfrac{\mathrm{num1}_{\spind-1}}{\mathrm{den}_{\spind-1}},\tfrac{\mathrm{num2}_{\spind-1}}{\mathrm{den}_{\spind-1}}\right) =
\left(\tfrac{\mathrm{num1}_\spind}{\mathrm{den}_\spind},\tfrac{\mathrm{num2}_\spind}{\mathrm{den}_\spind}\right)
\end{equation}
for all $2\leq  \spind \leq \lfloor \tfrac{1}{w}\rfloor$.
To calculate $\boldsymbol{\mathcal{Q}}_R\{\mathcal{Q}_R^{\spind-1}(w)\}(\,\cdot\,)$, use formulas \eqref{eq:kdshkhkfdhfd2}, since
$\mathcal{Q}_R^{\spind-1}(w) = \tfrac{1}{w} - \spind + 1 > 1$.
Observe that \eqref{eq:kdshkjhkjdhfkhfdfd} follows from the identities
\begin{align*}
\lambda\; \mathrm{num1}_\spind & = 
\big(1 + (\tfrac{1}{w}-\spind+1)\big)\big(\mathrm{den}_{\spind-1} + \mathrm{num1}_{\spind-1} + \mathrm{num2}_{\spind-1}\log 2\big)\\
& \quad 
- \mathrm{num2}_{\spind-1} \big(2 + (\tfrac{1}{w}-\spind+1)\big) \log(1 + \tfrac{w}{1-(\spind-1)w})\\
\lambda\; \mathrm{num2}_\spind & = \mathrm{num2}_{\spind-1}\,\big(\tfrac{1}{w}-\spind+1\big)\\
\lambda\; \mathrm{den}_\spind & = \mathrm{den}_{\spind-1}\big(1 + (\tfrac{1}{w}-\spind+1)\big)
+ \mathrm{num2}_{\spind-1}\log 2 \\
& \quad - \mathrm{num2}_{\spind-1} \big(1 + 2(\tfrac{1}{w}-\spind+1)\big)
\log(1 + \tfrac{w}{1-(\spind-1)w})
\end{align*}
where $\lambda = 2 + \tfrac{1}{w} - \spind> 2$. To verify each of these identities, divide both sides by $\lambda$, and use
$\mathrm{num2}_{\spind-1} = \mathbf{h}w\lambda$, to obtain the equivalent identities
\begin{align*}
\mathrm{num1}_\spind & = 
\mathrm{den}_{\spind-1} + \mathrm{num1}_{\spind-1} + \mathrm{num2}_{\spind-1}\log 2
 \\ & \quad 
- \mathbf{h}(3w+1-\spind w) \log(1 + \tfrac{w}{1-(\spind-1)w})\\
\mathrm{num2}_\spind & = \mathbf{h}(1-\spind w+w)\\
\mathrm{den}_\spind & = \mathrm{den}_{\spind-1}
+ \mathbf{h} w \log 2 
- \mathbf{h}(3w+2-2\spind w)
\log(1 + \tfrac{w}{1-(\spind-1)w})
\end{align*}
The last three identities are verified directly. \qed
\end{proof}
The following lemma will be used later.
\begin{lemma} \label{lem:kfdhkfdhf}
For every $w\in (0,1) \setminus \mathbb{Q}$ and every integer $\spind$ with $1\leq \spind \leq \lfloor \tfrac{1}{w}\rfloor$,
\begin{alignat*}{4}
 0 & \;\;\leq\;\; A_1(w,\spind) - \spind A_2(w,\spind) + \log 2 && \;\;\leq\;\; 6\tfrac{1}{w}\\
- 8 \log(1 + \tfrac{1}{w}) & \;\;\leq\;\; \hskip15mm A_2(w,\spind) && \;\;\leq\;\; 0
\end{alignat*}
Here, $A_1$ and $A_2$ are defined by \eqref{eq:kdhkfdhkdhkfdh}.
\end{lemma}
\begin{proof}
Observe that $\spind w \leq 1$. Calculate
\begin{multline*}
\big\{A_1(w,\spind) - \spind A_2(w,\spind) + \log 2\big\}\,w\;=\; w^2(\spind-1) \log 2 + \spind w (1-\spind w) \log 2\\
+ (1-w)\,w\log(1 + \tfrac{1}{w}) +
w \textstyle \sum_{k=1}^{\spind-1} \big(2k(1-kw) + (1-w)\big)\, \log(1 + \tfrac{w}{1-kw})
\end{multline*}
By inspection, the right hand side is non-negative, and bounded by
$$
\leq 3 + \tfrac{1}{\spind} \textstyle\sum_{k=1}^{\spind-1} \big(2k(1-kw) + (1-w)\big)\tfrac{w}{1-kw} \leq 6
$$
We have $A_2(w,\spind) \leq 0$, because $\spind w \leq 1$, the sum of the first two terms on the right hand side of \eqref{eq:kdhkfdhkdhkfdh2} is non-positive, and the third term is non-positive. Estimate
\begin{align*}
|A_2(w,\spind)|  &\leq (2+w) \log (1 + \tfrac{1}{w}) + 2 \textstyle\sum_{k=1}^{\spind-1} \log(1 + \tfrac{w}{1-kw})\\
& \leq 3 \log (1 + \tfrac{1}{w}) + 2\hskip-1pt\textstyle\sum_{k=1}^{\spind-1} \tfrac{w}{1-kw} \leq 3 \log (1 + \tfrac{1}{w}) + 2\big(1 + \textstyle\int_0^{\spind-1} \dd k \tfrac{w}{1-kw}\big)\\
& \leq 3 \log (1 + \tfrac{1}{w}) + 2\big(1 - \log (1-(\spind-1)w)\big) \leq 8 \log (1 + \tfrac{1}{w})
\end{align*}
since $2<3\log 2 \leq 3 \log(1+\tfrac{1}{w})$ and $-\log(1-\spind w+w) \leq - \log w = \log \tfrac{1}{w} \leq \log(1 + \tfrac{1}{w})$.
\qed
\end{proof}
\begin{proposition} \label{prop:jdhkhkds} For every $w\in (0,1)\setminus \mathbb{Q}$, every $p > 0$ and every integer $1\leq \spind \leq \lfloor \tfrac{1}{w}\rfloor$, let $(\mu',\nu')$ be the pair of rational functions over $\R$ in the pair of abstract variables $(\mu,\nu)$ given implicitly by
$$\Big(p' + \frac{\mu'}{\nu'},\;\frac{1+w'}{\nu'}\Big) = \boldsymbol{\mathcal{Q}}_R^\spind\{w\}\,\Big(p + \frac{\mu}{\nu},\, \frac{1+w}{\nu}\Big)$$
where $w' = \mathcal{Q}_R^\spind(w) = \tfrac{1}{w}-\spind$ and $p' = \spind - p/(1+p)$, that is $(p',0) = \boldsymbol{\mathcal{Q}}_R^\spind\{w\}(p,0)$. Then $\mu'$ is actually a linear polynomial over $\R$ in $\mu$, and $\nu'$ is actually a linear polynomial over $\R$ in the pair $(\mu,\nu)$. Explicitly
\begin{equation}\label{eq:fdhkdhkfd}
\begin{pmatrix}
\mu'\\
\nu' \end{pmatrix}
= \frac{1}{w} \begin{pmatrix} - \tfrac{1}{1+p} & 0\\ 1 & 1 + p \end{pmatrix}\begin{pmatrix}\mu\\ \nu\end{pmatrix}
+ \frac{1}{w}\begin{pmatrix} A_1(w,\spind) - p'\, A_2(w,\spind) \\ A_2(w,\spind) \end{pmatrix}
\end{equation}
The first and second entries of the vector
\begin{equation}\label{eq:dhdkjhdkd}
\frac{1}{w}\begin{pmatrix} A_1(w,\spind) - p'\, A_2(w,\spind) \\ A_2(w,\spind) \end{pmatrix}
\end{equation}
are bounded in absolute value by $\leq 2^4 (\tfrac{1}{w})^2$
and $\leq 2^3 \tfrac{1}{w}\log (1 + \tfrac{1}{w})$, respectively.
\end{proposition}
\begin{proof}
Equation \eqref{eq:fdhkdhkfd} follows from equation \eqref{eq:dhjkdhdkhdfkddf}.
To check the bounds, observe that
$$A_1(w,\spind) - p' A_2(w,\spind) =  \Big(A_1(w,\spind) - \spind\, A_2(w,\spind) + \log 2\Big)
- \log 2 + \tfrac{p}{1+p}\, A_2(w,\spind)$$
Now, use Lemma \ref{lem:kfdhkfdhf} and $\log 2 \leq \tfrac{1}{w}$ and $\log(1+\tfrac{1}{w}) \leq \tfrac{1}{w}$.
. \qed
\end{proof}
\begin{definition} \label{def:infcont}
For 
every sequence of strictly positive integers $(k_n)_{n\geq 0}$, we denote the associated infinite continued fraction by
$$\langle k_0,k_1,\ldots\rangle \;=\; \frac{1}{k_0 + \frac{1}{k_1 + \ldots}}
\qquad \in \qquad \big(\tfrac{1}{k_0+1},\tfrac{1}{k_0} \big)\setminus \mathbb{Q}$$
Every element of $(0,1)\setminus \mathbb{Q}$ has a unique continued fraction expansion of this form.
\end{definition}
We now show that when $\mathbf{h}=0$, the era-to-era maps can be realized as a left-shift operator
on two-sided sequences of positive integers.
\newcommand{\approxq}{p}
\begin{proposition}\label{prop:kshjajjjajsdkds}
Fix any two-sided sequence $(k_n)_{n\in \Z}$ of strictly positive integers 
and define two-sided sequences $(\approxq_n)_{n \in \Z}$ and $(w_n)_{n \in \Z}$ by
\begin{equation}\label{eq:kdkdhkjfd}
\tfrac{1}{1+\approxq_n}=
\langle k_n,k_{n-1},k_{n-2},\ldots\rangle
\qquad
w_n =
\langle k_{n+1},k_{n+2},k_{n+3}\ldots\rangle
\end{equation}
Then $w_{n+1} = \mathcal{E}_R(w_n)$ and $(p_{n+1},0) = \boldsymbol{\mathcal{E}}_R\{w_n\}(p_n,0)$ for all
$n \in \Z$, and  $\mathcal{E}_R^n(w_0) = w_{n}$ and 
$\boldsymbol{\mathcal{E}}_R^n\{w_0\}(\approxq_0,0)  = (\approxq_n,0)$ for all $n\geq 0$.
\end{proposition}
\begin{proof}
Use $\mathcal{E}_R(w) = \tfrac{1}{w}-\lfloor \tfrac{1}{w}\rfloor$ and
$\boldsymbol{\mathcal{E}}_R\{w\}(p,0) = (\lfloor \tfrac{1}{w}\rfloor - 1 + \tfrac{1}{1+p},\,0)$. \qed
\end{proof}
\begin{definition}\label{def:kdhkdhkhfd}
Fix any two-sided sequence $(k_n)_{n\in \Z}$ of strictly positive integers and define $(\approxq_n)_{n \in \Z}$ and $(w_n)_{n\in \Z}$ by \eqref{eq:kdkdhkjfd}. For every integer $n\geq 0$, let $(\mu_n,\nu_n)$ be
the pair of linear polynomials over $\R$ 
in the abstract variables $(\mu_0,\nu_0)$, with coefficients depending only on the fixed sequence $(k_n)_{n\in \Z}$,
given implicitly by
 \begin{subequations}
\begin{align}
\Big(\approxq_n + \frac{\mu_n}{\nu_n},\,\frac{1+w_{n}}{\nu_n}\Big)
& = \boldsymbol{\mathcal{E}}_R^n\{w_0\}\Big(\approxq_0 + \frac{\mu_0}{\nu_0},\,\frac{1+w_0}{\nu_0}\Big)
\intertext{or by the equivalent recursive prescription}
\label{eq:kdhkdhkj}
\Big(\approxq_{n+1} + \frac{\mu_{n+1}}{\nu_{n+1}},\,\frac{1+w_{n+1}}{\nu_{n+1}}\Big) & = \boldsymbol{\mathcal{E}}_R\{w_{n}\}\Big(\approxq_n + \frac{\mu_n}{\nu_n},\,\frac{1+w_{n}}{\nu_n}\Big)
\end{align}
\end{subequations}
By Proposition \ref{prop:jdhkhkds}, equation \eqref{eq:kdhkdhkj} is $V_{n+1} = X_n V_n + Y_n$, where $V_n = (\mu_n,\nu_n)^T$ and
\begin{align*}
X_n & = \frac{1}{w_{n}} \begin{pmatrix}
- \tfrac{1}{1+\approxq_n} & 0 \\
1 & 1 + \approxq_n
\end{pmatrix} &
Y_n & =
\frac{1}{w_{n}} \begin{pmatrix}
A_1(w_{n}) - \approxq_{n+1} A_2(w_{n})\\
A_2(w_{n})
\end{pmatrix}
\end{align*}
Here, $A_1(w) = A_1(w,\lfloor \tfrac{1}{w}\rfloor)$ and 
$A_2(w) = A_2(w,\lfloor \tfrac{1}{w}\rfloor)$, see equations \eqref{eq:kdhkfdhkdhkfdh}.
\end{definition}
\begin{example} We consider Definition \ref{def:kdhkdhkhfd} when $k_n = 1$ for all $n\in \Z$. Then $w_n = \approxq_n = w$ for all $n\in \Z$, where $w = \tfrac{1}{2}(\sqrt{5}-1) \in (0,1)\setminus \mathbb{Q}$. We have $w^2+w-1=0$ and $\lfloor \tfrac{1}{w}\rfloor = 1$ and
$$
X_n = \begin{pmatrix}
-1 && 0 \\
1+w &\,& 2+w
\end{pmatrix}
\qquad
Y_n = \begin{pmatrix} -2 \log (1+w) \\
(2+w) \log 2 - (6+4w) \log (1+w)
\end{pmatrix}
$$
for all $n\geq 0$. It follows that $\mu_{n+2} = \mu_n$ for all $n\geq 0$, that is, $\mu_{2n} = \mu_0$ and
$\mu_{2n+1} = -\mu_0 - 2\log(1+w)$. There are unique $\lambda_1=\lambda_1(\mu_0)$ and $\lambda_2 = \lambda_2(\mu_0)$, depending only on $\mu_0$, such that $\nu_{2n+2}-\lambda_1 = (2+w)^2(\nu_{2n}-\lambda_1)$ and
$\nu_{2n+3}-\lambda_2 = (2+w)^2(\nu_{2n+1}-\lambda_2)$. That is,
$\nu_{2n} = (2+w)^{2n} (\nu_0 - \lambda_1) + \lambda_1$ and
$\nu_{2n+1} = (2+w)^{2n}(\nu_1-\lambda_2) + \lambda_2$. 
Here, $\nu_1 = (2+w)\nu_0 + (1+w)\mu_0 + (2+w)\log 2 - (6+4w)\log (1+w)$.
\end{example}
\begin{definition}[Propagator] \label{def:propagator}
Let $(p_n)_{n\in \Z}$, $(w_n)_{n\in \Z}$, $(X_n)_{n \geq 0}$ be as in Definition \ref{def:kdhkdhkhfd}.
Then for all integers $n\geq m \geq 0$, let $P_{n,m}
= X_{n-1}\cdots X_m$.
Explicitly,
$$
P_{n,m} = 
\begin{pmatrix}
a_{n-1}\cdots a_m &\;\;& 0 \\
\sum_{\ell=m}^{n-1} x_{\ell}
&& c_{n-1}\cdots c_m
\end{pmatrix}
$$
where $x_{\ell} = c_{n-1}\cdots  c_{\ell+1} b_{\ell} a_{\ell-1}\cdots a_m$ 
whenever $n-1 \geq \ell \geq m$, and for all $\ell\geq 0$,
$$ X_{\ell} = \begin{pmatrix}
a_{\ell} & 0 \\
b_{\ell} & c_{\ell}
\end{pmatrix}
\qquad
a_{\ell} = \frac{-1}{w_{\ell}(1+\approxq_{\ell})}
\qquad
b_{\ell} = \frac{1}{w_{\ell}}
\qquad
c_{\ell} = \frac{1+\approxq_{\ell}}{w_{\ell}}
$$
In this definition, a sequence of dots\,\,$\cdots$ indicates that indices increase towards the left, one by one.
A product of the form $F_k \cdots F_j$ is equal to one if $k = j - 1$.
In particular, $P_{n,n} = \big(\begin{smallmatrix} 1 & 0 \\ 0 & 1 \end{smallmatrix}\big)$.
\end{definition}
\begin{lemma} \label{lemme:1}
In the context of  Definition \ref{def:kdhkdhkhfd}, we have
$V_n = P_{n,0}V_0 + \sum_{\ell=0}^{n-1}P_{n,\ell+1} Y_{\ell}$.
\end{lemma}
\begin{lemma} \label{lem:kdskjshkshks}
Recall Definition \ref{def:propagator}. For all integers $n \geq m \geq 0$, we have
\begin{subequations}
\begin{alignat}{4}
\label{eq:shdkhdkhdfk1} \tfrac{1}{2}\;  & \leq \hskip4mm \tfrac{w_{n-1}}{w_{m-1}} (-1)^{m+n}
a_{n-1} \cdots a_m && \leq \;2 \\
\label{eq:shdkhdkhdfk2} 
(1-\delta_{mn})\; \tfrac{1}{4}\;
 & \leq 
\hskip4mm
\tfrac{w_{n-1}}{w_{m-1}^2} \big( w_{n-2}\cdots w_{m-1}\big)^2\,
\textstyle\sum_{\ell=m}^{n-1} 
x_{\ell}
&& \leq
\;2\\
\label{eq:shdkhdkhdfk3} \tfrac{1}{2}\;
& \leq \hskip4mm
\tfrac{w_{n-1}}{w_{m-1}} \big( w_{n-2}\cdots w_{m-1}\big)^2\,
c_{n-1}\cdots c_m \;\; && \leq \;2
\end{alignat}
\end{subequations}
Moreover,
\begin{equation}\label{eq:ljfljdflfd}
w_{n-2}\cdots w_{m-1} \leq |\gold_-|^{n-m-1} = \gold_+^{-n+m+1}
\qquad \text{when}\qquad  n\geq m \geq 0
\end{equation}
Here, $\gold_{\pm} = \tfrac{1}{2}(1\pm \sqrt{5})$ are the roots of the polynomial $\gold^2-\gold-1$. Observe that
 $|\gold_-|<1$.
In this lemma, a sequence of dots\,\,$\cdots$ indicates that indices increase towards the left, one by one.
A product of the form $F_k \cdots F_j$ is equal to one if $k = j - 1$.
\end{lemma}
\begin{proof}
In this proof, abbreviate $v_{\ell} = 1/(1 + \approxq_{\ell+1}) = \langle k_{\ell+1},k_{\ell},k_{\ell-1} \ldots \rangle$. We have
\begin{subequations}
\begin{align}
\label{eq:khkjfdhkfd1} (-1)^{m+n} a_{n-1}\cdots a_m &= \frac{v_{n-2}\cdots v_{m-1}}{w_{n-2}\cdots w_{m-1}}\,\cdot\, \frac{w_{m-1}}{w_{n-1}}\\
\label{eq:khkjfdhkfd2} x_m = c_{n-1}\cdots c_{m+1} b_m & =  \frac{w_{n-2}\cdots w_{m}}{v_{n-2}\cdots v_{m}} \,\cdot\, \frac{w_{m-1}^2}{ w_{n-1}}\,\cdot\, \Big(\frac{1}{w_{n-2}\cdots w_{m-1}}\Big)^2  \\
\label{eq:khkjfdhkfd3} c_{n-1}\cdots c_m & =  \frac{w_{n-2}\cdots w_{m-1}}{v_{n-2}\cdots v_{m-1}} \,\cdot\, \frac{w_{m-1}}{ w_{n-1}}\,\cdot\, \Big(\frac{1}{w_{n-2}\cdots w_{m-1}}\Big)^2
\end{align}
\end{subequations}
where $n\geq m$ in \eqref{eq:khkjfdhkfd1} and
\eqref{eq:khkjfdhkfd3} and $n > m$ in \eqref{eq:khkjfdhkfd2}. Each right hand side is written as a product of positive quotients, whose first factor is contained in the closed interval $[\tfrac{1}{2},2]$, see Proposition \ref{prop:kfdhkjhdfkdf} (a) of Appendix \ref{app:cfprod}.
This implies \eqref{eq:shdkhdkhdfk1} and \eqref{eq:shdkhdkhdfk3}. If
$n=m$, the sum in \eqref{eq:shdkhdkhdfk2}
vanishes and the estimate is trivial. Suppose $n > m$.
We have $\sgn x_{\ell} = (-1)^{\ell+m}$.
If, in addition,  $\ell$ satisfies $n-2 \geq \ell \geq m$, we have
$|x_{\ell+1}|/|x_{\ell}| = v_{\ell}v_{\ell-1}
= v_{\ell}(v_{\ell}^{-1} - \lfloor v_{\ell}^{-1}\rfloor)
\leq \tfrac{1}{2}$, that is $ \tfrac{1}{2}|x_{\ell}| \geq |x_{\ell+1}|$.
Therefore, the alternating sum in  \eqref{eq:shdkhdkhdfk2} is non-negative and bounded from above by its first summand $x_m > 0$
and from below by $x_m + x_{m+1} \geq x_m - |x_{m+1}| \geq \tfrac{1}{2}x_m$.
Actually, $x_{m+1}$ is only defined when $n\geq m+2$, but $\tfrac{1}{2}x_m$ is a lower bound for all $n\geq m+1$. Now, estimate $x_m$, which is the left hand side of \eqref{eq:khkjfdhkfd2}.\\
Inequality \eqref{eq:ljfljdflfd} is a consequence of Proposition \ref{prop:kfdhkjhdfkdf} (b).
\qed
\end{proof}
{\bf Warning:} In the next proposition, the sequences $(w_j)_{j\in \Z}$ and $(p_j)_{j \in \Z}$ do not 
have the property that $w_j$ and $(1+p_j)^{-1}$ always lie in $(0,1)\setminus \mathbb{Q}$. Rather, they lie in  $(0,\infty)\setminus \mathbb{Q}$.
However, in the proof of Proposition \ref{prop:kjdhkhkd}, the auxiliary sequences
$(w^{\ast}_n)_{n\in\Z}$ and $(p^{\ast}_n)_{n\in \Z}$ do have the property
that $w^{\ast}_n$ and $(1+p^{\ast}_n)^{-1}$ always lie in $(0,1)\setminus \mathbb{Q}$.
The discussion beginning with
Proposition \ref{prop:kshjajjjajsdkds} and ending just above will be applied to the auxiliary sequences.
\begin{proposition} \label{prop:kjdhkhkd}
For all $w_0\in (0,1)\setminus \mathbb{Q}$ and
 $q_0 \in (0,\infty) \setminus \mathbb{Q}$, introduce
\begin{itemize}
\item a two sided sequence of strictly positive integers
$(k_n)_{n \in \Z}$ by
$$
(1+q_0)^{-1} = \langle k_0,k_{-1},k_{-2},\ldots \rangle
\qquad
w_0 = \langle k_1,k_2,k_3,\ldots\rangle
$$
\item (Era Pointer) $J: \Z \to \Z$ by $J(0) = 0$ and $J(n+1) = J(n) + k_{n+1}$
\item (Era Counter) $N: \Z \to \Z$ by $N(0) = 0$ and $N(j+1) = N(j) + \chi_{J(\Z)}(j)$, where $\chi_{J(\Z)}$ is the characteristic function of the image $J(\Z)\subset \Z$; equivalently
\begin{equation}\label{eq:minim}
N(j) = \min\{n \in \Z\,|\,J(n) \geq j \}
\end{equation}
\item sequences $(w_j)_{j \in \Z}$ and $(p_j)_{j \in \Z}$ by (observe that $w_0$ is defined consistently)
\begin{subequations}
\begin{align}
w_j & = \langle k_{N(j)+1},k_{N(j)+2},\ldots \rangle + J\big(N(j)\big) - j\\
\label{eq:pjpjpj}
p_j & = \langle k_{N(j)-1},k_{N(j)-2},\ldots \rangle + k_{N(j)} + j  - J\big(N(j)\big) - 1
\end{align}
\end{subequations}
\end{itemize}
{\bf Part 1.} Then $p_0 = q_0$ and $w_j,p_j > 0$ and
$\mathcal{Q}_R(w_j) = w_{j+1}$ and $\boldsymbol{\mathcal{Q}}_R\{w_j\}(p_j,0) = (p_{j+1},0)$
for all $j \in \Z$, and  $\mathcal{Q}_R^j(w_0) = w_j$ and $\boldsymbol{\mathcal{Q}}_R^j\{w_0\}(q_0,0) = (p_j,0)$
for all $j \geq 0$.\\
{\bf Part 2.} Let $\gold_+ = \tfrac{1}{2}\big(1 + \sqrt{5}\,\big)$ and set 
$$\mathbf{C}(w_0,q_0) = \textstyle\sup_{n\geq 0} \,(n+1) \gold_+^{-2n} k_n\, \max\{k_{n-1},k_{n-2}\}\;\;\;\in \;\;\;[1,\infty]$$
Suppose $\mathbf{C}(w_0,q_0) < \infty$. Fix any $0 < \mathbf{h}_0 \leq 2^{-14}(\mathbf{C}(w_0,q_0))^{-1}$.
Then, there are sequences $(q_j)_{j\geq 0}$, $(\mathbf{h}_j)_{j\geq 0}$ 
of real numbers
such that for every $j \geq 0$, the denominator appearing in the pair of rational functions
$\boldsymbol{\mathcal{Q}}_R\{w_j\}$, given by \eqref{eq:kdshkhkfdhfdX1} or \eqref{eq:kdshkhkfdhfdX2},
is strictly positive at $(q_j,\mathbf{h}_j)$, and
$$(q_{j+1},\mathbf{h}_{j+1}) = \boldsymbol{\mathcal{Q}}_R\{w_j\}(q_j,\mathbf{h}_j)$$
or $(q_j,\mathbf{h}_j) = \boldsymbol{\mathcal{Q}}_R^j\{w_0\}(q_0,\mathbf{h}_0)$. For all $j \geq 0$,
\begin{itemize}
\item $0 < \mathbf{h}_j \leq 2^6\,\mathbf{h}_0\, \gold_+^{-2N(j)}$ and
$$\frac{1}{4} \;\;\leq\;\; \frac{\mathbf{h}_j}{\mathbf{h}_0}\,\frac{1+w_0}{1+w_j}\,\prod_{\ell=0}^{N(j)-1}
\frac{1}{w_{J(\ell)} w_{J(\ell-1)}} \;\;\leq\;\; 4$$
\item $q_j \in (0,\infty)\setminus \Z$
and $|q_j - p_j| \leq 2^{12}\, \mathbf{h}_0\, N(j)\, \gold_+^{-2N(j)} k_{N(j)}$
\item 
\rule{0pt}{11pt}$q_j \in (0,1)$ if and only if $p_j \in (0,1)$ if and only if $j-1 \in J(\Z)$
\item
\rule{0pt}{11pt}$\max\{\tfrac{1}{w_j},w_j,\tfrac{1}{q_j},\tfrac{1}{|q_j-1|},q_j\}
\leq 2^4 \max\{k_{N(j)-2},k_{N(j)-1},k_{N(j)},k_{N(j)+1}\}$
\end{itemize}
{\bf Part 3.} Let the map $\mathcal{Q}_L: (0,\infty)^3\to (0,\infty)^2\times \R$ be given as in Definition \ref{def:kdhkhskhkdssPT2}. Then the sequences
$(\mathbf{h}_j)_{j\geq 0}$, $(w_j)_{j\geq 0}$, $(q_j)_{j\geq 0}$ in Part 2 satisfy for all $j\geq 0$:
$$(\mathbf{h}_j,w_j,q_j) = \mathcal{Q}_L(\mathbf{h}_{j+1},w_{j+1},q_{j+1})$$
\end{proposition}
\begin{example}
In Part 1 of
Proposition \ref{prop:kjdhkhkd}, suppose the continued fraction expansions begin as follows:
$(1+q_0)^{-1} = \langle 1,2,\ldots\rangle$ and $w_0 = \langle 3,1,2,4\ldots \rangle$. Then,
\newcommand{\YES}{\hskip2mm 1 \hskip2mm}
\newcommand{\NO}{\hskip2mm 0 \hskip2mm}
$$\begin{array}{c | c c c c c c c c c c c c c c}
j & -3 & -2 & -1 & 0 & 1 & 2 & 3 & 4 & 5 & 6 & 7 & 8 & 9 & 10\\
\hline
\chi_{J(\Z)}(j) & \YES & \NO & \YES & \YES &\NO  & \NO & \YES & \YES & \NO & \YES & \NO & \NO & \NO & \YES \\
N(j) & -2 & -1 & -1 & 0 & 1 & 1 & 1 & 2 & 3 & 3 & 4 & 4 & 4 & 4\\
J(N(j))   & -3 & -1 & -1 & 0 & 3 & 3 & 3 & 4 & 6 & 6& 10 & 10 & 10 & 10
\end{array}$$
\end{example}
\begin{proof}[of Proposition \ref{prop:kjdhkhkd}] Two basic properties of $J$ and $N$ are, for all $j \in \Z$:
\begin{itemize}
\item $N \circ J$ is the identity; consequently $J(N(j)) = j$ if and only if $j \in J(\Z)$
\item $J(N(j)) \geq j$ and $J(N(j)-1) \leq j-1$ by \eqref{eq:minim}; consequently
\begin{equation}\label{eq:khkdhkfddfdxsysxyd}
j \leq J(N(j)) \leq k_{N(j)} + j - 1
\end{equation}
\end{itemize}
The second bullet implies $w_j > 0$ and $p_j > 0$, for all $j\in \Z$. The first bullet implies that
$w_j \in (0,1)$ if and only if $j \in J(\Z)$. Therefore, we have
\begin{align*}
\mathcal{Q}_R(w_j) & = \begin{cases} \tfrac{1}{w_j}-1 & \text{if $j \in J(\Z)$}\\
w_j - 1 & \text{if $j \notin J(\Z)$}
\end{cases}
&
\;\;\boldsymbol{\mathcal{Q}}_R\{w_j\}(p_j,0) &= \begin{cases} (\tfrac{1}{1+p_j},0) & \text{if $j \in J(\Z)$}\\
(1+p_j,0) & \text{if $j \notin J(\Z)$}
\end{cases}
\end{align*}
In the case $j \notin J(\Z)$, we have $N(j+1) = N(j)$, and therefore $w_{j+1} = w_j - 1$ and
$p_{j+1} = p_j + 1$, as required. In the case $j \in J(\Z)$, we have $N(j+1) = N(j)+1$ and
$J(N(j+1)) = J(N(j)+1) = J(N(j)) + k_{N(j)+1} = j + k_{N(j)+1}$, which implies
\begin{align*}
w_{j+1} & = \langle k_{N(j)+2},k_{N(j)+3},\ldots \rangle + k_{N(j)+1} - 1 = \tfrac{1}{w_j} - 1\\
p_{j+1} & = \langle k_{N(j)},k_{N(j)-1},\ldots\rangle = \tfrac{1}{1+p_j}
\end{align*}
as required. {\bf Part 1} is checked.\\
To prove {\bf Part 2}, we first construct two sequences $(q_j)_{j \geq 0}$
and $(\mathbf{h}_j)_{j \geq 0}$. Then we verify that they have the desired properties. Below, a sequence of dots\,\,$\cdots$ in any product of the form
$F_m \cdots F_n$ indicates that indices increase towards the left, one by one. The product is equal to one if $m=n-1$. Define sequences $(w_n^{\ast})_{n \in \Z}$ and $(p_n^{\ast})_{n \in \Z}$ by
$w^{\ast}_n = w_{J(n)} \in (0,1) \setminus \mathbb{Q}$ and $p^{\ast}_n = p_{J(n)} \in (0,\infty)\setminus \mathbb{Q}$. Equivalently,
$$ \tfrac{1}{1+p^{\ast}_n} = \langle k_n,k_{n-1},k_{n-2},\ldots \rangle
\qquad
w_n^{\ast} = \langle k_{n+1},k_{n+2},k_{n+3},\ldots \rangle$$
so that
$w_{n+1}^{\ast} = \mathcal{E}_R(w_n^{\ast})$ and $(p_{n+1}^{\ast},0) = 
\boldsymbol{\mathcal{E}}_R\{w_n^{\ast}\}(p_n^{\ast},0)$, by Proposition \ref{prop:kshjajjjajsdkds}.
Let
$(V_n^{\ast})_{n\geq 0}$, with $V_n^{\ast} = (\mu_n^{\ast},\nu_n^{\ast})^T$, 
as in Definition \ref{def:kdhkdhkhfd},
be the 
solution to $V_{n+1}^{\ast} = X_n^{\ast} V_n^{\ast} + Y_n^{\ast}$ for all $n \geq 0$ with
 $\mu_0^{\ast} = 0$ and $\nu_0^{\ast} = (1+w_0^{\ast})/\mathbf{h}_0 > 0$, where
$$X_n^{\ast} = \frac{1}{w_n^{\ast}}\begin{pmatrix}
-\tfrac{1}{1+p_n^{\ast}} & 0 \\
1 & 1 + p_n^{\ast} \end{pmatrix}
\qquad
Y_n^{\ast} = 
\frac{1}{w_n^{\ast}}\begin{pmatrix} A_1(w_n^{\ast}) - p_{n+1}^{\ast}\,A_2(w_n^{\ast}) \\
A_2(w_n^{\ast})
\end{pmatrix}
$$
\newcommand{\repindex}{s}%
Let $(V_j)_{j \geq 0}$, with $V_j = (\mu_j,\nu_j)^T$, be given by
$V_0 = V_0^{\ast}$ and for all $j \geq 1$ by
$V_j = X_{N(j)-1}^{\ast} V_{N(j)-1}^{\ast} + Y_j$, where
$$
Y_j = \frac{1}{w^{\ast}_{\repindex}} \begin{pmatrix}
A_1(w_{\repindex}^{\ast},j-J(\repindex)) - p_j\,A_2(w_{\repindex}^{\ast},j-J(\repindex))\\
A_2(w_{\repindex}^{\ast},j-J(\repindex))
\end{pmatrix}\bigg|_{\repindex = N(j)-1}
$$
The functions $A_1$ and $A_2$,
on the right hand side, are well defined at $(w_\repindex,j-J(\repindex))$, where
$\repindex = N(j)-1$, because 
$1 \leq j - J(\repindex) \leq k_{N(j)} = \lfloor 1/w^{\ast}_{\repindex} \rfloor$. 
The following two observations will be used later on:
\begin{itemize}
\item Recall \eqref{eq:pjpjpj}. For all $j \geq 1$, $\repindex = N(j)-1$, we have
$p_j = (j-J(\repindex)) - p^{\ast}_{\repindex}/(1+p^{\ast}_{\repindex})$, and consequently the estimates
after \eqref{eq:dhdkjhdkd} apply to $Y_j$, $j\geq 1$. They also apply to
$Y^{\ast}_n$, $n\geq 0$, because $p^{\ast}_{n+1} = \lfloor 1/w^{\ast}_n\rfloor - p^{\ast}_n/(1+p^{\ast}_n)$.
\item $Y_{J(n)} = Y^{\ast}_{n-1}$ for all $n\geq 1$, and consequently
 $V_{J(n)} = V^{\ast}_n$. The last identity is also true when $n=0$, because $J(0) = 0$.
\end{itemize}
As in Definition \ref{def:propagator}, set
 $P_{n,m}^{\ast}
= X_{n-1}^{\ast}\cdots X_m^{\ast}$
for all $n\geq m \geq 0$. For all $j \geq 1$, $\repindex=N(j)-1$, Lemma \ref{lemme:1} implies
$$V_j = X^{\ast}_\repindex \big(P_{\repindex,0}^{\ast} V_0^{\ast} + \textstyle\sum_{\ell=0}^{\repindex-1} P_{\repindex,\ell+1}^{\ast} Y_{\ell}^{\ast}\big)
+ Y_j
= P_{\repindex+1,0}^{\ast} V_0^{\ast} +
\textstyle\sum_{\ell=0}^{\repindex-1} P_{\repindex+1,\ell+1}^{\ast} Y_{\ell}^{\ast} + Y_j
$$
The last equation, the estimates after \eqref{eq:dhdkjhdkd}, and the estimates in Lemma \ref{lem:kdskjshkshks} imply
\begin{subequations}\label{eq:est1}
\begin{align}
|\mu_j| & \leq \frac{2^5}{w^{\ast}_\repindex} \sum_{\ell=0}^\repindex \frac{1}{w_{\ell}^{\ast}}\\
\nu_j & \geq \frac{1}{2 w_{\repindex}^{\ast}} \Big(\frac{1}{w^{\ast}_{\repindex-1}\cdots w_{-1}^{\ast}}\Big)^2
\bigg(w^{\ast}_{-1}\nu_0 - 2^8 \sum_{\ell=0}^{\repindex}\big(w_{\ell-1}^{\ast}\cdots w_{-1}^{\ast}\big)^2\; \log\Big(1 + \frac{1}{w^{\ast}_{\ell}}\Big)\hskip-1pt\bigg)\\
\nu_j & \leq \frac{2}{w_{\repindex}^{\ast}} \Big(\frac{1}{w^{\ast}_{\repindex-1}\cdots w_{-1}^{\ast}}\Big)^2
\bigg(w^{\ast}_{-1}\nu_0 + 2^6 \sum_{\ell=0}^{\repindex}\big(w_{\ell-1}^{\ast}\cdots w_{-1}^{\ast}\big)^2\; \log\Big(1 + \frac{1}{w^{\ast}_{\ell}}\Big)\hskip-1pt\bigg)
\end{align}
\end{subequations}
for all $j\geq 1$ and $\repindex=N(j)-1$.
All three estimates are also true when $j=0$, $\repindex=-1$.
Abbreviate $\mathbf{C} = \mathbf{C}(w_0,q_0) \geq 1$. We have $k_n \leq \mathbf{C}\gold_+^{2n}$
for all $n\geq 0$. Estimate 
\begin{align*}
& 2^8 \textstyle\sum_{\ell=0}^\repindex\big(w_{\ell-1}^{\ast}\cdots w_{-1}^{\ast}\big)^2\; \log\big(1 + 1/w^{\ast}_{\ell}\big)\\
& \;\;\leq\; 2^8 w_{-1}^{\ast} \textstyle\sum_{\ell=0}^{\infty} (\gold_+)^{-2\ell + 2} \log(2 + k_{\ell+1})
\hskip 20mm \text{see inequality \eqref{eq:ljfljdflfd}}\\
& \;\;\leq\; 2^8 w_{-1}^{\ast} \textstyle\sum_{\ell=0}^{\infty} (\gold_+)^{-2\ell + 2} (2+\log k_{\ell+1})\\
& \;\;\leq\; 2^{13} w_{-1}^{\ast} \big(1 + \log \mathbf{C}\big) \;\leq\; 2^{13} w_{-1}^{\ast} \mathbf{C}
\;\leq\; 2^{-1} w^{\ast}_{-1}\tfrac{1}{\mathbf{h}_0} \;\leq\; 2^{-1}w^{\ast}_{-1}\nu_0
\end{align*}
Hence, for all $j \geq 0$,
\begin{align}\label{eq:est2}
\frac{1}{4}\;\;\leq\;\;\frac{w_{N(j)-1}^{\ast}}{w^{\ast}_{-1}} \big(w^{\ast}_{N(j)-2}\cdots w_{-1}^{\ast}\big)^2 \frac{\mathbf{h}_0}{1+w_0^{\ast}}\;\nu_j & \;\;\leq\;\; 4
\end{align}
Define sequences $(\mathbf{h}_j)_{j\geq 0}$
and $(q_j)_{j \geq 0}$ by $\mathbf{h}_j = (1+w_j)/\nu_j > 0$
and  $q_j = p_j + \mu_j/\nu_j$. These definitions are consistent when $j=0$.
Observe that
$1 + w_j \leq 2 + J(N(j))-j
\leq 1+k_{N(j)} \leq 2/w^{\ast}_{N(j)-1}$. Therefore, the estimates \eqref{eq:ljfljdflfd}, \eqref{eq:est1}, \eqref{eq:est2} imply for $j\geq 1$:
\begin{subequations}\label{ddkhkfdhkdhkfd}
\begin{align}
\frac{1}{4} & \leq \frac{\mathbf{h}_j}{\mathbf{H}_j} \leq 4
\quad \text{where}\quad
\mathbf{H}_j = \mathbf{h}_0\, \frac{1+w_j}{1+w_0^{\ast}}\,\prod_{\ell=0}^{N(j)-1}
\big(w^{\ast}_{\ell} w^{\ast}_{\ell-1}\big)\\
\mathbf{H}_j & 
\leq 2\,\mathbf{h}_0\,
\big(\textstyle\prod_{\ell=0}^{N(j)-2} w^{\ast}_{\ell}\big)
\big(\textstyle\prod_{\ell=-1}^{N(j)-2} w^{\ast}_{\ell}\big)
\leq 2^4\,\mathbf{h}_0\, \gold_+^{-2N(j)}\\
\notag |q_j - p_j| 
& \leq 
\frac{2^7\mathbf{h}_0}{w_{N(j)-1}^{\ast}} \Big(\sum_{\ell=0}^{N(j)-1}\hskip-2pt\frac{1}{w_{\ell}^{\ast}}\Big)
 \prod_{\ell=0}^{N(j)-1}  (w_{\ell}^{\ast} w_{\ell-1}^{\ast})
\leq 2^{12}\, \mathbf{h}_0\, N(j)\, \gold_+^{-2N(j)} k_{N(j)}
\end{align}
\end{subequations}
The left hand sides are also less than or equal to the right hand sides when $j=0$.
Using \eqref{eq:pjpjpj}, one estimates
\begin{align*}
\mathrm{dist}_{\R} (p_j,\Z) & = \mathrm{dist}_{\R}\big(\langle k_{N(j)-1},k_{N(j)-2},\ldots \rangle,\{0,1\}\big)\\
& \geq \min \big\{\frac{1}{k_{N(j)-1}+1},\frac{1}{k_{N(j)-2}+2}\big\}
 \geq \frac{1}{3 \max \{k_{N(j)-1},k_{N(j)-2}\}}
\end{align*}
By the definition of $\mathbf{C}$ and by the assumption $\mathbf{h}_0 \leq 2^{-14} \mathbf{C}^{-1}$, we have
$|q_j-p_j| \leq \tfrac{3}{4} \mathrm{dist}_{\R} (p_j,\Z) < \mathrm{dist}_{\R} (p_j,\Z)$ for all $j\geq 0$. Therefore,
$q_j \in (0,\infty)\setminus \Z$. Moreover, $q_j\in (0,1)$ iff $p_j \in (0,1)$ iff
$k_{N(j)} + j - J(N(j))-1 = 0$ iff $J(N(j)-1) = j-1$ iff $N(j)-1 = N(j-1)$ iff $j-1\in J(\Z)$.\\
For every $j\geq 0$,
\begin{align*}
w_j & \leq J(N(j))-j+1 \leq k_{N(j)}\\
1/w_j & \leq k_{N(j)+1} + 1\\
q_j & \leq p_j + 1 \leq k_{N(j)} + j - J(N(j)) + 1 \leq k_{N(j)} + 1\\
\big(\mathrm{dist}_{\R}(q_j,\Z)\big)^{-1} & \leq 4\, \big(\mathrm{dist}_{\R}(p_j,\Z)\big)^{-1} \leq 12 \max\{k_{N(j)-1},k_{N(j)-2}\}
\end{align*}
Finally, we show that for all $j \geq 0$,
\begin{enumerate}[(a)]
\item the denominator of $\boldsymbol{\mathcal{Q}}_R\{w_j\}$, given by \eqref{eq:kdshkhkfdhfdX1} or \eqref{eq:kdshkhkfdhfdX2},
is strictly positive at $(q_j,\mathbf{h}_j)$
\item $(q_{j+1},\mathbf{h}_{j+1}) = \boldsymbol{\mathcal{Q}}_R\{w_j\}(q_j,\mathbf{h}_j)$
\end{enumerate}
For all $j \geq 0$, we have
$$2\mathbf{h}_jk_{N(j)+1} \leq 2 (2^6\mathbf{h}_0 \gold_+^{-2N(j)}) (\mathbf{C} \gold_+^{2N(j)+2})
\leq 2^9 \mathbf{h}_0 \mathbf{C} \leq 2^{-5}$$
This implies $\mathbf{h}_j \log (1+1/w_j) \leq 2\mathbf{h}_jk_{N(j)+1} < 2^{-1}$, which by inspection of
\eqref{eq:kdshkhkfdhfdX1} and \eqref{eq:kdshkhkfdhfdX2} implies (a). To show (b),
observe that by construction of $(V_n^{\ast})_{n \geq 0}$,
$$\Big(p_{n+1}^{\ast} + \frac{\mu_{n+1}^{\ast}}{\nu_{n+1}^{\ast}},\;\frac{1+w_{n+1}^{\ast}}{\nu_{n+1}^{\ast}}\Big) = \boldsymbol{\mathcal{E}}_R\{w_n^{\ast}\}\Big(p_n^{\ast} + \frac{\mu_n^{\ast}}{\nu_n^{\ast}},\;\frac{1+w_n^{\ast}}{\nu_n^{\ast}}\Big)$$
for all $n \geq 0$, see Definition \ref{def:kdhkdhkhfd} and Proposition \ref{prop:jdhkhkds}. 
Since $V_{J(n)} = V_n^{\ast}$ for all $n \geq 0$
and since $\lfloor 1/w_n^{\ast}\rfloor = k_{n+1} = J(n+1)-J(n)$, the last equation is equivalent to
\begin{equation}\label{eq:dskhkdhd}
(q_{J(n+1)},\mathbf{h}_{J(n+1)}) 
= 
\boldsymbol{\mathcal{Q}}_R^{J(n+1)-J(n)}\{w_{J(n)}\}(q_{J(n)},\mathbf{h}_{J(n)})
\end{equation}
By Proposition \ref{prop:jdhkhkds} and by the construction of $(V_j)_{j \geq 0}$,
for all $j \geq 1$, $s = N(j)-1$,
$$\Big(p_j + \frac{\mu_j}{\nu_j},\;\frac{1+1/w_s^{\ast} - j + J(s)}{\nu_j}\Big) = \boldsymbol{\mathcal{Q}}_R^{j-J(s)}\{w_s^{\ast}\}\Big(p_s^{\ast} + \frac{\mu_s^{\ast}}{\nu_s^{\ast}},\;\frac{1+w_s^{\ast}}{\nu_s^{\ast}}\Big)$$
Since $1/w_s^{\ast} - j + J(s) = w_j$, this implies
$(q_j,\mathbf{h}_j) = \boldsymbol{\mathcal{Q}}_R^{j-J(s)}\{w_{J(s)} \}(q_{J(s)},\mathbf{h}_{J(s)})$,
for all $j \geq 1$, $s=N(j)-1$.
The last identity and \eqref{eq:dskhkdhd} imply
$(q_j,\mathbf{h}_j) = \boldsymbol{\mathcal{Q}}_R^j\{w_0\}(q_0,\mathbf{h}_0)$ for all $j \geq 0$, which is equivalent to (b).\\
To prove {\bf Part 3}, check that for all $j\geq 0$ the following implication holds:
$$
\left. \begin{aligned}
w_{j+1} & = \mathcal{Q}_R(w_j)\\
(q_{j+1},\mathbf{h}_{j+1}) & = \boldsymbol{\mathcal{Q}}_R\{w_j\}(q_j,\mathbf{h}_j)
\end{aligned}\;\;\right\}
\quad \Longrightarrow \quad
\mathcal{Q}_L(\mathbf{h}_{j+1},w_{j+1},q_{j+1}) = (\mathbf{h}_j,w_j,q_j)
$$
To make this calculation, distinguish the cases $j \in J(\Z)$ and $j\notin J(\Z)$, and recall $w_j,q_j \in (0,\infty)\setminus \Z$ and that $w_j\in (0,1)$ iff $j \in J(\Z)$ iff $q_{j+1}\in (0,1)$.
\qed
\end{proof}

%% file: SectionLimit.tex
\section{An abstract semi-global existence theorem}
This section is logically self-contained, and the notation is introduced from scratch. The objects in this section are abstractions of concrete objects that appear in other sections of this paper. This relationship is reflected in the choice of notation: abstract objects are named after their concrete counterparts, whenever possible. \emph{This section is an independent unit. Definitions in other sections are irrelevant here and must be ignored.}
\begin{definition} For every integer $d \geq 1$, denote by $\|\cdot\|$ the Euclidean distance in $\R^d$. Set
$B[\delta,\mathbf{f}] = \{\mathbf{g} \in \R^d\,|\, \|\mathbf{g}-\mathbf{f}\| \leq \delta\}$ for every $\delta \geq 0$ and every $\mathbf{f}\in \R^d$.
\end{definition}
\begin{proposition}\label{prop:abstract}
Fix an integer $d \geq 1$. Suppose:
\begin{enumerate}[(a)]
\item\label{item:abs1} $\mathcal{F} \subset \R^d$ is a nonempty open subset
and
 $B\mathcal{F} = \{(\delta,\mathbf{f})\in [0,\infty)\times \mathcal{F}\,\big|\, B[\delta,\mathbf{f}] \subset \mathcal{F}\}$.
\item\label{item:prop} $\Pi_j: \mathcal{F}\to \R^d$ is a continuous map, for every integer $j\geq 1$.
\item\label{item:abs2} $\mathcal{Q}_L: \mathcal{F} \to \R^d$ and
$\mathrm{Err}:  B\mathcal{F}  \to [0,\infty)$ and
$\mathrm{Lip}:  B\mathcal{F}  \to [0,\infty)$ are maps such that
for all $(\delta,\mathbf{f})\in B\mathcal{F}$:
\begin{subequations}
\begin{align}
\label{eq:kdhkd1}
\textstyle\sup_{j \geq 1}\sup_{\mathbf{g}\in B[\delta,\mathbf{f}]}
\;\|\Pi_j(\mathbf{g}) - \mathcal{Q}_L(\mathbf{g})\| & \leq \mathrm{Err}(\delta,\mathbf{f})\\
\label{eq:kdhkd2}
{\textstyle\sup_{\mathbf{g},\mathbf{g}'\in B[\delta,\mathbf{f}],\; \mathbf{g}\neq \mathbf{g}'}}
\;\frac{\|\mathcal{Q}_L(\mathbf{g}) - \mathcal{Q}_L(\mathbf{g}'\,)\|}{\|\mathbf{g}-\mathbf{g}'\,\|} & \leq \mathrm{Lip}(\delta,\mathbf{f})
\end{align}
\end{subequations}
\item\label{item:zeta} $(\delta_j,\mathbf{f}_j)_{j \geq 0}$ is a sequence in $B\mathcal{F}$
so that $\mathbf{f}_{j-1} = \mathcal{Q}_L(\mathbf{f}_j)$ for all $j \geq 1$, and so that
\begin{equation}\label{eq:kfdhkhdkd}
\textstyle\sum_{n=j+1}^{\infty} \big\{ \prod_{k=j+1}^{n-1} \mathrm{Lip}(\delta_{k},\mathbf{f}_{k})\big\}\, \mathrm{Err} (\delta_n,\mathbf{f}_n)\;\;\leq\;\;\delta_j
\end{equation}
for all $j \geq 0$.
\end{enumerate}
Then, there exists a sequence $(\mathbf{g}_j)_{j \geq 0}$
with $\mathbf{g}_j \in B[\delta_j,\mathbf{f}_j] \subset \mathcal{F}$ such that for all $j\geq 1$:
$$\mathbf{g}_{j-1} = \Pi_j (\mathbf{g}_j)$$
\end{proposition}
\begin{proof}
For all integers $0 \leq j \leq \ell$, set
$$E_j^{\ell} =
\textstyle\sum_{n=j+1}^\ell  \big\{\prod_{k=j+1}^{n-1} \mathrm{Lip}(\delta_k,\mathbf{f}_k)\big\}\,\mathrm{Err} (\delta_n,\mathbf{f}_n)\;\;\in\;\;[0,\infty)
$$
Then $E_j = \lim_{\ell \to \infty} E_j^{\ell}$ is the left hand side of \eqref{eq:kfdhkhdkd}. Observe that $E_j^{j} = 0$
and $E_j^{\ell} \leq E_j \leq \delta_j$ by (\ref{item:zeta}). Moreover, $E_{j-1}^{\ell} = \mathrm{Lip}(\delta_{j},\mathbf{f}_{j}) E_{j}^{\ell}
+ \mathrm{Err}(\delta_{j},\mathbf{f}_{j})$ when $1 \leq j \leq \ell$.\\
For all integers $0 \leq m \leq \ell$, let $\text{(A)}^{m,\ell}$ be the statement:
\emph{There is a finite sequence $\mathbf{g}^{m,\ell} = (\mathbf{g}^{m,\ell}_j)_{m \leq j \leq \ell}$
with $\mathbf{g}^{m,\ell}_j \in B[E^{\ell}_j,\mathbf{f}_j] \subset B[\delta_j,\mathbf{f}_j] \subset \mathcal{F}$ for all $m\leq j \leq \ell$, such that $\mathbf{g}^{m,\ell}_{\ell} = \mathbf{f}_{\ell}$ and $\mathbf{g}^{m,\ell}_{j-1} = \Pi_j(\mathbf{g}^{m,\ell}_{j})$ when $m+1 \leq j \leq \ell$.} Observe that if $\text{(A)}^{m,\ell}$ is true, then the sequence $\mathbf{g}^{m,\ell}$ is unique.\\
For every fixed $\ell \geq 0$, we show by induction over $m$, one-by-one from $m=\ell$ down to $m=0$, that $\text{(A)}^{m,\ell}$ is true. The base case $\text{(A)}^{\ell,\ell}$ is trivial. For the induction step, let $1 \leq m \leq \ell$ and suppose $\text{(A)}^{m,\ell}$ is true. Define  $\mathbf{g}^{m-1,\ell}$ by $\mathbf{g}^{m-1,\ell}_j = \mathbf{g}^{m,\ell}_j
\in B[E^{\ell}_j ,\mathbf{f}_j] \subset B[\delta_j,\mathbf{f}_j] \subset \mathcal{F}$ when $m\leq j \leq \ell$, and set $\mathbf{g}^{m-1,\ell}_{m-1} = \Pi_{m}(\mathbf{g}^{m-1,\ell}_{m}) = \Pi_m(\mathbf{g}^{m,\ell}_m) \in \R^d$.
The statement $\text{(A)}^{m-1,\ell}$ is true, if $\mathbf{g}^{m-1,\ell}_{m-1}\in B[E^{\ell}_{m-1},\mathbf{f}_{m-1}]$, which follows from
\begin{align*}
\|\mathbf{g}^{m-1,\ell}_{m-1} - \mathbf{f}_{m-1}\| & = \|\Pi_m(\mathbf{g}^{m,\ell}_m) - \mathcal{Q}_L(\mathbf{f}_m)\|\\
& \leq \|\Pi_m(\mathbf{g}^{m,\ell}_m) - \mathcal{Q}_L(\mathbf{g}^{m,\ell}_m)\|
+ \|\mathcal{Q}_L(\mathbf{g}^{m,\ell}_m) - \mathcal{Q}_L(\mathbf{f}_m)\|\\
& \leq \mathrm{Err}(\delta_m,\mathbf{f}_m) + \mathrm{Lip}(\delta_m,\mathbf{f}_m) E_m^{\ell} = E_{m-1}^{\ell}
\end{align*}
We have shown that $\text{(A)}^{m,\ell}$ is true for all $0 \leq m \leq \ell$. For all integers $0 \leq j \leq \ell$,
set $\mathbf{g}^{\ell}_j = \mathbf{g}^{0,\ell}_j \in B[\delta_j,\mathbf{f}_j]$, where $\mathbf{g}^{0,\ell} = (\mathbf{g}^{0,\ell}_j)_{0\leq j \leq \ell}$ is the sequence in $\text{(A)}^{0,\ell}$. 
For every fixed $j \geq 0$, the sequence $\mathbf{g}_j = (\mathbf{g}^{\ell}_j)_{\ell \geq j}$ in the compact $B[\delta_j,\mathbf{f}_j]$ has a convergent subsequence $(\mathbf{g}^{\ell}_j)_{\ell \in \mathcal{L}_j}$, where $\mathcal{L}_j\subset [j,\infty)\cap \Z$ is infinite. One may choose $\mathcal{L}_0\supset \mathcal{L}_1 \supset \ldots$, that is $\mathcal{L}_{j-1}\supset \mathcal{L}_j$ for all $j \geq 1$. Pick a sequence $(\ell_j)_{j \geq 0}$ with $\ell_j \in \mathcal{L}_j$ for all $j\geq 0$, such that $\ell_{j-1} < \ell_j$ for all $j\geq 1$.  Set $\mathcal{L} = \{\ell_j\,|\, j \geq 0\}$. 
By construction, all but a finite number of elements of $\mathcal{L}$ are in $\mathcal{L}_j$, for every $j\geq 0$.
That is, $(\mathbf{g}^{\ell}_j)_{\ell \in \mathcal{L}\cap [j,\infty)}$ converges. Set $\mathbf{g}_j = \lim_{\ell \to \infty,\; \ell \in \mathcal{L} \cap [j,\infty)} \mathbf{g}^{\ell}_j \in B[\delta_j,\mathbf{f}_j]$. For all $j \geq 1$, 
\begin{align*}
\Pi_j(\mathbf{g}_j) & = 
\textstyle\lim_{\ell \to \infty,\;\ell \in \mathcal{L}\cap [j,\infty)}
\Pi_j(\mathbf{g}^{\ell}_j) \\
& = \textstyle\lim_{\ell \to \infty,\;\ell \in \mathcal{L}\cap [j,\infty)} \mathbf{g}^{\ell}_{j-1}
 = \mathbf{g}_{j-1}
\end{align*}
because $\Pi_j$ is continuous by (\ref{item:prop}). \qed
\end{proof}

%% file: SectionSummary.tex
\section{Main Theorems}\label{sec:main}
In this section, $\tau_{\ast}$, $\mathbf{K}$, $\mathcal{F}$ are given just as in Definitions
\ref{cccc}, \ref{ccccc}, \ref{cccccc}, and $\mathcal{Q}_L$ is the map in Definition \ref{def:kdhkhskhkdssPT2}.
\begin{definition}
Let $\|\cdot\|$ be the Euclidean distance in $\R^3$.
For every $\delta \geq 0$ and every $\mathbf{f}\in \R^3$, set $B[\delta,\mathbf{f}] = \{\mathbf{g} \in \R^3\,|\, \|\mathbf{g}-\mathbf{f}\| \leq \delta\}$.
\end{definition}
\begin{definition}\label{def:dffddflj}
Let $\mathcal{F}\subset (0,\infty)^3$ be as in Definition \ref{cccccc}. For all $\zeta \geq 1$ set
$$B_{\zeta}\mathcal{F} = \big\{(\delta,\mathbf{f})\in [0,\infty)\times \mathcal{F}\;|\; B[\zeta \delta,\mathbf{f}]\subset \mathcal{F}\big\}\qquad \text{and}
\qquad B\mathcal{F} = B_1\mathcal{F}$$
\end{definition}
\begin{lemma}\label{lem:dskhdkfhds}
For all $(\delta,\mathbf{f})\in B\mathcal{F}$ set
\begin{alignat*}{4}
W(\delta,\mathbf{f}) & = \max\{\tfrac{1}{w-\delta},w+\delta,\tfrac{1}{q-\delta},\tfrac{1}{|q-1|-\delta},q+\delta\} && \in [1,\infty)\\
W(\mathbf{f}) & = W(0,\mathbf{f})
=  \max\{\tfrac{1}{w},w,\tfrac{1}{q},\tfrac{1}{|q-1|},q\}
&\qquad& \in [1,\infty)
\end{alignat*}
where $\mathbf{f} = (\mathbf{h},w,q)$. Then:
\begin{enumerate}[(a)]
\item\label{lem:dskhdkfhdsl} $W(\mathbf{g}) \leq W(\delta,\mathbf{f})$ for all $\mathbf{g}\in B[\delta,\mathbf{f}]$. 
\item\label{lem:dskhdkfhdsl2} If $(\delta,\mathbf{f})\in B_2\mathcal{F} \subset B\mathcal{F}$ then $W(\delta,\mathbf{f}) \leq 2 W(\mathbf{f})$.
\end{enumerate}
\end{lemma}
\begin{lemma}\label{lem:xy1} Let $\mathrm{Err}: B\mathcal{F}\to [0,\infty)$ be given by $$\mathrm{Err}(\delta,\mathbf{f}) = 2^{40} \big(\tfrac{1}{\mathbf{h}-\delta}\big)^2\,W(\delta,\mathbf{f})^5 \exp(-\tfrac{1}{\mathbf{h}} 2^{-9} W(\delta,\mathbf{f})^{-2}\big)$$
where $\mathbf{f} = (\mathbf{h},w,q)$. Then
for all $(\delta,\mathbf{f}) \in B\mathcal{F}$, we have $\mathbf{K}(\mathbf{g}) \leq \mathrm{Err}(\delta,\mathbf{f})$ for all $\mathbf{g}\in B[\delta,\mathbf{f}] \subset \mathcal{F}$
(see Definition \ref{ccccc}).
\end{lemma}
\begin{proof} Let $\mathbf{g} = (\mathbf{h}',w',q') \in B[\delta,\mathbf{f}]$. Then 
$\tau_{\ast}(\mathbf{g}) \geq \tfrac{1}{2} W(\mathbf{g})^{-2}$
and $0 < \mathbf{h} - \delta \leq \mathbf{h}' \leq \mathbf{h} + \delta \leq 2 \mathbf{h}$ and $\tfrac{1}{\mathbf{h}'}\geq \tfrac{1}{2\mathbf{h}}$. Hence, $\mathbf{K}(\mathbf{g}) \leq 2^{40} (\tfrac{1}{\mathbf{h}-\delta})^2 W(\mathbf{g})^5 \exp(-\tfrac{1}{\mathbf{h}} 2^{-9}  W(\mathbf{g})^{-2})$.
Now use Lemma \ref{lem:dskhdkfhds} (\ref{lem:dskhdkfhdsl}). \qed
\end{proof}
\begin{lemma}\label{lem:xy2}
Let $\mathcal{Q}_L$ be as in Definition \ref{def:kdhkhskhkdssPT2}. Set $\mathrm{Lip}: B \mathcal{F} \to [0,\infty)$,\; $\mathrm{Lip}(\delta,\mathbf{f}) = 2^{13} W(\delta,\mathbf{f})^3$.
Then $\|\mathcal{Q}_L(\mathbf{g}) - \mathcal{Q}_L(\mathbf{g}'\,)\|
\leq \mathrm{Lip}(\delta,\mathbf{f})\,\|\mathbf{g}-\mathbf{g}'\,\|$ for all
$\mathbf{g},\mathbf{g'}\in B[\delta,\mathbf{f}]$. 
\end{lemma}
\begin{proof}
Let $\mathbf{f} = (\mathbf{h},w,q)$.
If $\mathbf{g}=\mathbf{g}'$, there is nothing to prove. Suppose $\mathbf{g}\neq \mathbf{g}'$. 
In Lemma \ref{lem:dfskdshkhk} of Appendix \ref{app:dslkhjlkd}, set $\mathbf{f}_1 = (\mathbf{h}_1,w_1,q_1)= \mathbf{g}$ and $\mathbf{f}_2 = (\mathbf{h}_2,w_2,q_2)= \mathbf{g}'$. Observe that $0< \mathbf{h}_i \leq 1$ by $\mathbf{g},\mathbf{g}'\in B[\delta,\mathbf{f}] \subset \mathcal{F}$. 
Since $\delta < |q-1|$, either $q,q_1,q_2<1$ or $q,q_1,q_2>1$. We have
$w_{\mathrm{max}} \leq \max\{W(\mathbf{g}),W(\mathbf{g}')\}$ and 
$q_{\mathrm{max}}\leq \max\{W(\mathbf{g}),W(\mathbf{g}')\}$ and 
$q_{\mathrm{min}}^{-1} = \max\{q_1^{-1},q_2^{-1}\} \leq \max\{W(\mathbf{g}),W(\mathbf{g}')\}$. Now use 
$\log(2+w_{\mathrm{max}}) \leq 1+w_{\mathrm{max}}$ and Lemma \ref{lem:dskhdkfhds} (\ref{lem:dskhdkfhdsl}). \qed
\end{proof}
\begin{theorem}[Main Theorem 1]\label{thm1}
Recall the definitions of
$\mathcal{P}_L$ and $\mathcal{Q}_L$ (Definition \ref{def:kdhkhskhkdssPT2}),
$\mathcal{F}$ (Definition \ref{cccccc}),
$\Pi$ (Proposition \ref{prop:skhdkjhfd}),
$B_{\zeta}\mathcal{F}$ (Definition \ref{def:dffddflj}),
$W$ (Lemma \ref{lem:dskhdkfhds}), $\mathrm{Err}$ (Lemma \ref{lem:xy1}),
$\mathrm{Lip}$ (Lemma \ref{lem:xy2}). Suppose:
\begin{enumerate}[(a)]
\item\label{item:ds1} $(\mathbf{f}_j)_{j\geq 0}$, with $\mathbf{f}_j=(\mathbf{h}_j,w_j,q_j) \in \mathcal{F}$,
satisfies $\mathbf{f}_{j-1}
= \mathcal{Q}_L(\mathbf{f}_j)$ for all $j\geq 1$.
\item\label{item:ds2} The sequence $(\delta_j)_{j \geq 0}$ given by
$$\delta_j = \sum_{\ell=j+1}^{\infty} \Big\{\prod_{k=j+1}^{\ell-1} 2^{16} W(\mathbf{f}_k)^3\Big\}\;2^{47} \big(\tfrac{1}{\mathbf{h}_\ell}\big)^2\, W(\mathbf{f}_\ell)^5 \exp\Big(-\tfrac{1}{\mathbf{h}_\ell} 2^{-11} W(\mathbf{f}_\ell)^{-2}\Big)$$
satisfies $\delta_j < \infty$ and $(\delta_j,\mathbf{f}_j)\in B_2\mathcal{F}$ for all $j \geq 0$.
\item\label{item:ds3}
$\pi_0 \in S_3$ and $(\pi_j)_{j\geq 0}$ is the unique sequence in $S_3$ that satisfies
$\pi_{j-1} = \mathcal{P}_L(\pi_j,\mathbf{f}_j)$ for all $j\geq 1$.
\item\label{item:ds4} $\sigma_{\ast} \in \{-1,+1\}^3$.
\end{enumerate}
Then, there exists a sequence $(\mathbf{g}_j)_{j \geq 0}$ with $\mathbf{g}_j \in B[\delta_j,\mathbf{f}_j]\subset \mathcal{F}$ such that for all $j \geq 1$:
$$\mathbf{g}_{j-1} = \Pi[\pi_j,\sigma_{\ast}](\mathbf{g}_j)
\qquad \text{and}
\qquad \pi_{j-1} = \mathcal{P}_L(\pi_j,\mathbf{g}_j)
$$
\end{theorem}
\begin{proof}
We use Proposition \ref{prop:abstract}, with the understanding
that the abstract objects of Proposition \ref{prop:abstract} in the left column are given by the special objects in the right column:
\begin{center}
\begin{tabular}{r | l }
$d$ & 3\\
$\mathcal{F}$ & $\mathcal{F}$ as in Definition \ref{cccccc}\\
$\Pi_j$ & $\Pi[\pi_j,\sigma_{\ast}]$, see Proposition \ref{prop:skhdkjhfd}
 and the hypotheses Theorem \ref{thm1} (\ref{item:ds3}), (\ref{item:ds4})\\
 $\mathcal{Q}_L$ & $\mathcal{Q}_L|_{\mathcal{F}}$, with $\mathcal{Q}_L$ as in Definition \ref{def:kdhkhskhkdssPT2}\\
$\mathrm{Err}$ & $\mathrm{Err}$ as in Lemma \ref{lem:xy1}\\
$\mathrm{Lip}$ & $\mathrm{Lip}$ as in Lemma \ref{lem:xy2}\\
$(\delta_j,\mathbf{f}_j)$ & $(\delta_j,\mathbf{f}_j)$ as in 
 hypotheses Theorem \ref{thm1} (\ref{item:ds1}) and (\ref{item:ds2})
\end{tabular}
\end{center}
We check that the assumptions (\ref{item:abs1}), (\ref{item:prop}), (\ref{item:abs2}), (\ref{item:zeta})
of Proposition \ref{prop:abstract} are satisfied:
\begin{enumerate}[(a)]
\item[(\ref{item:abs1})] The definitions of $B\mathcal{F}$ in 
Proposition \ref{prop:abstract} and in Definition \ref{def:dffddflj} are consistent.
\item[(\ref{item:prop})]
$\Pi[\pi_j,\sigma_{\ast}]: \mathcal{F}\to (0,\infty)^2\times \R \subset \R^3$ is continuous, by Proposition \ref{prop:skhdkjhfd}.
\item[(\ref{item:abs2})] The domains of definition of $\mathcal{Q}_L|_{\mathcal{F}}$ and $\mathrm{Err}$ and $\mathrm{Lip}$ are just as required by Proposition  \ref{prop:abstract} (\ref{item:abs2}). For all
$(\delta,\mathbf{f}) \in B\mathcal{F}$ and $\mathbf{g},\mathbf{g}' \in B[\delta,\mathbf{f}]\subset \mathcal{F}$ and $j \geq 1$,
\begin{align*}
\|\Pi[\pi_j,\sigma_{\ast}](\mathbf{g}) - \mathcal{Q}_L|_{\mathcal{F}} (\mathbf{g})\|
& 
\leq \mathbf{K}(\mathbf{g}) \leq \mathrm{Err}(\delta,\mathbf{f})\\
\|\mathcal{Q}_L|_{\mathcal{F}}(\mathbf{g}) - \mathcal{Q}_L|_{\mathcal{F}}(\mathbf{g}'\,)\|
& \leq \mathrm{Lip}(\delta,\mathbf{f})\,\|\mathbf{g}-\mathbf{g}'\,\|
\end{align*}
by Proposition \ref{prop:skhdkjhfd} (a) and by Lemmas \ref{lem:xy1}
and \ref{lem:xy2}. That is, \eqref{eq:kdhkd1} and \eqref{eq:kdhkd2} hold.
\item[(\ref{item:zeta})] 
By assumption, $(\delta_j,\mathbf{f}_j)\in B_2\mathcal{F} \subset B\mathcal{F}$ for all $j \geq 0$.
Hence
$
\tfrac{1}{2}\mathbf{h}_j \leq 
\mathbf{h}_j - \delta_j$ and, by Lemma \ref{lem:dskhdkfhds} (\ref{lem:dskhdkfhdsl2}), we have $W(\delta_j,\mathbf{f}_j) \leq 2 W(\mathbf{f}_j)$. Consequently, for all $j \geq 0$,
\begin{align*}
& \textstyle\sum_{\ell=j+1}^{\infty} \big\{ \prod_{k=j+1}^{\ell-1} \mathrm{Lip}(\delta_{k},\mathbf{f}_{k})\big\}\, \mathrm{Err} (\delta_\ell,\mathbf{f}_\ell)\\
& = \textstyle\sum_{\ell=j+1}^{\infty} \big\{ \prod_{k=j+1}^{\ell-1} 2^{13} W(\delta_{k},\mathbf{f}_{k})^3 \big\}\,\
\,2^{40}\big(\tfrac{1}{\mathbf{h}_\ell - \delta_\ell}\big)^2 W(\delta_\ell,\mathbf{f}_\ell)^5\\
& \hskip 60mm \times \exp \big(-\tfrac{1}{\mathbf{h}_\ell} 2^{-9} W(\delta_\ell,\mathbf{f}_\ell)^{-2}\big)\\
& \leq \textstyle\sum_{\ell=j+1}^{\infty} \big\{ \prod_{k=j+1}^{\ell-1} 2^{16} W(\mathbf{f}_{k})^3 \big\}\,
2^{47}\big(\tfrac{1}{\mathbf{h}_\ell}\big)^2 W(\mathbf{f}_\ell)^5
\exp \big(-\tfrac{1}{\mathbf{h}_\ell} 2^{-11} W(\mathbf{f}_\ell)^{-2}\big)
\end{align*}
The last expression is equal to $\delta_j$, and \eqref{eq:kfdhkhdkd} is checked.
\end{enumerate}
Now,  Theorem \ref{thm1} follows from Proposition  \ref{prop:abstract}. \qed
\end{proof}
\begin{theorem}[Main Theorem 2] \label{thm:main2}
Suppose the vector $\mathbf{f}_0 = (\mathbf{h}_0,w_0,q_0)$ satisfies the assumptions of
Proposition \ref{prop:kjdhkhkd}, that is
\begin{subequations}\label{eq:kdhkdhkdff}
\begin{align}
w_0 & \in (0,1)\setminus \mathbb{Q} & \mathbf{C}(w_0,q_0) & < \infty\\
q_0 & \in (0,\infty)\setminus \mathbb{Q}
& 0  < \mathbf{h}_0 & \leq 2^{-14}(\mathbf{C}(w_0,q_0))^{-1}
\end{align}
\end{subequations}
Let $(k_n)_{n\in \Z}$ and $J:\Z \to \Z$ (Era Pointer) and $N:\Z \to \Z$ (Era Counter) and $(w_j)_{j \in \Z}$, $(q_j)_{j \geq 0}$, $(\mathbf{h}_j)_{j \geq 0}$ be just as in Proposition \ref{prop:kjdhkhkd}. Introduce the sequence $(\mathbf{f}_j)_{j \geq 0}$ by
$$\mathbf{f}_j = (\mathbf{h}_j,w_j,q_j) \;\;\in\;\; (0,\infty)^3$$
Introduce sequences $(\mathbf{H}_j)_{j \geq 0}$ and $(K_j)_{j \geq 0}$ by
\begin{alignat*}{4}
\mathbf{H}_j & = \mathbf{h}_0\,\tfrac{1+w_j}{1+w_0} \textstyle\prod_{\ell=0}^{N(j)-1} w_{J(\ell)} w_{J(\ell-1)} && > 0\\
K_j & = \max\{k_{N(j)-2},k_{N(j)-1},k_{N(j)},k_{N(j)+1}\} &\qquad& \geq 1
\end{alignat*}
Suppose:
\begin{enumerate}[(a)]
\item\label{item:dkhdkhdm1} $\mathbf{H}_j  < 2^{-21}(K_j)^{-2}$ for all $j \geq 0$.
\item\label{item:dkhdkhd} $2^{71} \big(\tfrac{1}{\mathbf{H}_j}\big)^2 (K_j)^5\exp \big( - \tfrac{1}{\mathbf{H}_j} 2^{-21}(K_j)^{-2}\big)<  1$ for all $j\geq 0$.
\item\label{item:dkhdkhd2} The sequence $(\slaa{\delta}_j)_{j \geq 0}$ given by 
$$
\slaa{\delta}_j
= \sum_{\ell=j+1}^{\infty} \Big\{\prod_{k=j+1}^{\ell-1} 2^{28} (K_k)^3\Big\}\;2^{71} \big(\tfrac{1}{\mathbf{H}_\ell}\big)^2\, (K_\ell)^5 \exp\Big(-\tfrac{1}{\mathbf{H}_\ell} 2^{-21} (K_\ell)^{-2} \Big)
\quad > 0
$$
satisfies $\slaa{\delta}_j \leq  2^{-4}\mathbf{H}_j < \infty$.
\item
$\pi_0 \in S_3$ and $(\pi_j)_{j\geq 0}$ is the unique sequence in $S_3$ that satisfies
$\pi_{j-1} = \mathcal{P}_L(\pi_j,\mathbf{f}_j)$ for all $j\geq 1$.
\item $\sigma_{\ast}\in \{-1,+1\}^3$.
\end{enumerate}
Then $(\slaa{\delta}_j,\mathbf{f}_j) \in B_2\mathcal{F}$ for all $j \geq 0$ and there exists a sequence $(\mathbf{g}_j)_{j \geq 0}$ with $\mathbf{g}_j \in B[\slaa{\delta}_j,\mathbf{f}_j]\subset \mathcal{F}$ such that for all $j \geq 1$:
$$\mathbf{g}_{j-1} = \Pi[\pi_j,\sigma_{\ast}](\mathbf{g}_j)
\qquad \text{and}
\qquad \pi_{j-1} = \mathcal{P}_L(\pi_j,\mathbf{g}_j)
$$
\end{theorem}
\begin{proof}
By
Proposition \ref{prop:kjdhkhkd}
and by hypotheses (\ref{item:dkhdkhdm1}), (\ref{item:dkhdkhd2}) in Theorem \ref{thm:main2}, for all $j \geq 0$:
\begin{subequations}\label{eq:dhkfhfkfkhkhkd}
\begin{align}
\label{eq:dhkfhfkfkhkhkd1} 2^{-2} \mathbf{H}_j  \leq \mathbf{h}_j & \leq 2^2 \mathbf{H}_j \\
\label{eq:dhkfhfkfkhkhkd2} \max\{\tfrac{1}{w_j},w_j,\tfrac{1}{q_j},\tfrac{1}{|q_j-1|},q_j\} & \leq 2^4 K_j\\
\notag 2^{-4}(K_j)^{-1} & \leq \min\{w_j,q_j,|q_j-1|\}\\
\notag  2 \slaa{\delta}_j & \leq 2^{-1} \min\{w_j,q_j,|q_j-1|,\mathbf{h}_j\}
\end{align}
\end{subequations}
Hence,  $B[2\slaa{\delta}_j,\mathbf{f}_j] \subset (0,\infty)^3$ for every $j \geq 0$.
Furthermore, for all $j \geq 0$ and all $(\mathbf{h}',w',q') \in B[2\slaa{\delta}_j,\mathbf{f}_j] \subset (0,\infty)^3$, we have $q'\neq 1$ and
\begin{subequations}\label{eq:kdhkdhdfk}
\begin{equation}
2^{-3}\mathbf{H}_j \leq 2^{-1}\mathbf{h}_j \leq \mathbf{h}_j - 2\slaa{\delta}_j\leq \mathbf{h}' \leq \mathbf{h}_j + 2\slaa{\delta}_j \leq 2 \mathbf{h}_j \leq 2^3 \mathbf{H}_j
\end{equation}
and
\begin{equation}
\begin{split}
& \max\big\{\tfrac{1}{w'},w',\tfrac{1}{q'},\tfrac{1}{|q'-1|},q'\big\}\\
& \leq 
\max\big\{\tfrac{1}{w_j-2\slaz{\delta}_j},w_j+2\slaa{\delta}_j,\tfrac{1}{q_j - 2\slaz{\delta}_j },\tfrac{1}{|q_j-1|-2\slaz{\delta}_j},q_j+2\slaa{\delta}_j\big\}\\
& \leq 2\, \max\big\{\tfrac{1}{w_j},w_j,\tfrac{1}{q_j},\tfrac{1}{|q_j-1|},q_j\big\}\; \leq 2^5\, K_j
\end{split}
\end{equation}
\end{subequations}
The last two estimates \eqref{eq:kdhkdhdfk}
imply
$\tau_{\ast}(\mathbf{h}',w',q') \geq 2^{-11} (K_j)^{-2}$ and  
$$\mathbf{K}(\mathbf{h}',w',q') \leq 2^{71} \big(\tfrac{1}{\mathbf{H}_j}\big)^2 (K_j)^5\exp \big( - \tfrac{1}{\mathbf{H}_j} 2^{-21}(K_j)^{-2}\big) < 1$$
The last inequality is hypothesis (\ref{item:dkhdkhd}) in Theorem \ref{thm:main2}. Furthermore,
$$\mathbf{h}' \leq 2^3\mathbf{H}_j < 2^{-18}(K_j)^{-2} \leq 2^{-7}\tau_{\ast}(\mathbf{h}',w',q')$$
The second inequality is hypothesis (\ref{item:dkhdkhdm1}) in Theorem \ref{thm:main2}.
These estimates are true for all $(\mathbf{h}',w',q') \in B[2\slaa{\delta}_j,\mathbf{f}_j]$, and therefore $B[2\slaa{\delta}_j,\mathbf{f}_j]\subset \mathcal{F}$ for all $j \geq 0$, in particular $\mathbf{f}_j \in \mathcal{F}$ (see Definition \ref{cccccc}). In other words, $(\slaa{\delta}_j,\mathbf{f}_j) \in B_2\mathcal{F}$.\\
The last result and the fact that $\mathbf{f}_{j-1} = \mathcal{Q}_L(\mathbf{f}_j)$ for all $j \geq 1$ (see Proposition \ref{prop:kjdhkhkd}) imply that Theorem \ref{thm:main2} follows from Theorem \ref{thm1}, if we can show that
$\delta_j \leq \slaa{\delta}_j$ for all $j\geq 0$, where $\delta_j$ is given as in Theorem \ref{thm1}. 
The inequality 
$\delta_j \leq \slaa{\delta}_j$ is a consequence of
$W(\mathbf{f}_\ell) \leq 2^4 K_\ell$ and $2^{-2} \mathbf{H}_\ell \leq \mathbf{h}_\ell \leq 2^2\mathbf{H}_\ell$, where $j,\ell\geq 0$.
\qed
\end{proof}
\begin{theorem}[Main Theorem 3] \label{thm_main3}
Fix constants $\mathbf{D} \geq 1$,\, $\gamma \geq 0$. Suppose the vector $\mathbf{f}_0 = (\mathbf{h}_0,w_0,q_0) \in (0,\infty)^3$ satisfies
\begin{enumerate}[(i)]
\item $w_0 \in (0,1)\setminus \mathbb{Q}$ and $q_0 \in (0,\infty)\setminus \mathbb{Q}$.
\item\label{item:bvsvhsgfs1} $k_n \leq \mathbf{D}\,\max\{1,n\}^{\gamma}$ for all $n \geq -2$,
with $(k_n)_{n\in \Z}$ as in Proposition \ref{prop:kjdhkhkd}, that is
$$(1+q_0)^{-1} = \langle k_0,k_{-1},k_{-2},\ldots \rangle
\qquad w_0 = \langle k_1,k_2,k_3,\ldots \rangle$$
\item\label{item:bvsvhsgfs2} $0< \mathbf{h}_0< \mathbf{A}^{\sharp}$ where $\mathbf{A}^{\sharp} = \mathbf{A}^{\sharp}(\mathbf{D},\gamma) = 2^{-56}\mathbf{D}^{-4} (4(\gamma+1))^{-4(\gamma+1)}$.
\end{enumerate}
Then
\begin{itemize}
\item The assumptions \eqref{eq:kdhkdhkdff} and (\ref{item:dkhdkhdm1}), (\ref{item:dkhdkhd}), (\ref{item:dkhdkhd2}) 
of Theorem \ref{thm:main2} hold.
\item Set
 $\gold_+ = \tfrac{1}{2}(1+\sqrt{5})$. The sequence $(\slaa{\delta}_j)_{j \geq 0}$ in 
Theorem \ref{thm:main2} satisfies for all $j \geq 0$:
\begin{equation}\label{eq:kdhdkhfdkhdkd}
\slaa{\delta}_j \leq \exp\big(-\tfrac{1}{\mathbf{h}_0}\mathbf{A}^{\sharp}\gold_+^{N(j)}\big)
\qquad\text{and}\qquad
N(j) \geq \big(\mathbf{D}^{-1}j\big)^{1/(\gamma+1)}
\end{equation}
where $N:\Z \to \Z$ (Era Counter) is the map in Proposition \ref{prop:kjdhkhkd}.
\end{itemize}
If $\gamma > 1$
and $\mathbf{D} > \tfrac{1}{\log 2} \tfrac{\gamma}{\gamma-1}$, then the set of all vectors $\mathbf{f}_0\in (0,\infty)^3$ that satisfy (i), (ii), (iii) has positive Lebesgue measure.
\end{theorem}
\begin{proof}
\emph{Preliminaries.} The following facts will be used without further comment:
\begin{itemize}
\item$a^{-x}x^b \leq \big(\tfrac{b}{e\,\log a}\big)^b$ for all real numbers $a > 1$, $b > 0$, $x\geq 0$ where $e=\exp(1)$.
\item$a^b \leq c^d$ for all real numbers $1\leq a \leq c$ and $0 \leq b \leq d$.
\item\rule{0pt}{12pt}$1 < \gold_+ < 2$ and $1 < e \log \gold_+ < 2$ where $e = \exp(1)$ and $\gold_+ = \tfrac{1}{2}(1+\sqrt{5})$.
\end{itemize}
Fix $\mathbf{D}\geq 1$ and $\gamma \geq 0$ as in Theorem \ref{thm_main3}. For all 5-tuples of real numbers $s = (s_1,s_2,s_3,s_4,s_5) \geq (0,0,0,1,0)$, set $\mathbf{A}(s)
= 2^{-s_1-s_2\gamma }\mathbf{D}^{-s_3} (s_4(\gamma+1))^{-s_5(\gamma+1)}$.
Observe that $0 < \mathbf{A}(s) \leq 2^{-s_1} \leq 1$ and $\mathbf{A}(s) \leq \mathbf{A}(s')$ if $s \geq s'$.\\
\emph{Basic smallness assumptions.} $k_n \leq \mathbf{D}\, \max\{1,n\}^{\gamma}$ for all $n \geq -2$ and $\mathbf{h}_0 < \mathbf{A}(\kappa)$. The vector $\kappa = (\kappa_1,\kappa_2,\kappa_3,\kappa_4,\kappa_5) \geq (0,0,0,1,0)$ will be fixed during the proof.\\
\emph{Estimates 1.} Recall Proposition \ref{prop:kjdhkhkd} and $\gold_+ = \tfrac{1}{2}(1+\sqrt{5})$. For all $j\geq 0$, $n \geq 0$:
\begin{align*}
\mathbf{C}(w_0,q_0) & = \textstyle\sup_{n \geq 0} (n+1)\gold_+^{-2n}k_n \max\{k_{n-1},k_{n-2}\}\\
& \leq 2\mathbf{D}^2 \textstyle\sup_{n\geq 0} \gold_+^{-2n} \max\{1,n\}^{2(\gamma+1)}\\
& \leq 2\mathbf{D}^2 (\gamma+1)^{2(\gamma+1)} = \mathbf{A}(1,0,2,1,2)^{-1}
\displaybreak[0]\\
J(n) & = \textstyle\sum_{\ell=1}^n k_{\ell}
\leq \mathbf{D} \sum_{\ell=1}^n \ell^{\gamma}
\leq \mathbf{D} n^{\gamma+1}\displaybreak[0]\\
j & \leq J(N(j)) \leq \mathbf{D}
N(j)^{\gamma+1} \displaybreak[0]\\
N(j) & \geq (\mathbf{D}^{-1} j )^{1/(\gamma+1)} \displaybreak[0]\\ 
\mathbf{H}_j & \leq 2^4\,\mathbf{h}_0\, \gold_+^{-2N(j)}
\qquad \text{see \eqref{ddkhkfdhkdhkfd}}
\displaybreak[0]\\
\mathbf{H}_j & \geq 2^{-1}\mathbf{h}_0 \textstyle\prod_{\ell=0}^{N(j)-1} (k_{\ell}+1)^{-1}(k_{\ell+1}+1)^{-1}\\
& \geq 2^{-1} \mathbf{h}_0 \textstyle\prod_{\ell=0}^{N(j)-1} (2\mathbf{D}(\ell+1)^{\gamma})^{-2} 
\geq 2^{-1}\mathbf{h}_0\,\max\{1,2 \mathbf{D} N(j)^{\gamma}\}^{-2N(j)}
 \displaybreak[0]\\
K_j & \leq \mathbf{D}(N(j)+1)^{\gamma}  \leq \mathbf{D} 2^{\gamma} \max\{1,N(j)\}^{\gamma} \displaybreak[0]\\
\mathbf{H}_j K_j^2 & \leq 2^{4+2\gamma} \mathbf{D}^2 \mathbf{h}_0 \gold_+^{-2N(j)}  \max\{1,N(j)\}^{2\gamma}\\
& \leq 2^{4+2\gamma} \mathbf{D}^2 \mathbf{h}_0 \textstyle\sup_{n\geq 0}\gold_+^{-2n}  \max\{1,n\}^{2(\gamma+1)}\\
& \leq 2^{4+2\gamma} \mathbf{D}^2 \mathbf{h}_0\, (\gamma+1)^{2(\gamma+1)}
= \mathbf{h}_0\, \mathbf{A}(4,2,2,1,2)^{-1}
\end{align*}
Require $\kappa \geq (25,2,2,1,2)$. Then $\mathbf{H}_j < 2^{-21}(K_j)^{-2}$ and $\mathbf{h}_0
\leq 2^{-14} (\mathbf{C}(w_0,q_0))^{-1}$. \\
\emph{Estimates 2.} Let $(\slaa{\delta}_j)_{j\geq 0}$ be as in Theorem \ref{thm:main2}. We claim that with proper choice of $\kappa$:
\begin{enumerate}[(A)]
\item $\slaa{\delta}_{J(n)} \leq 2^{-5} \mathbf{h}_0 \big(2 \mathbf{D} (n+1)^{\gamma}\big)^{-2(n+1)}\,\exp(-\tfrac{1}{\mathbf{h}_0} \mathbf{A}(\kappa) \gold_+^{n+1}) $
for all $n \geq 0$.
\item $\slaa{\delta}_j \leq 2^{-4}\mathbf{H}_j$
and $\slaa{\delta}_j \leq \exp(-\tfrac{1}{\mathbf{h}_0} \mathbf{A}(\kappa)\,{\gold_+}^{N(j)})$ for all $j \geq 0$.
\end{enumerate}
We first check (A) $\Longrightarrow$ (B). Note that $\slaa{\delta}_j \geq \slaa{\delta}_{j+1}$, $j\geq 0$. Fix any $j\geq 0$. Set $n = N(j+1)-1 \geq 0$. By (A), by $j \geq J(n)$
(see the line before \eqref{eq:khkdhkfddfdxsysxyd}) and
by $n +1 \geq N(j)$,
\begin{align*}
\slaa{\delta}_j \leq \slaa{\delta}_{J(n)}
& \leq 2^{-5} \mathbf{h}_0 \big(2 \mathbf{D} (n+1)^{\gamma}\big)^{-2(n+1)}
\exp(-\tfrac{1}{\mathbf{h}_0} \mathbf{A}(\kappa)\gold_+^{n+1})\\
& \leq \Big(2^{-5} \mathbf{h}_0
\max\{1,2\mathbf{D}N(j)^{\gamma}\}^{-2N(j)}\Big)
\exp(-\tfrac{1}{\mathbf{h}_0} \mathbf{A}(\kappa)\gold_+^{N(j)})
\end{align*}
See the second bullet in the preliminaries.
On the right hand side, both factors are $\leq 1$ (use $\mathbf{h}_0 < \mathbf{A}(\kappa) \leq 1$). By the lower bound on $\mathbf{H}_j$ derived above, claim (B) follows.

We now check (A).
For all $n \geq 0$:
\begin{align*}
& \slaa{\delta}_{J(n)}\\
& = \textstyle\sum_{m=n}^{\infty}\textstyle\sum_{\ell=J(m)+1}^{J(m+1)} \Big\{\prod_{k=J(n)+1}^{\ell-1}  2^{28} (K_k)^3\Big\}\;2^{71} \big(\tfrac{1}{\mathbf{H}_\ell}\big)^2\, (K_\ell)^5 \\
& \hskip 76mm \times \exp\Big(-
(2^{21}\mathbf{H}_{\ell} K_{\ell}^2\big)^{-1}\Big) \displaybreak[0] \\
& \leq \textstyle\sum_{m=n}^{\infty} \textstyle\sum_{\ell=J(m)+1}^{J(m+1)}
 \big(2^{15} \max_{1\leq k \leq \ell} K_k\big)^{3\ell +2}
 \big(\tfrac{1}{2\mathbf{H}_\ell}\big)^2 \exp\Big(-
(2^{21}\mathbf{H}_{\ell} K_{\ell}^2\big)^{-1}\Big) \displaybreak[0]\\
& \leq \textstyle\sum_{m=n}^{\infty} k_{m+1} 
 \big( 2^{15+\gamma}  \mathbf{D}  (m+1)^{\gamma} \big)^{3 J(m+1) +2}
 \big(\tfrac{1}{\mathbf{h}_0}\big)^2
 \big(2 \mathbf{D} (m+1)^{\gamma}\big)^{4(m+1)}\\
& \hskip45mm \times
 \exp\Big(- 2^{-25-2\gamma} \mathbf{D}^{-2} \tfrac{1}{\mathbf{h}_0} \gold_+^{2(m+1)}   (m+1)^{-2\gamma}\Big) \displaybreak[0]\\
 & \leq  \big(\tfrac{1}{\mathbf{h}_0}\big)^2 \textstyle\sum_{m=n+1}^{\infty} 
 \big( 2^{15+\gamma} \mathbf{D}   m^{\gamma} \big)^{(10 \mathbf{D}m^{\gamma+1})}
 \exp\Big(- 2^{-25-2\gamma} \mathbf{D}^{-2} \tfrac{1}{\mathbf{h}_0} \gold_+^{2m}   m^{-2\gamma}\Big)
 \end{align*}
 Since $2^5 \tfrac{1}{\mathbf{h}_0}(2 \mathbf{D}(n+1)^{\gamma})^{2(n+1)}
 \leq \tfrac{1}{\mathbf{h}_0}
 (2^6 \mathbf{D} m^{\gamma})^{2m}$ for all $m \geq n+1$, we have
 \begin{align*}
 & \mathbf{S}(n) \stackrel{\text{def}}{=} \slaa{\delta}_{J(n)} \, 2^5 \tfrac{1}{\mathbf{h}_0}(2 \mathbf{D}(n+1)^{\gamma})^{2(n+1)}\\
  & \leq  \big(\tfrac{1}{\mathbf{h}_0}\big)^3 \textstyle\sum_{m=n+1}^{\infty} 
 \big( 2^{15+\gamma} \mathbf{D}   m^{\gamma} \big)^{(12 \mathbf{D}m^{\gamma+1})}
 \exp\Big(- 2^{-25-2\gamma} \mathbf{D}^{-2}\tfrac{1}{\mathbf{h}_0} \gold_+^{2m}   m^{-2\gamma}\Big) \displaybreak[0]\\
  & \leq  \big(\tfrac{1}{\mathbf{h}_0}\big)^3 \textstyle\sum_{m=n+1}^{\infty} 
\exp \Big( 
12 \mathbf{D}m^{\gamma+1} \log \big(  2^{15+\gamma} \mathbf{D}  m^{\gamma} \big)
- 2^{-25-2\gamma} \mathbf{D}^{-2} \tfrac{1}{\mathbf{h}_0} \gold_+^{2m}   m^{-2\gamma}\Big) \displaybreak[0]\\
   & \leq  \big(\tfrac{1}{\mathbf{h}_0}\big)^3 \textstyle\sum_{m=n+1}^{\infty} 
\exp \Big( 
2^9 \mathbf{D}^2(\gamma+1) m^{\gamma+2} 
- 2^{-25-2\gamma} \mathbf{D}^{-2} \tfrac{1}{\mathbf{h}_0} \gold_+^{2m}   m^{-2\gamma}\Big)
\end{align*}
The second term in the argument of the exponential dominates the first term, if we require $\kappa \geq (35,2,4,\tfrac{3}{2},3)$. More precisely, the absolute value of the second term is at least twice the absolute value of the first term. In fact,
\begin{align*}
& 2^{35+2\gamma}\mathbf{D}^{4} (\gamma+1) \textstyle\sup_{m\geq 1}\gold_+^{-2m}m^{3\gamma+2}\\
& \leq 2^{35+2\gamma}\mathbf{D}^4 \big(\tfrac{3}{2}(\gamma+1)\big)^{3(\gamma+1)}
=  \mathbf{A}(35,2,4,\tfrac{3}{2},3)^{-1} \leq \mathbf{A}(\kappa)^{-1} \leq \tfrac{1}{\mathbf{h}_0}
\end{align*}
Therefore,
 \begin{align*}
\mathbf{S}(n)
&  \leq  \big(\tfrac{1}{\mathbf{h}_0}\big)^3 \textstyle\sum_{m=n+1}^{\infty} 
\exp \Big( 
- 2^{-26-2\gamma} \mathbf{D}^{-2} \tfrac{1}{\mathbf{h}_0} \gold_+^{2m}   m^{-2\gamma}\Big)
\intertext{%
Moreover, $2^{26+2\gamma}\mathbf{D}^2\sup_{m\geq 1} \gold_+^{-m}m^{2\gamma}
\leq 2^{26+2\gamma} \mathbf{D}^2 (2(\gamma+1))^{2(\gamma+1)} = 2^{-2} \mathbf{A}_{\ast}^{-1}$, where
$\mathbf{A}_{\ast} = \mathbf{A}(28,2,2,2,2)$.
Require $\kappa \geq (28,2,2,2,2)$. Then $\mathbf{h}_0 \leq \mathbf{A}_{\ast}$ and
} \displaybreak[0]
\mathbf{S}(n) & 
 \leq  \big(\tfrac{1}{\mathbf{h}_0}\big)^3 \textstyle\sum_{m=n+1}^{\infty} 
\exp \big( 
- 4 \tfrac{1}{\mathbf{h}_0} \mathbf{A}_{\ast} \gold_+^m\big)\\
& \leq 
\exp \big( 
- \tfrac{1}{\mathbf{h}_0} \mathbf{A}_{\ast} \gold_+^{n+1}\big)
\Big(\big(\tfrac{1}{\mathbf{h}_0}\big)^3 \exp \big( 
- \tfrac{1}{\mathbf{h}_0} \mathbf{A}_{\ast} \big)\Big)
\textstyle\sum_{m=1}^{\infty} 
\exp \big( 
- 2 \gold_+^m\big)
\end{align*}
We have $\sum_{m=1}^{\infty} \exp\big(-2\gold_+^m\big)
\leq \tfrac{1}{2} \sum_{m=1}^{\infty} \gold_+^{-m} = \tfrac{1}{2} (\gold_+-1)^{-1}
= \tfrac{1}{2} \gold_+\leq 1$.
Require $\kappa \geq (56,4,4,2,4)$. Then $\mathbf{h}_0 \leq \mathbf{A}(\kappa)\leq  \mathbf{A}(56,4,4,2,4)
= \mathbf{A}_{\ast}^2$, and
$$
\big(\tfrac{1}{\mathbf{h}_0}\big)^3 \exp \big( 
- \tfrac{1}{\mathbf{h}_0} \mathbf{A}_{\ast} \big)
\leq 
\big(\tfrac{1}{\mathbf{h}_0}\big)^3 \exp \big( 
- (\tfrac{1}{\mathbf{h}_0})^{1/2} \big)
\leq 8!\, \mathbf{h}_0 \leq 2^{16}\mathbf{h}_0 \leq 1
$$
Since $\mathbf{A}_{\ast} \geq \mathbf{A}(\kappa)$, we have
$\mathbf{S}(n) \leq \exp(-\tfrac{1}{\mathbf{h}_0} \mathbf{A}(\kappa) \gold_+^{n+1})$. Fix
$\kappa = (56,4,4,2,4)$. 
All the inequalities for $\kappa$ hold, and claim (A) is proved. Let $\mathbf{A}^{\sharp} = \mathbf{A}(56,0,4,4,4)$, as in the statement of 
Theorem \ref{thm_main3}. Since
$\mathbf{A}^{\sharp} \leq \mathbf{A}(\kappa)$, the condition
$\mathbf{h}_0 < \mathbf{A}^{\sharp}$
in the statement of Theorem \ref{thm_main3} implies the condition $\mathbf{h}_0 < \mathbf{A}(\kappa)$ used in this proof.\\
So far, we have verified the estimate
\eqref{eq:kdhdkhfdkhdkd}, and we have verified the assumptions
Theorem \ref{thm:main2} (\ref{item:dkhdkhdm1}), (\ref{item:dkhdkhd2}) 
and \eqref{eq:kdhkdhkdff}. 
In the assumption Theorem \ref{thm:main2} (\ref{item:dkhdkhd}), the cases $j \geq 1$ follow from 
Theorem \ref{thm:main2} (\ref{item:dkhdkhdm1}), (\ref{item:dkhdkhd2}). Since
$\mathbf{H}_0 = \mathbf{h}_0$ and
$K_0 \leq \mathbf{D}$, the remaining $j=0$ case
in Theorem \ref{thm:main2} (\ref{item:dkhdkhd})
follows from
\begin{multline*}
2^{71} (\tfrac{1}{\mathbf{H}_0})^2 (K_0)^5\exp ( - \tfrac{1}{\mathbf{H}_0} 2^{-21}(K_0)^{-2}) \leq 
2^{71} (\tfrac{1}{\mathbf{h}_0})^2 \mathbf{D}^5\exp ( - \tfrac{1}{\mathbf{h}_0} 2^{-21}\mathbf{D}^{-2})\\
\leq
2^{71} (\tfrac{1}{\mathbf{h}_0})^2
\mathbf{D}^5\, 8!\, \big(\mathbf{h}_0 2^{21} \mathbf{D}^2\big)^{8}
\leq 2^{255}\mathbf{D}^{21} \mathbf{h}_0^6 \leq (\mathbf{h}_0/\mathbf{A}^{\sharp})^6
< 1
\end{multline*}
\emph{Lebesgue measure of the set of admissible $\mathbf{f}_0$.} The set of all $\mathbf{f}_0 = (\mathbf{h}_0,w_0,q_0) \in (0,\infty)^3$ that satisfy (i), (ii), (iii) is a product $(0,\mathbf{A}^{\sharp}) \times F_w\times F_q$
(depending on $\mathbf{D}$ and $\gamma$), where $F_w \subset (0,1)\setminus \mathbb{Q}$ and $F_q \subset (0,\infty)\setminus \mathbb{Q}$. Both
$(0,\mathbf{A}^{\sharp})$ and $F_q$ have positive measure, because $\mathbf{A}^{\sharp} > 0$ and $(\tfrac{1}{2},\tfrac{2}{3})\setminus \mathbb{Q} \subset F_q$.  In fact, if $q_0 \in (\tfrac{1}{2},\tfrac{2}{3})$,
then $1/(1+q_0) = 1/(1+1/(1+1/(1+x)))$ with $x = (2q_0-1)/(1-q_0)\in (0,1)$, that is $k_0=k_{-1}=k_{-2} = 1 \leq \mathbf{D}$.  Suppose $\gamma > 1$
and $\mathbf{D} > (\log 2)^{-1} \gamma/(\gamma-1)$.
Let $G(x) = \tfrac{1}{x} - \lfloor \tfrac{1}{x}\rfloor$ be the Gauss map from $(0,1)\setminus \mathbb{Q}$ to itself.
We have $k_{n+1} = \lfloor 1/G^n(w_0) \rfloor$ for all $n \geq 0$. For all $n \geq 0$, set
$$X_n = \big\{w_0\in (0,1)\setminus \mathbb{Q}\;\big|\; G^n(w_0) < \mathbf{D}^{-1}(n+1)^{-\gamma}\big\}
= G^{-n}\big(\,\big(0,\;\mathbf{D}^{-1}(n+1)^{-\gamma}\big) \setminus \mathbb{Q}\big)$$
where $G^{-n}$ is the $n$-th inverse image of sets. Let $\mu_G$ be the probability measure on $(0,1)\setminus \mathbb{Q}$ with density $(\log 2)^{-1}(1+x)^{-1}$ (with respect to the Lebesgue measure). It is well-known that $\mu_G(X) = \mu_G(G^{-1}(X))$ for all measurable $X\subset (0,1)\setminus \mathbb{Q}$. Therefore,
$$\mu_G(X_n) = \mu_G\big(\,\big(0,\;\mathbf{D}^{-1}(n+1)^{-\gamma}\big) \setminus \mathbb{Q}\big)
=  \tfrac{1}{\log 2} \log\big(1 + \tfrac{1}{\mathbf{D}(n+1)^{\gamma}}\big)
\leq \tfrac{1}{\log 2} \tfrac{1}{\mathbf{D} (n+1)^{\gamma}}$$
Let $X_n^c$ be the complement of $X_n$ in $(0,1)\setminus \mathbb{Q}$. Then $\bigcap_{n \geq 0}X_n^c \subset F_w$, since $w_0 \in X_n^c$ implies $k_{n+1} = \lfloor 1/G^n(w_0) \rfloor \leq 1/G^n(w_0) \leq \mathbf{D} (n+1)^{\gamma}$. We have
\begin{align*}
& \mu_G(F_w)  \geq \mu_G(\textstyle\bigcap_{n \geq 0}X_n^c)
= 1-\mu_G(\textstyle \bigcup_{n \geq 0} X_n)
\geq 1 - \sum_{n \geq 0} \mu_G(X_n)\\
&  \geq 
1 - \tfrac{1}{\mathbf{D} \log 2} \textstyle\sum_{n \geq 0}  \tfrac{1}{(n+1)^{\gamma}}
\geq 1 - \tfrac{1}{\mathbf{D} \log 2} \big(1 + \textstyle\int_1^{\infty} x^{-\gamma}\dd x\big)
= 1 - \tfrac{1}{\mathbf{D} \log 2} \tfrac{\gamma}{\gamma-1} > 0
\end{align*}
Consequently, also the Lebesgue measure of $F_w$ is positive.
\qed
\end{proof}

%% file: SectionLightCone.tex
\section{Causal structure and particle horizons} \label{sec:fkhkhh}
In this section we show that
the spatially homogeneous vacuum spacetimes corresponding to those solutions of \eqref{eq:kshkjhkf} that are obtained by combining  Theorems \ref{thm:main2} and \ref{thm_main3} and Propositions \ref{prop:kdhfkshkds} and \ref{prop:skhdkjhfd}, have ``particle horizons'' (see \cite{Mis} for this notion), contradicting a conjecture in \cite{Mis}.
\begin{theorem}
Let $\mathbf{D}$, $\gamma$, $\mathbf{f}_0 = (\mathbf{h}_0,w_0,q_0)$ be as in Theorem \ref{thm_main3}. Let
$\mathbf{f}_j = (\mathbf{h}_j,w_j,q_j)$, $\pi_j$, $\sigma_{\ast}$ and $\mathbf{g}_j$ be as in Theorem
\ref{thm:main2}. Adopt the remaining notation of Theorems \ref{thm:main2} and \ref{thm_main3}.
Denote the components of $\mathbf{g}_j \in \mathcal{F}$ by $(\mathbf{h}_j',w_j',q_j')$.
Recall that $\mathbf{h}_j' \in (0,1)$.
\\
Fix a constant $\lambda_0 > 0$. Set $\tau_0 = 0$. Introduce sequences $(\lambda_j)_{j \geq 0}$ and $(\tau_j)_{j \geq 0}$ by
\begin{align*}
\lambda_{j} & = \lambda_{j-1}\;\big\{\Lambda[\pi_{j},\sigma_{\ast}](\mathbf{g}_{j})\big\}^{-1}\;\in\; (0,\lambda_0]
 && \text{for all $j \geq 1$}\\
\tau_{j} & = \tau_{j-1} + (\mathbf{h}_{j}' \lambda_{j})^{-1} \big\{\tau_{1+}(\mathbf{g}_{j}) - \tau_{2-}[\pi_{j},\sigma_{\ast}](\mathbf{g}_{j})\big\}
 && \text{for all $j \geq 1$}
\end{align*}
Then:
\begin{enumerate}
\item[(a)] $\tau_j > \tau_{j-1}$ for all $j \geq 1$ and $\lim_{j \to \infty}\tau_j = +\infty$.
\item[(b)] The solution to \eqref{eq:kshkjhkf} with initial data $\Phi(0) = \lambda_0\, \Phi_{\star}(\pi_0,\mathbf{g}_0,\sigma_{\ast})$ exists for all $\tau \geq 0$, that is $\Phi=\alpha\oplus \beta: [0,\infty) \to \mathcal{D}(\sigma_{\ast})$, and $\Phi(\tau_j) = \lambda_j\, \Phi_{\star}(\pi_j,\mathbf{g}_j,\sigma_{\ast})$ for all $j \geq 0$.
\item[(c)] For all $j \geq 1$ we have the bound 
$$\textstyle\mathbf{M}_j \stackrel{\mathrm{def}}{=} \sup_{\tau \in (\tau_{j-1},\tau_j)} \max_{(\xa,\xb,\xc)\in \mathcal{C}}\; \alpha_{\xb,\xc}[\Phi](\tau)
\;\leq\; -2^{-2} \lambda_j \min\{(w'_j)^2,(w'_j)^{-1} \}$$
\end{enumerate}
Set $\zeta_{\xa}(\tau) = -\tfrac{1}{2}\int_0^{\tau} \dd s\, \alpha_{\xa}(s)$ for $\xa=1,2,3$ (see 
Proposition \ref{prop:khdkhkhkhss}) and for all $s \geq 0$ set 
$$
\mathbf{L}(s) \;\textstyle \stackrel{\mathrm{def}}{=}\;
\int_s^{\infty} \dd \tau\, \max_{(\xa,\xb,\xc)\in \mathcal{C}} \exp \big( - \zeta_{\xb}-\zeta_{\xc}\big)$$
(see the right hand side of \eqref{eq:dfkhkdhkfddhfkdhdk}
in 
Proposition \ref{prop:dskhdkhfd}).
Then $\mathbf{L}(s) \leq \mathbf{L}(0) < \infty$ and $\lim_{s \to \infty}\mathbf{L}(s) = 0$.
\end{theorem}
\begin{proof}
Proposition \ref{prop:skhdkjhfd} implies
$\tfrac{1}{2} \leq \mathbf{h}_j' \lambda_j(\tau_j-\tau_{j-1}) \leq 2^2$ and
$\tau_j \geq \tau_{j-1} + (2\lambda_0)^{-1}$,
which implies (a).
Theorems \ref{thm:main2}, \ref{thm_main3}, and Propositions
\ref{prop:kdhfkshkds}, \ref{prop:skhdkjhfd} imply (b).
Proposition \ref{prop:skhdkjhfd} (e.5) implies (c). Estimate
\begin{align*}
\mathbf{L}(0) & \leq
\textstyle \sum_{\ell=1}^{\infty} \int_{\tau_{\ell-1}}^{\tau_{\ell}}\dd \tau \exp\Big(
\tfrac{1}{2} \int_0^{\tau} \dd \tau'\, \max_{(\xa,\xb,\xc)\in \mathcal{C}} \alpha_{\xb,\xc}[\Phi](\tau')
 \Big)\\
 & \leq 
 \textstyle \sum_{\ell=1}^{\infty} \int_{\tau_{\ell-1}}^{\tau_{\ell}}\dd \tau \exp\Big(
\tfrac{1}{2} \sum_{m=1}^{\ell-1} \int_{\tau_{m-1}}^{\tau_m} \dd \tau'\,\max_{(\xa,\xb,\xc)\in \mathcal{C}} \alpha_{\xb,\xc}[\Phi](\tau')
\Big)\\
 & \leq 
 \textstyle \sum_{\ell=1}^{\infty} 
 (\tau_{\ell} - \tau_{\ell-1})
 \exp\Big(
\tfrac{1}{2} \sum_{m=1}^{\ell-1}
(\tau_m-\tau_{m-1})\, \mathbf{M}_m
\Big)\\
 & \leq 
2^2 \textstyle \sum_{\ell=1}^{\infty} 
(\mathbf{h}'_{\ell}\lambda_{\ell})^{-1} 
 \exp\Big(
-2^{-4} 
 \sum_{m=1}^{\ell-1}
(\mathbf{h}'_m)^{-1} \, \min\{(w_m')^2,(w_m')^{-1}\}
\Big)
\end{align*}
By Theorem \ref{thm:main2}, we have 
$(\slaa{\delta}_j,\mathbf{f}_j) \in B_2\mathcal{F}$
and $\mathbf{g}_j \in B[\slaa{\delta}_j,\mathbf{f}_j]$ for all $j \geq 0$.
Hence, $\tfrac{1}{2}\mathbf{h}_j \leq \mathbf{h}_j' \leq 2 \mathbf{h}_j$
and $\tfrac{1}{2}w_j \leq w_j' \leq 2 w_j$. 
Proposition \ref{prop:skhdkjhfd} (b) implies
$\Lambda[\pi_j,\sigma_{\ast}](\mathbf{g}_j)
\leq 1 + \lambda_L(\mathbf{g}_j) \leq 
3 + w_j'
\leq 3(1+w_j)
$ and
$\lambda_{\ell}^{-1} \leq \lambda_0^{-1} 
\prod_{k=1}^{\ell} 3(1+w_k)
$. Therefore,
\begin{align}
\notag \mathbf{L}(0) 
 & \leq 
2^3 (\lambda_0)^{-1} \textstyle \sum_{\ell=1}^{\infty} 
(\mathbf{h}_{\ell})^{-1} 
\Big(\prod_{k=1}^{\ell} 3(1+w_k)\Big)\\
\notag & \hskip 30mm \times
\textstyle \exp\Big(
-2^{-7} 
 \sum_{k=1}^{\ell-1}
(\mathbf{h}_k)^{-1} \, \min\{(w_k)^2, (w_k)^{-1}\}
\Big)\\
\notag & \leq 
2^5 (\lambda_0)^{-1} \textstyle
\sum_{m=0}^{\infty}\sum_{\ell=J(m)+1}^{J(m+1)}
(\mathbf{H}_{\ell})^{-1} 
(2^7 \max_{1 \leq k \leq \ell} K_k)^{\ell}\\
\label{eq:dshkdhk} & \hskip 30mm \times
\textstyle \exp\Big(
-2^{-17} 
(1-\delta_{\ell 1})
(\mathbf{H}_{\ell-1})^{-1}  (K_{\ell-1})^{-2}
\Big)
\end{align}
where $\mathbf{H}_j$ and $K_j$ are as in Theorem \ref{thm:main2}, and $\delta_{\ell 1}$ is a Kronecker delta. See
\eqref{eq:dhkfhfkfkhkhkd}. In the exponential, we have bounded the sum over $k=1,\ldots,\ell-1$ from below by its $k=\ell-1$ summand if $\ell \geq 2$
and by zero otherwise.
The sum over $\ell = J(m)+1,\ldots, J(m+1)$ has $k_{m+1} \leq \mathbf{D}(m+1)^{\gamma}$ many terms.
By the proof of Theorem \ref{thm_main3}, for every $m \geq 0$, the following estimates, uniformly in $\ell =J(m)+1,\ldots, J(m+1)$, hold:
\begin{itemize}
\item\rule{0pt}{10pt}$(\mathbf{H}_{\ell})^{-1} \leq 2 (\mathbf{h}_0)^{-1}  (2 \mathbf{D} (m+1)^{\gamma})^{2(m+1)}$
\item\rule{0pt}{10pt}$(2^7 \max_{1 \leq k \leq \ell} K_k)^{\ell}
\leq  (2^{7+ \gamma} \mathbf{D}(m+1)^{\gamma})^{\mathbf{D} (m+1)^{\gamma+1}}$
\item\rule{0pt}{10pt}$(\mathbf{H}_{\ell-1})^{-1} \geq 2^{-4} (\mathbf{h}_0)^{-1} \rho_+^{2m}$
where $\rho_+ = \tfrac{1}{2}(1+\sqrt{5})$
\item\rule{0pt}{10pt}$(K_{\ell-1})^{-2} \geq \mathbf{D}^{-2} 2^{-2\gamma} (m+1)^{-2\gamma}$
\end{itemize}
By these estimates, in particular the fact that $(\mathbf{H}_{\ell-1})^{-1}$ grows at least exponentially in $m$,
the right hand side of \eqref{eq:dshkdhk} is finite, and $\mathbf{L}(0) < \infty$. \qed
\end{proof}

%% file: SectionAppendix1le2.tex
\section{Bounds for a particular product of continued fractions} \label{app:cfprod}
This appendix is entirely self-contained, the notation is completely local.
 Its single purpose is to prove Proposition \ref{prop:kfdhkjhdfkdf} below, which is used
in the proof of Lemma \ref{lem:kdskjshkshks}.
\begin{definition}
For all integers $m$ and $n$ and all sequences $(x_i)_{i \in \mathcal{I}}$ where $\mathcal{I}\subset \Z$,
define $x_{m:n}$ to be the ordered sequence $x_m,x_{m+1},\ldots,x_{n-1},x_n$ if $m \leq n$ and the empty sequence if $m > n$. In the first case, it is required that $[m,n]\cap \Z \subset \mathcal{I}$. Similarly, define
$x_{m::n}$ to be the ordered sequence $x_m,x_{m-1},\ldots,x_{n+1},x_n$ if $m \geq n$ and the empty sequence if $m < n$. In the first case, it is required that $[n,m]\cap \Z \subset \mathcal{I}$.
\end{definition}
\begin{definition}[Continued fractions]
For every integer $n\geq 0$ and every finite sequence of strictly positive integers $(k_i)_{1 \leq i \leq n}$ set recursively
$$\langle k_{1:n} \rangle = \begin{cases} 0 & n = 0\\
\big(k_1 + \langle k_{2:n} \rangle\big)^{-1} & n \geq 1
\end{cases}
\qquad \in \qquad [0,1] \cap \mathbb{Q}
$$
For every infinite sequence $(k_i)_{i\geq 1}$ of strictly positive integers, set
$$\langle
k_1,k_2,\ldots \rangle = \lim_{n \to \infty} \langle k_{1:n} \rangle\;\;\in\;\; (0,1) \setminus \mathbb{Q}
$$
\end{definition}
\begin{example}
$\langle\,\rangle =  \langle k_{1:0}\rangle = 0$ and $\langle k_1 \rangle =
\langle k_{1:1} \rangle = 
1/k_1$ and $\langle k_1,k_2\rangle = \langle k_{1:2}\rangle = 1/(k_1+1/k_2)$.
\end{example}
\begin{definition}[Fibonacci numbers] $F_1 = F_2 = 1$ and $F_n = F_{n-1} + F_{n-2}$, $n \geq 3$.
\end{definition}
\begin{proposition}\label{prop:kfdhkjhdfkdf}
For every two-sided sequence of strictly positive integers $(k_i)_{i \in \Z}$, define two-sided sequences $(v_i)_{i \in \Z}$ and $(w_i)_{i \in\Z}$ by $v_i = \langle k_i,k_{i-1},k_{i-2},\ldots \rangle$ and $w_i = \langle k_i,k_{i+1},k_{i+2},\ldots\rangle$.
Then, for all integers $M < N$:
\begin{enumerate}[(a)]
\item $\tfrac{1}{2}\leq \prod_{i=M+1}^N (v_i/w_i) \leq 2$
\item $\prod_{i=M+1}^N w_i \leq (F_{N-M+1})^{-1}
\leq (\tfrac{1}{2}(\sqrt{5}-1))^{N-M-1}$
\end{enumerate}
\end{proposition}
The proof of Proposition \ref{prop:kfdhkjhdfkdf} is given at the end of this appendix.
\begin{definition} \label{def:khkdhkfdh}
Let $P_0(\,)=1$ and $P_1(x_1) = x_1$ and for all $n\geq 2$, set
\begin{equation}\label{eq:kdshkdjhf}
P_n(x_{1:n} ) = x_1 P_{n-1}(x_{2:n}) + P_{n-2}(x_{3:n})
\end{equation}
\end{definition}
\begin{example}
$P_2(x_{1:2}) = 1+x_1x_2$ and $P_3(x_{1:3}) = x_1 + x_3 + x_1x_2x_3$
and $P_4(x_{1:4}) = 1 + x_1x_2 + x_3x_4 + x_1x_4 + x_1x_2x_3x_4$.
\end{example}
\begin{lemma} \label{lem:kjhkhkhaaaaaaasaa}
Recall Definition \ref{def:khkdhkfdh}. For all integers $n\geq 0$, we have:
\begin{enumerate}[(a)]
\item $P_n$ is a polynomial of degree $n$, jointly in its $n$ arguments, with coefficients in $\{0,1\}$
\item\label{item:lsep} $P_n$ is a polynomial of degree 1, separately in each of its $n$ arguments
\item $P_n(1,\ldots,1) = F_{n+1}$
\item\label{item:sym} $P_n(x_{1:n}) = P_n(x_{n::1})$
for all $x_1,\ldots,x_n \in \R$
\item\label{item:cf} $\langle k_{1:n} \rangle = P_{n-1}(k_{2:n}) / P_n(k_{1:n})$ for all strictly positive integers $(k_i)_{1\leq i \leq n}$, $n\geq 1$
\end{enumerate}
\end{lemma}
\begin{proof}
(a) through (\ref{item:cf}) are all shown by induction, using \eqref{eq:kdshkdjhf}. To show (\ref{item:sym}), observe that $\text{(\ref{item:sym})}_0$, $\text{(\ref{item:sym})}_1$, 
$\text{(\ref{item:sym})}_2$ and $\text{(\ref{item:sym})}_3$ hold.
For the induction step, let $n\geq 4$ and suppose $\text{(\ref{item:sym})}_0$ through $\text{(\ref{item:sym})}_{n-1}$ hold. Then, using
only \eqref{eq:kdshkdjhf} and the induction hypothesis,
\begin{align*}
& P_n(x_{1:n}) - P_n(x_{n::1})\\
& = x_1 P_{n-1}(x_{2:n}) + P_{n-2}(x_{3:n})
- x_nP_{n-1}(x_{n-1::1}) - P_{n-2}(x_{n-2::1})\\
& =
x_1 P_{n-1}(x_{n::2})
+ P_{n-2}(x_{n::3})
- x_nP_{n-1}(x_{1:n-1})
- P_{n-2}(x_{1:n-2})\\
& = 
x_1\big(x_n P_{n-2}(x_{n-1::2}) + P_{n-3}(x_{n-2::2})\big)
+ \big(x_n P_{n-3}(x_{n-1::3}) + P_{n-4}(x_{n-2::3})\big)\\
& \hskip4mm
-x_n \big( x_1 P_{n-2}(x_{2:n-1}) + P_{n-3}(x_{3:n-1})\big)
-\big(
x_1 P_{n-3}(x_{2:n-2}) + P_{n-4}(x_{3:n-2})
\big)
\end{align*}
Verify that all the terms cancel, by the induction hypothesis. This implies $\text{(\ref{item:sym})}_n$. To show (\ref{item:cf}), observe that
$\text{(\ref{item:cf})}_1$ holds. Let $n\geq 2$ and suppose
$\text{(\ref{item:cf})}_{n-1}$ holds. Then,
\begin{equation*}
\langle k_{1:n} \rangle = \big(k_1 + \langle k_{2:n}\rangle\big)^{-1}
= \bigg( k_1 + \frac{P_{n-2}(k_{3:n})}{P_{n-1}(k_{2:n})}\bigg)^{-1} 
= \frac{P_{n-1}(k_{2:n})}{k_1P_{n-1}(k_{2:n}) + P_{n-2}(k_{3:n})}
\end{equation*}
Now,  \eqref{eq:kdshkdjhf} implies $\text{(\ref{item:cf})}_{n}$. \qed
\end{proof}
\begin{lemma}\label{lem:fdhkhfdkdh}
For all integers
$m-1 \leq M < N \leq n$ and all $x_m,\ldots,x_n \in [1,\infty)$,
\begin{equation}\label{eq:dskdfhdkfhfkd}
2\,P_{M-m+1}(x_{m:M})P_{n-M}(x_{M+1:n}) - P_{N-m+1}(x_{m:N})P_{n-N}(x_{N+1:n}) \geq 0
\end{equation}
Moreover, if $m = M+1$, then the factor $2$ on the left hand side can be dropped, that is,
\begin{equation}\label{eq:dskdfhdkfhfkd2}
P_{n-M}(x_{M+1:n}) - P_{N-M}(x_{M+1:N})P_{n-N}(x_{N+1:n}) \geq 0
\end{equation}
\end{lemma}
\begin{proof}
In this proof, we use
the recursion relation \eqref{eq:kdshkdjhf} and the
reflected recursion relation that is obtained by applying
 Lemma \ref{lem:kjhkhkhaaaaaaasaa} (\ref{item:sym}) to all three terms of  \eqref{eq:kdshkdjhf}.
Fix $M$ and $N$.
Inequality \eqref{eq:dskdfhdkfhfkd} is proved by induction over $m$ and $n$, where $m \leq M+1$ and $n \geq N$. Denote the left hand side of \eqref{eq:dskdfhdkfhfkd} by $Q_{m,n}$. Then,
\begin{align*}
Q_{M+1,N} & = P_{N-M}(x_{M+1:N}) \geq 0\\
Q_{M+1,N+1} & =  2\,P_{N+1-M}(x_{M+1:N+1}) - P_{N-M}(x_{M+1:N})\, x_{N+1} \\
& = P_{N-M}(x_{M+1:N})\,x_{N+1} + 2\, P_{N-1-M}(x_{M+1:N-1}) \geq 0\\
Q_{M,N} & = 2\,x_M P_{N-M}(x_{M+1:N}) - P_{N-M+1}(x_{M:N}) \\
& = x_M P_{N-M}(x_{M+1:N}) - P_{N-M-1}(x_{M+2:N})\\
& \geq x_M x_{M+1} P_{N-M-1}(x_{M+2:N})
- P_{N-M-1}(x_{M+2:N}) \geq 0\\
Q_{M,N+1} & = 
2\,x_M P_{N+1-M}(x_{M+1:N+1}) - P_{N-M+1}(x_{M:N}) x_{N+1}\\
& =
2\,x_Mx_{N+1} P_{N-M}(x_{M+1:N}) 
+ 2\,x_M P_{N-1-M}(x_{M+1:N-1}) \\
& \qquad
- x_M P_{N-M}(x_{M+1:N}) x_{N+1} - P_{N-M-1}(x_{M+2:N}) x_{N+1}
\\
& \geq 
x_Mx_{N+1} P_{N-M}(x_{M+1:N})  - P_{N-M-1}(x_{M+2:N}) x_{N+1}\\
& \geq x_Mx_{N+1}x_{M+1}P_{N-M-1}(x_{M+2:N})- P_{N-M-1}(x_{M+2:N}) x_{N+1} \geq 0
\end{align*}
These four cases and the two recursion relations
\begin{itemize}
\item $Q_{m,n} = x_m Q_{m+1,n} + Q_{m+2,n}$ when $m \leq M-1$
and $n \geq N$
\item $Q_{m,n} = x_n Q_{m,n-1} + Q_{m,n-2}$ when $m \leq M+1$ and $n \geq N + 2$
\end{itemize}
imply \eqref{eq:dskdfhdkfhfkd}. Inequality \eqref{eq:dskdfhdkfhfkd2} is shown in 
an entirely similar way.
\qed
\end{proof}
\begin{proof}[of Proposition \ref{prop:kfdhkjhdfkdf}]
Recall
(\ref{item:sym}), (\ref{item:cf}) 
in Lemma \ref{lem:kjhkhkhaaaaaaasaa}.
Let $m-1 \leq M < N \leq n$.
\begin{align*}
& \prod_{i=M+1}^{N} \frac{
\langle k_{i::m} \rangle}{\langle k_{i:n}\rangle} = \prod_{i=M+1}^{N} \frac{P_{i-m}(k_{i-1::m})}{
P_{i-m+1}(k_{i::m})}\cdot
\frac{P_{n-i+1}(k_{i:n})}{P_{n-i}(k_{i+1:n})}\\
& = \frac{P_{M-m+1}(k_{M::m})}{P_{N-m+1}(k_{N::m})}\cdot
\frac{P_{n-M}(k_{M+1:n})}{P_{n-N}(k_{N+1:n})}
 = \frac{P_{M-m+1}(k_{m:M})}{P_{N-m+1}(k_{m:N})}\cdot
\frac{P_{n-M}(k_{M+1:n})}{P_{n-N}(k_{N+1:n})}
\end{align*}
The right hand side is $\geq \tfrac{1}{2}$, by inequality \eqref{eq:dskdfhdkfhfkd}.
Now, let $m \to -\infty$ and $n \to + \infty$ to obtain
$\prod_{i=M+1}^{N} (v_i/w_i) \geq \tfrac{1}{2}$.
By symmetry, we also have 
$\prod_{i=M+1}^{N} (w_i/v_i) \geq \tfrac{1}{2}$. This implies
(a) in Proposition \ref{prop:kfdhkjhdfkdf}. Similarly, using \eqref{eq:dskdfhdkfhfkd2}, 
\begin{align*}
& \prod_{i=M+1}^{N} \langle k_{i:n}\rangle = 
\frac{P_{n-N}(k_{N+1:n})}{P_{n-M}(k_{M+1:n})} \leq \frac{1}{P_{N-M}(k_{M+1:N})}
\leq \frac{1}{P_{N-M}(1,\ldots,1)}
\end{align*}
Let $n \to + \infty$ to obtain $\prod_{i=M+1}^N w_i \leq 1/F_{N-M+1}$.
\qed
\end{proof}

%% file: SectionAppendixQL.tex
\section{The modulus of continuity of the map $\mathcal{Q}_L$ introduced in Definition \ref{def:kdhkhskhkdssPT2}}
\label{app:dslkhjlkd}
\begin{lemma}\label{lem:dfskdshkhk}
Let $\mathcal{Q}_L: (0,\infty)^3\to (0,\infty)^2\times \R$ be the map introduced in Definition \ref{def:kdhkhskhkdssPT2}.
For all $\mathbf{f}_i = (\mathbf{h}_i,w_i,q_i)\in (0,\infty)^3$ with $0 < \mathbf{h}_i \leq 1$, $i=1,2$, with $\mathbf{f}_1\neq \mathbf{f}_2$, such that $q_1$ and $q_2$ are either both $<1$ or both $>1$,
$$
\frac{
\|\mathcal{Q}_L(\mathbf{f}_2) - \mathcal{Q}_L(\mathbf{f}_1)\|_{\R^3}
}{
\|\mathbf{f}_2-\mathbf{f}_1\|_{\R^3}
}
\leq
\begin{cases}
2^{12} q_{\mathrm{min}}^{-2} \log (2 + w_{\mathrm{max}}) & \text{if $q_1,q_2 < 1$}\\
2^{11} q_{\mathrm{max}} &  \text{if $q_1,q_2 > 1$}
\end{cases}
$$
Here, $w_{\mathrm{max}} = \max\{w_1,w_2\}$ and $q_{\mathrm{max}} = \max\{q_1,q_2\}$ and
$q_{\mathrm{min}} = \min\{q_1,q_2\}$.
\end{lemma}
\begin{proof}
We prove the following claim, which implies the Lemma: \emph{Each of the nine partial derivatives of $\mathcal{Q}_L: \mathbf{f} = (\mathbf{h},w,q) \mapsto \mathcal{Q}_L(\mathbf{f})$ is bounded in absolute value by
\begin{equation} \label{eq:pdbd}
\begin{cases}
2^{10} q^{-2} \log (2+w) & \text{if $\mathbf{f}\in (0,1]\times (0,\infty) \times (0,1)$}\\
2^9 q & \text{if $\mathbf{f}\in (0,1]\times (0,\infty) \times (1,\infty)$}
\end{cases}
\qquad > 1
\end{equation}}Let $0 < \mathbf{h}\leq 1$ and $q\neq 1$.  Let $(\mathbf{h}_L,w_L,q_L) = \mathcal{Q}_L(\mathbf{f})$ and 
let $\mathrm{num1}_L$, $\mathrm{num2}_L$, $\mathrm{den}_L$ be as in Definition \ref{def:kdhkhskhkdssPT2}. 
We first estimate the partial derivatives of $q_L = \mathrm{num1}_L/\mathrm{den}_L$ and $\mathbf{h}_L = \mathrm{num2}_L/\mathrm{den}_L$. Each of $\mathrm{num1}_L$, $\mathrm{num2}_L$, $\mathrm{den}_L$ is of the form
$$L_1(w,q) + L_2(w,q)q + L_3(w,q)\mathbf{h} + L_4(w,q)\mathbf{h}\log \lambda_L(\mathbf{f})$$
with $\lambda_L(\mathbf{f}) = 1+1/w_L(\mathbf{f})$ as in Definition \ref{def:kdhkhskhkdssPT2} and with $L_i(w,q)=a_i(q)w + b_i(q)$
where $a_i(q)$ and $b_i(q)$ are constant separately for $q < 1$ and for $q > 1$ and satisfy $-3\leq a_i(q),b_i(q) \leq 3$, where $i=1,2,3,4$. Let $k=1,2$ and $\mathrm{num}k_L = L_1 + L_2q + L_3\mathbf{h} + L_4 \mathbf{h}\log\lambda_L$
and $\mathrm{den}_L = L_1' + L_2'q + L_3'\mathbf{h} + L_4'\mathbf{h}\log \lambda_L$ with $L_i = a_iw+b_i$
and $L_i'=a_i'w+b_i'$ (\emph{Warning}: the prime does \emph{not} denote a derivative). 
Then
\begin{multline*} 
\big(\tfrac{\p}{\p x} \mathrm{num}k_L\big) \, \mathrm{den}_L - \big(\tfrac{\p}{\p x} \mathrm{den}_L\big)\, \mathrm{num}k_L\\
= \begin{cases}
(L_3 + L_4\,\log \lambda_L)(L_1'+L_2'q) - (L_3'+L_4'\log \lambda_L)(L_1+L_2q) & \text{if $x=\mathbf{h}$}\\
\rule{0pt}{27pt}\begin{aligned}
+(a_1+a_2q+a_3\mathbf{h}+a_4\mathbf{h}\log \lambda_L)
(b_1'+b_2'q+b_3'\mathbf{h}+b_4'\mathbf{h}\log \lambda_L)\\
-
(a_1'+a_2'q+a_3'\mathbf{h}+a_4'\mathbf{h}\log \lambda_L)
(b_1+b_2q+b_3\mathbf{h}+b_4\mathbf{h}\log \lambda_L)\\
+ \mathbf{h} \Big\{L_4 \big(L_1'+L_2'q+L_3'\mathbf{h})
- L_4'\big(L_1+L_2q+L_3\mathbf{h})\Big\} \tfrac{\p}{\p w}\log \lambda_L
\end{aligned} & \text{if $x=w$}\\
\rule{0pt}{14pt}L_2(L_1'+L_3'\,\mathbf{h} + L_4'\mathbf{h}\,\log \lambda_L)
- L_2'(L_1 + L_3\,\mathbf{h}+L_4\mathbf{h}\,\log \lambda_L)
& \text{if $x=q$}
\end{cases}
\end{multline*}
Recall that $|\mathbf{h}|\leq 1$ and $|a_i|,|a_i'|,|b_i|,|b_i'|\leq 3$
and $|L_i|,|L_i'|\leq 3(1+w)$  and $\log \lambda_L \geq 0$. 
\begin{itemize}
\item If $q < 1$, then $|\tfrac{\p}{\p w}\log \lambda_L| \leq (1+w)^{-1}$ and
\begin{multline*}
\big|\big(\tfrac{\p}{\p x} \mathrm{num}k_L\big) \, \mathrm{den}_L - \big(\tfrac{\p}{\p x} \mathrm{den}_L\big)\, \mathrm{num}k_L\big|\\
\leq \left\{
\begin{aligned}
& 36(1+w)^2(1+\log \lambda_L) && \text{if $x=\mathbf{h}$}\\
& 18(3+\log \lambda_L)^2 + 54(1+w) && \text{if $x=w$}\\
& 18(1+w)^2(2+\log \lambda_L) && \text{if $x=q$}
\end{aligned}
\right\} \leq 2^{10} (1+w)^2 \log (2+w)
\end{multline*}
For the second inequality, use $\tfrac{1}{2}\leq \log \lambda_L \leq 1+w$
and $(\log \lambda_L)^2 \leq 1+w$.
\item If $q > 1$, then $|\log \lambda_L| \leq 1$ and $|\tfrac{\p}{\p w}\log \lambda_L| \leq (1+w)^{-2}$ and
$a_2' = b_2' = 0$ and
\begin{multline*}
\big|\big(\tfrac{\p}{\p x} \mathrm{num}k_L\big) \, \mathrm{den}_L - \big(\tfrac{\p}{\p x} \mathrm{den}_L\big)\, \mathrm{num}k_L\big|\\
\leq \left\{
\begin{aligned}
& 72(1+w)^2\,q && \text{if $x=\mathbf{h}$}\\
& 270\,q &&  \text{if $x=w$}\\
& 27(1+w)^2  && \text{if $x=q$}
\end{aligned}
\right\} \leq 2^9\,(1+w)^2\,q
\end{multline*}
\end{itemize}
To finish the proof, observe that $\mathrm{den}_L \geq (1+w)\min\{1,q\} > 0$.
Each partial derivative of $q_L = \mathrm{num1}_L/\mathrm{den}_L$ and $\mathbf{h}_L = \mathrm{num2}_L/\mathrm{den}_L$ is bounded in absolute value by \eqref{eq:pdbd}. And so are the partial derivatives of $w_L$, because $\p w_L/\p \mathbf{h} = \p w_L/\p q = 0$, and because $\p w_L/\p w=-(1+w)^{-2}$ if $q < 1$ and $\p w_L/\p w=1$ if $q > 1$. \qed
\end{proof}